\documentclass[journal, draftcls, onecolumn]{IEEEtran}

\usepackage{amsmath, mathrsfs, graphics, epsfig, amssymb, color, subfigure}
\usepackage[center]{caption}

\newtheorem{lemma}{\bf Lemma}
\newtheorem{defn}{\bf Definition}
\newtheorem{thm}{\bf Theorem}

\newtheorem{rem}{\bf Remark}
\newtheorem{cor}{\bf Corollary}
\newtheorem{ex}{\bf Example}
\newtheorem{cl}{\bf Claim}

\newtheorem{conj}{\bf Conjecture}

\begin{document}

\title{The Diversity Multiplexing Tradeoff of the MIMO Half-Duplex Relay Channel} 

\author{{\Large Sanjay Karmakar \qquad Mahesh K. Varanasi}
\thanks{S. Karmakar and M. K. Varanasi are both with the Department
of Electrical Computer and Energy Engineering, University of Colorado at Boulder, Boulder,
CO, 30809 USA e-mail: (sanjay.karmakar@colorado.edu, varanasi@colorado.edu).
}\thanks{The material in this paper was presented in part at the IEEE International Symposium on Information Theory, 2010, Austin, Texas~\cite{SKMV_ISIT2010_Sym_relay}.}
}
\markboth{Submitted, IEEE Trans. Inform. Th., Jun. 2011}%
{Shell \MakeLowercase{\textit{et al.}}: Bare Demo of IEEEtran.cls for Journals}

\maketitle

\begin{abstract}
The fundamental diversity-multiplexing tradeoff of the three-node, multi-input, multi-output (MIMO), quasi-static, Rayleigh faded, half-duplex relay channel is characterized for an arbitrary number of antennas at each node and in which opportunistic scheduling (or dynamic operation) of the relay is allowed, i.e., the relay can switch between receive and transmit modes at a channel dependent time. In this most general case, the diversity-multiplexing tradeoff is characterized as a solution to a simple, two-variable optimization problem. This problem is then solved in closed form for special classes of channels defined by certain restrictions on the numbers of antennas at the three nodes. The key mathematical tool developed here that enables the explicit characterization of the diversity-multiplexing tradeoff is the joint eigenvalue distribution of three mutually correlated random Wishart matrices. Besides being relevant here, this distribution result is interesting in its own right. Previously, without actually characterizing the diversity-multiplexing tradeoff, the optimality in this tradeoff metric of the dynamic compress-and-forward (DCF) protocol based on the classical compress-and-forward scheme of Cover and El Gamal was shown by Yuksel and Erkip. However, this scheme requires global channel state information (CSI) at the relay. In this work, the so-called quantize-map and forward (QMF) coding scheme due to Avestimehr {\em et} {\em al} is adopted as the achievability scheme with the added benefit that it achieves optimal tradeoff with only the knowledge of the (channel dependent) switching time at the relay node. Moreover, in special classes of the MIMO half-duplex relay channel, the optimal tradeoff is shown to be attainable even without this knowledge. Such a result was previously known only for the half-duplex relay channel with a single antenna at each node, also via the QMF scheme. More generally, the explicit characterization of the tradeoff curve in this work enables the in-depth comparisons herein of full-duplex versus half-duplex relaying as well as static versus dynamic relaying, both as a function of the numbers of antennas at the three nodes.



\end{abstract}

\newpage
\section{Introduction}
\label{sec_introduction}

Cooperative communication techniques can advantageously utilize the fading environment of a wireless network to provide better reliability and/or rate~\cite{SEB1,SEB2}. The simplest theoretical abstraction of a cooperative communication network is the $3$-node relay channel (RC), where the relay node helps the communication between the source and destination nodes by forwarding an appropriately processed version of the source message received at the relay node to the destination. Moreover, multiple antennas at the three nodes can markedly boost rate and reliability performance by allowing for the exploitation of the inherent combined MIMO and cooperative communication gains. 

\begin{figure}[htp]
  \begin{center}
    \subfigure[$\textrm{CN}_1$: A mobile set acts as a relay]{\label{cooperative-networks-a}\includegraphics[scale=.25]{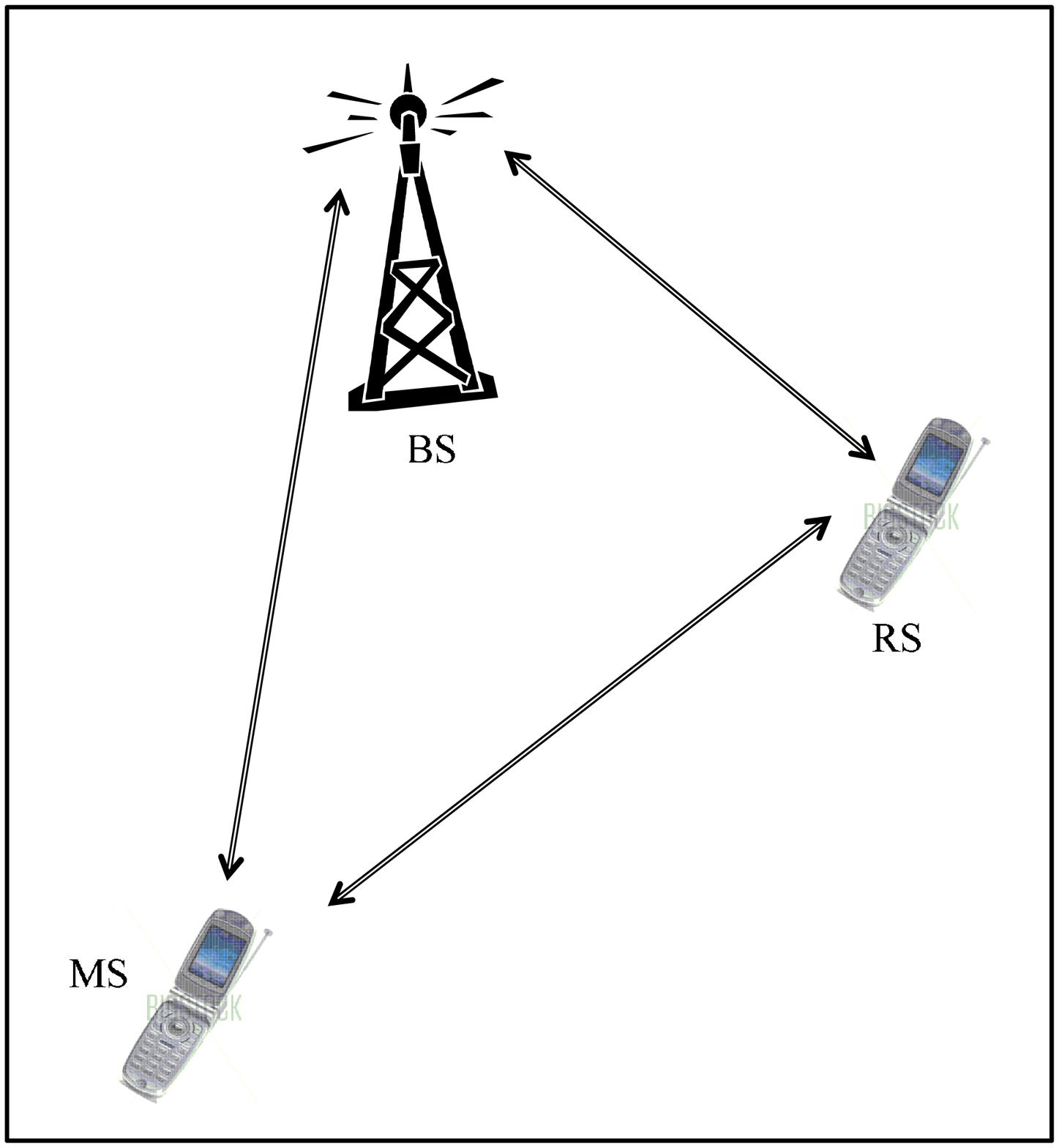}}
    \subfigure[$\textrm{CN}_2$: A smaller base station acts as a relay]{\label{cooperative-networks-b}\includegraphics[scale=.25]{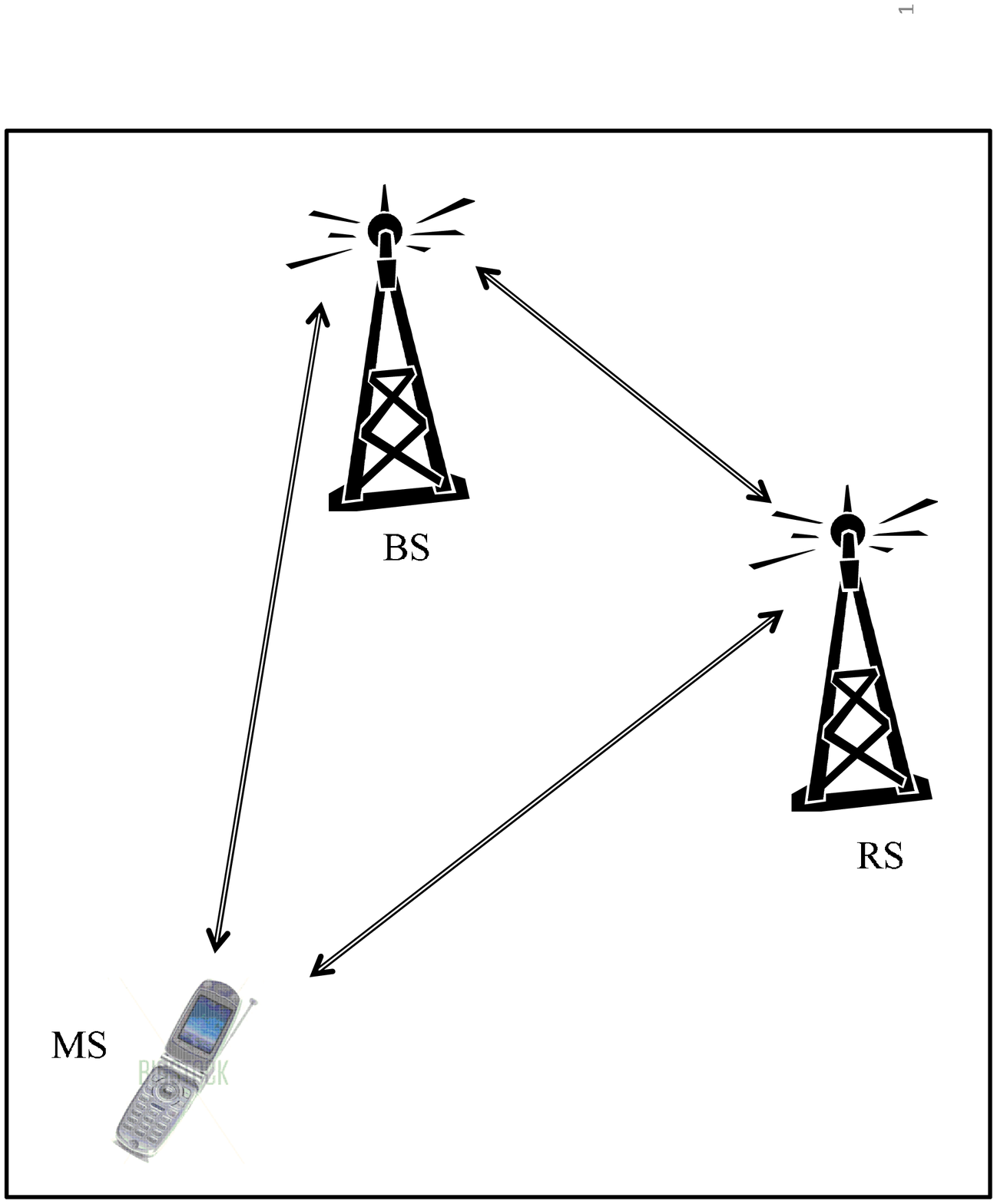}}
    \subfigure[$\textrm{CN}_3$: A sensory network with a mobile relay station (MRS)]{\label{cooperative-networks-c}\includegraphics[scale=0.25]{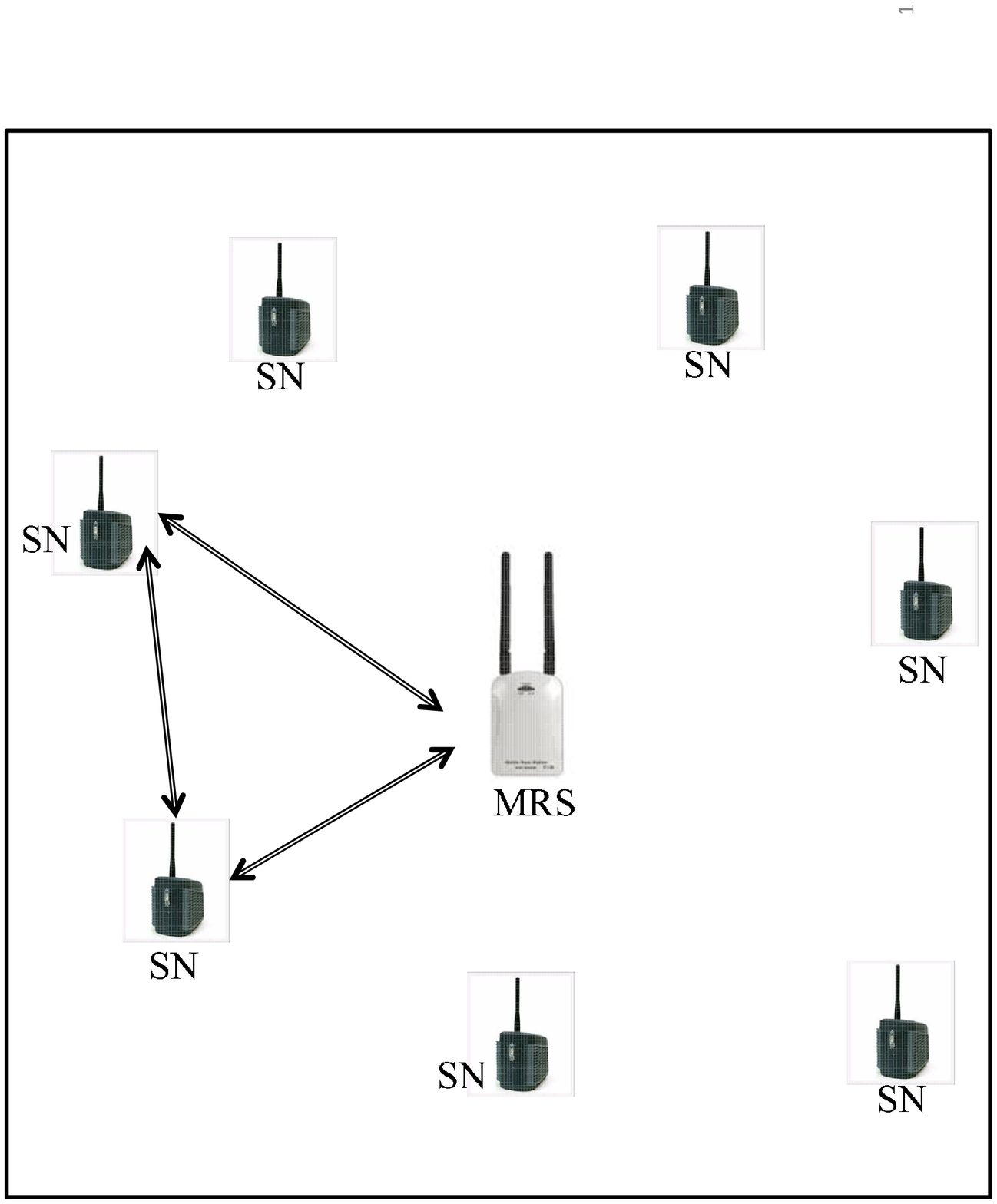}}
      \end{center}
  \caption{Three Examples of Cooperative Networks.}
  \label{cooperative-networks}
\end{figure}

MIMO relay channel communications can be considered for various applications. For instance, Fig. \ref{cooperative-networks} depicts three different cooperative communication scenarios to which the theory of this work applies. Fig. \ref{cooperative-networks-a} depicts a cellular network, denoted $\textrm{CN}_1$, wherein a mobile user (or mobile set (MS)) uses another mobile user as the relay station (RS) to communicate its message to and from the base station (BS). This cooperative model was proposed in \cite{SEB1}. Fig. \ref{cooperative-networks-b} depicts a scenario where, in a cellular network (denoted $\textrm{CN}_2)$, a particular cell area is divided into more than one sub-cell and each sub-cell is served by an additional dedicated node (a smaller BS) to provide better quality of service. Thus each user in these sub-cells can use this dedicated node to relay their messages to and from the BS. The $\textrm{CN}_2$ network is different from $\textrm{CN}_1$ in the sense that in it the relay station can host a larger number of antennas. It is under consideration to be implemented in LTE-advanced and WiMAX technologies~\cite{YHXM} and being standardized for broadband wireless access by the IEEE 802.16's relay task group~\cite{relayTG} for expanded throughput and coverage with deployment of relay stations of complexity and cost lower than that of legacy base stations but higher than that of mobile stations \cite{relayTG}.
A third example of a cooperative network (denoted $\textrm{CN}_3$) is the sensor network of Fig. \ref{cooperative-networks-c} (cf. \cite{WSC}), where a more capable mobile relay station (i.e., with more antennas) helps several less capable sensor nodes to communicate with each other. It is possible to give other examples, see for instance, the application of cooperative communication in ad-hoc networks in \cite{Wang_Zhang_Madsen}. Note that the numbers of antennas at the different nodes vary across the applications and also depend on whether the communication is uplink or downlink (such as in the $\textrm{CN}_1$ and $\textrm{CN}_2$ networks), which points to the importance of studying MIMO relay channels with an arbitrary number of antennas at each node.

The relay has two phases of operation: the listen phase, in which it receives the signal from the source, and the transmit phase, in which it transmits some version of the received signal to the transmitter. If the relay can simultaneously operate in both phases it is called a full-duplex (FD) relay and the corresponding channel is called a full duplex relay channel (FD-RC). Otherwise, if the relay can only operate in one phase at a time it is called a half-duplex (HD) relay and the corresponding channel a half duplex relay channel (HD-RC). Due to the large difference between the power levels of the transmitted and received signals however, it is difficult, if not impossible, to design FD relays cost- and space-efficiently. The focus of this paper is hence on MIMO HD-RCs.

Cooperative protocols proposed and analyzed for the HD-RC can be divided into different classes. If a protocol uses the CSI at the relay to opportunistically decide the switching time -- the time at which it switches between the listen and transmit phases -- it is called a {\em dynamic} protocol. Dynamic protocols considered in the literature include the dynamic decode-and-forward (DDF) protocol of \cite{KHP} and the dynamic compress-and-forward (DCF) protocol of \cite{YEE}. Otherwise, if the relay is restricted to switch between the listen and transmit phases at a pre-determined, channel independent time, it is called a {\em static} protocol. An important example is the static compress-and-forward (SCF) protocol of \cite{YEE,OCM}. An HD-RC on which protocols are restricted to be static is called a static HD-RC, and one on which they are not, is called a dynamic HD-RC (or simply HD-RC) since any static protocol can be thought as a special case of the family of dynamic protocols. On the other hand, the transmit-receive phases on an HD-RC can be thought of as states and additional information can be conveyed to the receiver through the sequence of these states. A cooperative protocol that uses these states to send additional information is called a {\em random} protocol, otherwise it is a {\em fixed} protocol (see \cite{GK,YEE}).

In this paper we focus on the general three-node MIMO HD-RC, i.e., in which there are an arbitrary number of antennas at each node and in which there are no constraints on the relay operation so that it can operate via the static or dynamic and random or fixed mode. In order to avoid repeated use of a complete descriptor of a channel we will use simplified ones when the meaning is unambiguous from the context. For example, we may refer to the dynamic MIMO HD-RC sometimes simply as the relay channel because this channel is the central focus of this paper. Similarly, we may refer to the  MIMO FD-RC or the static MIMO HD-RC as the FD-RC or the static RC, respectively, when the meaning is clear.

In spite of its apparent simplicity, neither the capacity nor the diversity-multiplexing tradeoff \cite{tse1} of the $3$-node MIMO HD-RC is known till date. However, in a recent paper \cite{Avestimehr_Diggavi_Tse} the capacity of this channel was characterized within a constant number of bits. It was proved that the so called quantize-map and forward (QMF) scheme can achieve a rate which is within a constant number of bits to the cut-set upper bound of the channel. On a slow fading HD-RC however, the instantaneous end-to-end mutual information, and therefore the cut-set bound of the channel, is a random quantity. A meaningful measure of performance on this channel is hence the {\em outage probability} which is a measure of reliability as a function of the communication rate in that it represents the (minimum) probability with which a particular rate cannot be supported on the channel. The result of \cite{Avestimehr_Diggavi_Tse} provides upper and lower bounds on this outage probability, both in terms of the instantaneous cut-set upper bound of the channel, denoted as $\bar{C}(\mathcal{H})$ for a channel realization of $\mathcal{H}$, as
\begin{equation*}
    \Pr\{\bar{C}(\mathcal{H})<R\}\leq \mathcal{P}_{out}\leq \Pr\{\bar{C}(\mathcal{H})<R-\kappa\},
\end{equation*}
where $R$ is the operating rate and $\kappa$ is a positive constant independent of the channel parameters and the signal-to-noise ratio (SNR) of the channel (e.g., see Theorem~8.5 in \cite{Avestimehr_Diggavi_Tse}). The exact evaluation of the outage probability requires both these bounds to be tight which in turn requires the exact capacity of the channel. Instead, in this paper we focus on the asymptotic (in SNR) behavior of the tradeoff between rate and reliability as captured by the DMT metric, first introduced in the context of the point-to-point MIMO channel by Zheng and Tse in \cite{tse1}. Since it was proved that a random protocol can increase the capacity by at most one bit in \cite{YEE}, there is no distinction between random and fixed protocols in the DMT framework. Thus, from the DMT perspective, characterizing the DMT of the HD-RC by allowing for dynamic operation of the relay but restricting it to the fixed mode still amounts to characterizing the fundamental DMT of the HD-RC. It is noted that the DMT of the static MIMO HD-RC was recently obtained by Leveque et al in \cite{OCM}.
It is shown here that in general a restriction that relay protocols be static fundamentally limits DMT performance over the MIMO HD-RC.


Since its first introduction and pioneering work by Van der Meulen~\cite{Meulen} and the subsequent significant progress made by Cover and El Gamal in \cite{Cover_Gamal}, the relay channel and its more general versions have been analyzed from both the capacity perspective in \cite{ElGamal-Aref, Wang_Zhang_Madsen,Madsen_Zhang,Kramer_Gastpar_Gupta,Avestimehr_Diggavi_Tse,Noisy_NC} and from the diversity, or more generally, the DMT perspectives for the 3-node relay network in quasi-static fading channels in \cite{LTW0,LW,NabarRU:relay,PrasadN:CO-OP:ISIT04,KHP,NpV,YEE,SKMV_DDF_journal,OCM,SKMV_ISIT2010_Sym_relay}. The earliest works demonstrating the improved reliability of the relay channel in terms of the diversity gain compared to the corresponding point-to-point (PTP) channel were reported in~\cite{LTW0,LW,NabarRU:relay,PrasadN:CO-OP:ISIT04}, where a number of simple cooperative protocols were proposed and their DMT performance was analyzed. Later in \cite{PrasadN:CO-OP:ISIT04,KHP,NpV}, more efficient protocols were introduced.  Notable among these were the dynamic decode-and-forward (DDF) protocol which is DMT optimal on a single antenna relay channel for a range of low multiplexing gains and the so called enhanced dynamic decode-and-forward (EDDF) of \cite{NpV}. All of the above protocols were analyzed for the relay channel with single antenna nodes (\cite{PrasadN:CO-OP:ISIT04} considers multiple antennas at the destination).

Multiple antenna relay channels were studied by Yuksel and Erkip in \cite{YEE}, where the DMTs of a number of protocols were evaluated and the DMT optimality of the compress-and-forward (CF) coding scheme of \cite{Cover_Gamal} for the MIMO FD- and HD-RCs was proved. In the DCF protocol of \cite{YEE}, the relay node utilizes all the instantaneous channel realizations, i.e., {\em global CSI}, to compute the quantized signal and the optimal switching time of the relay node. However, global CSI at the relay is not necessary to achieve DMT optimal performance as we discuss next.

The static QMF protocol of \cite{Avestimehr_Diggavi_Tse} can achieve the cut-set bound of the HD-RC to within a constant gap that is independent of CSI and SNR for a {\em fixed} scheduling of the relay, i.e., for an a priori fixed time $t_d$ at which the relay switches from listen to the transmit mode (e.g., see Theorem~8.3 in \cite{Avestimehr_Diggavi_Tse}), and do so without knowledge of CSI at the relay. However, on a slow fading dynamic HD-RC the cut-set bound, denoted by $\bar{C}(\mathcal{H},t_d)$, is a function of global CSI including $t_d$ and hence the optimal switching time $t_d^*$ that maximizes the cut-set bound can be a function of the instantaneous channel matrices. If (just) this switching time information is hence available at the relay node it then follows that the QMF protocol can achieve a rate that is within a constant gap to $\bar{C}(\mathcal{H},t_d^*)$ without requiring global CSI at the relay. Henceforth, the QMF protocol that operates with a {\em dynamic } (channel-dependent) switching time of the relay node will be referred to as the {\em dynamic QMF} protocol. Since a constant gap is irrelevant in the DMT metric, the dynamic QMF protocol in which the relay switches from listen to transmit modes at $t_d^*$ achieves the fundamental DMT of the MIMO HD-RC with only knowledge of $t_d^*$ at the relay, as opposed to global CSI $\mathcal{H}$ required by the DCF protocol of \cite{YEE}. The above discussion shows that the DMT of the static MIMO HD-RC found in \cite{OCM} for the static CF protocol requiring global CSI at the relay also applies to the static QMF protocol and hence to the static MIMO HD-RC without any CSI at the relay. 

In this paper, we are interested in establishing the DMT of the dynamic MIMO HD-RC. While the optimality in the DMT metric of the DCF was shown in \cite{YEE} and that of the dynamic QMF protocol \cite{Avestimehr_Diggavi_Tse} is evident from the discussion above, the characterization of this optimal performance, i.e., the fundamental DMT of the MIMO HD-RC, is not yet known and is the subject of this paper. The key mathematical tool that prevented its computation thus far is the joint eigenvalue distribution of three mutually correlated random Wishart matrices. Here we obtain this distribution as a stepping stone to characterizing the DMT of the MIMO HD-RC. Not only is this distribution result interesting in its own right as a problem in random matrices, it is also arises in establishing the DMT of the MIMO interference channel as was done by the authors in \cite{Sanjay_Varanasi_IC_DMT, Sanjay_Varanasi_ZIC_DMT}.


The explicit DMT of the MIMO HD-RC evidently would serve as a theoretical benchmark against which the performances of the various cooperative protocols proposed and analyzed in the literature can be compared. Further, cooperative protocols which are suboptimal but cost-efficient provide the system designer with an option to trade performance and complexity if their performance loss can be quantified relative to optimal performance. Moreover, the answers of a number of interesting and open questions hinge on the explicit characterization of the DMT of the MIMO HD-RC. For instance, while the DMT of the MIMO FD-RC is an upper bound to that of the MIMO HD-RC, it is not known whether the latter is strictly worse than that of the former. The question is especially intriguing in light of the result by Pawar {\em et} {\em at} in \cite{PAT} where it was shown that the DMT of the single-antenna (or single-input, single-output (SISO)) HD-RC is identical to that of the FD-RC. Comparing with the DMT of the MIMO FD-RC which was found in \cite{YEE}, this question can be resolved if the explicit DMT of the MIMO HD-RC can be characterized. There are also open questions regarding the comparative performances of the static and dynamic MIMO HD-RCs. Although intuitively it seems that the dynamic HD-RC should have a better DMT than the static HD-RC, there is no theoretical proof of this thesis to date. For instance, in the SISO case, there is no difference in the DMTs of the static and dynamic HD-RCs as shown in \cite{PAT} because the DMT of the static QMF protocol coincides with that of the SISO FD-RC. The question here is whether this result continues to hold in the more general static and dynamic MIMO HD-RCs. This question can be answered if the DMT of the (dynamic) MIMO HD-RC were to be found, since in this case, one could simply compare with the DMT of the static MIMO HD-RC of \cite{OCM}.

This paper answers the two questions raised above in the negative. In particular, the results of this paper, examples from which are shown in Fig. \ref{fig_dmt_1k1_channel-a} and Fig. \ref{fig_dmt_1k1_channel-b} depicting the DMTs of the HD-RC with single-antenna source and destinations but with two and four antennas at the relay (the $(1,2,1)$ and $(1,4,1)$ RCs (applicable for example, to $\textrm{CN}_3$ of Fig. \ref{cooperative-networks}(c)), respectively, show that in general neither is the DMT of the static MIMO HD-RC always equal to that of the corresponding dynamic MIMO HD-RC, nor is the DMT of the MIMO FD-RC always identical to that of the corresponding MIMO HD-RC.
\begin{figure}[htp]
  \begin{center}
    \subfigure[$(1,2,1)$ relay channel]{\label{fig_dmt_1k1_channel-a}\includegraphics[scale=0.5]{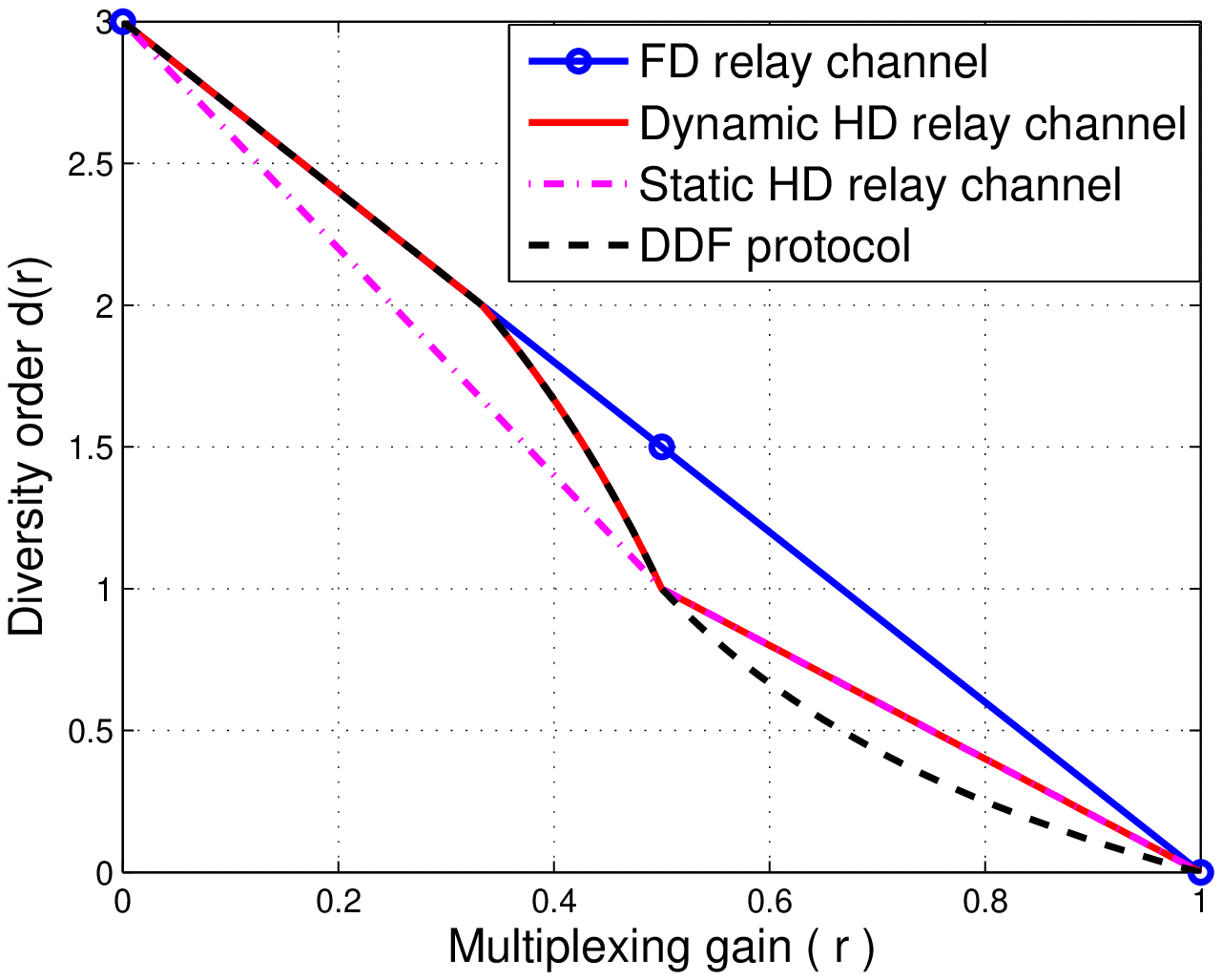}}
    \subfigure[$(1,4,1)$ relay channel]{\label{fig_dmt_1k1_channel-b}\includegraphics[scale=0.5]{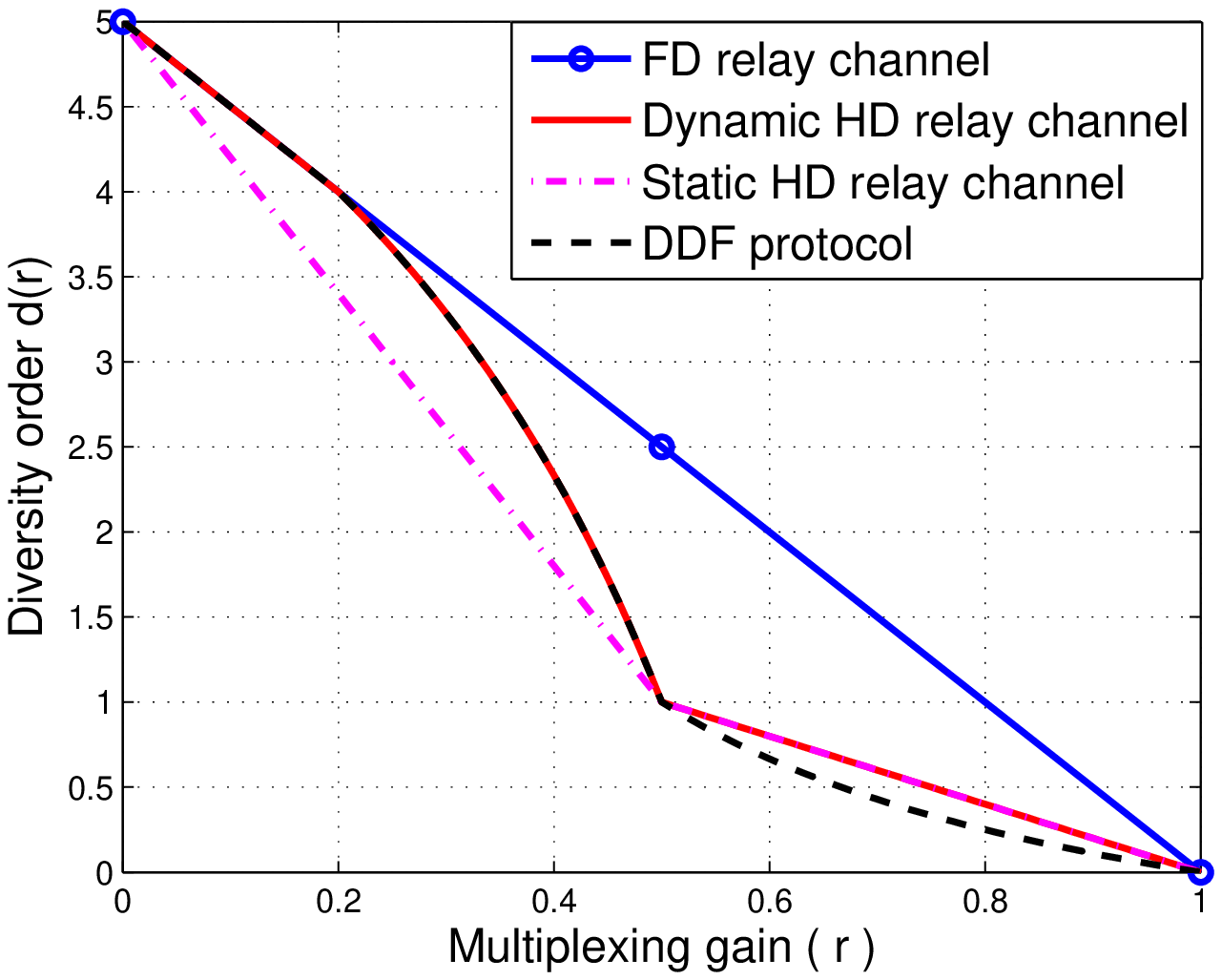}}
  \end{center}
  \caption{Comparison of DDF and SCF protocol with the fundamental DMT of the (1,k,1) relay channel.}
  \label{fig_dmt_1k1_channel}
\end{figure}

Besides resolving the above discussed problems, the explicit DMT computed in this paper provides sharper answers about the MIMO HD-RC. Denoting an RC with $m$, $k$, and $n$ antennas at the source, relay and destination as the $(m,k,n)$ RC, they include, but are not limited to the following:
\begin{itemize}
\item While in general the DMT performance of the MIMO HD-RC is inferior to that of the corresponding MIMO FD-RC, it is found that for two classes of channels, namely (a) the $(m,k,n)$ RCs with $m>n\geq k$ and (b) the $(n,1,n)$ RCs, the DMTs of the HD- and FD-RCs are identical (see Remark~\ref{rem:n1n-dmt-FD-HD-equal}). Therefore, for these classes of RCs, an FD relay does not improve the DMT performance over that of the corresponding HD-RC.

\item It is well known from \cite{YEE} that the fundamental DMT of $(m,k,n)$ FD-RC is given by
    \begin{equation}
    \label{dmt-fd-rc}
        \min \left\{d^{ptp}_{(m+k),n}(r),d^{ptp}_{m,(n+k)}(r) \right\},~0\leq r\leq \min\{m,n\},
    \end{equation}
    where $d^{ptp}_{p,q}(r)$ represents the DMT of the $p\times q$ MIMO point-to-point channel \cite{tse1}. From this it is clear that an additional antenna at the relay node strictly improves the DMT performance of an FD-RC at all multiplexing-gains. However, this is not true for the MIMO HD-RC. When $k$ is large enough, an extra antenna at the relay node does not further improve the DMT of the HD-RC for high multiplexing gains (see Remark~\ref{rem:extra-antenna-at-the-relay}).

\item In general, for a set of high multiplexing gain values, the optimal DMT of the MIMO HD-RC channel can be achieved without CSI at the relay node. As the number of antennas at the relay node increases, the size of this set increases (e.g., see Fig. \ref{fig_dmt_comparison2-b}).

\item Finally, it is proved that the DMT of the $(1,k,1)$ and the $(n,1,n)$ HD-RC can be achieved by the QMF protocol without CSI at the relay node, i.e., neither global CSI nor even the switching time information is necessary at the relay node. To the best of our knowledge, this is the first result regarding the achievability of the DMT of a non-SISO HD-RC without CSI at the relay node.
\end{itemize}

The rest of the paper is organized as follows. In Section~\ref{sec:system_model}, we describe the system model and provide some preliminaries including the asymptotic joint distribution of the eigenvalues of three specially correlated random matrices which will be used later to derive the fundamental DMT of the MIMO HD-RC. In Section~\ref{sec:DMT-HD-RC}, we characterize the fundamental DMT of the MIMO HD-RC  as the solution of a simple optimization problem in three steps: 1) in Subsection~\ref{subsec:upper-bound} we obtain an upper bound on the instantaneous capacity; 2) in Subsection \ref{subsec:lower-bound}, we obtain a lower bound on the instantaneous capacity as the achievable rate of the dynamic QMF protocol, which is within a constant gap from the upper bound, and finally, in Subsection \ref{subsec:optimization-problem}, we characterize the DMT as a solution to an optimization problem, which we subsequently simplify to a 2-variable optimization problem. In Section~\ref{sec:closed-form-expressions}, we provide closed-form expressions for the DMT of different classes of MIMO HD-RCs including the class of symmetric $(n,k,n)$ RCs and then prove in Section~\ref{sec:only-CSIR} that the DMT of $(1,k,1)$ RC and $(n,1,n)$ HD-RC can be achieved without CSI at the relay node. Section \ref{sec:conclusion} concludes the paper.

\begin{proof}[{\bf Notations}]
$(x)^+$, $x\land y$, $|\mathcal{S}|$, $\det(X)$ and $(X)^{\dagger}$ represent the
number $\max\{x,0\}$, the minimum of $x$ and $y$, the size of the set $\mathcal{S}$, the determinant and the conjugate transpose of the matrix, $X$, respectively. The symbol $\textrm{diag}(.)$ represents a square diagonal matrix of corresponding size with the elements in its argument on the diagonal. $I_n$ represents an $n\times n$ identity matrix. We denote the field of real and complex numbers by $\mathbb{R}$ and $\mathbb{C}$, respectively. The set of real numbers between $r_1\in \mathbb{R}$ and $r_2 ( \, \geq r_1 ) \in \mathbb{R}$ will be denoted by $[r_1, r_2]$. The set of all $n\times m$ matrices with complex entries is denoted as $\mathbb{C}^{n\times m}$. The distribution of a complex Gaussian random vector with zero mean and covariance matrix $\Sigma$ is denoted as $\mathcal{CN}(0,\Sigma)$. The trace of a square matrix $A$, is denoted as $\textrm{Tr}(A)$. $A\succeq B$ (or $A\succ B$) would mean that $(A-B)$ is a positive-semidefinite (psd) matrix (or positive-definite (pd) matrix), respectively. $\Pr(\mathcal{E})$ represents the probability of the event $\mathcal{E}$. All the logarithms in this text are to the base $2$. Finally, any two functions $f(\rho)$ and $g(\rho)$ of $\rho$, where $\rho$ is the signal-to-noise ratio (SNR) defined later, are said to be exponentially equal and denoted as $f(\rho)\dot{=}g(\rho)$ if,
\begin{equation}
\label{eq:exp-equality}
\lim_{\rho \to \infty} ~\frac{\log(f(\rho))}{\log(\rho)} = \lim_{\rho \to \infty} ~\frac{\log(g(\rho))}{\log(\rho)},
\end{equation}
$\dot{\leq}$ and $\dot{\geq }$ signs are defined similarly.
\end{proof}

\begin{figure}[htb]
\centering
\includegraphics[width=8cm]{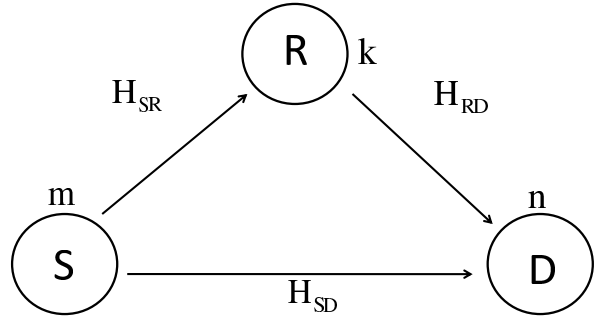}
\caption{System model of the MIMO 3-node relay channel.}
\label{fig_channel_model}
\end{figure}

\section{System model and preliminaries}
\label{sec:system_model}

We consider a quasi-static, Rayleigh faded MIMO HD-RC, with a single relay node as shown in Fig. \ref{fig_channel_model}. It is assumed that the source, destination and relay nodes have $m$, $n$ and $k$ antennas, respectively. Let $H_{SR}\in\mathbb{C}^{k\times m}$, $H_{SD}\in\mathbb{C}^{n\times m}$ and $H_{RD}\in\mathbb{C}^{n\times k}$ model the fading channel matrices between the source and relay, source and destination and relay and destination nodes, respectively. For economy of notation, the set of these channel matrices is denoted as $\mathcal{H}$. It is assumed that these matrices are mutually independent and their elements are independent and identically distributed (i.i.d.) as $\mathcal{CN}(0,1)$.

The channel matrices remain constant within a block of $L_b$ channel uses, where $L_b$ is the block length of the source codeword. Suppose that during the first $t_d L_b$ symbol times the relay node only listens to the source transmission and during the remaining $(1-t_d)L_b$ symbol times it transmits its own codeword $X_r~\in~\mathbb{C}^{k\times (1-t_d)L_b}$, where $t_d\in (0,1)$. In what follows, the listening phase and the transmitting phases of the relay node will be denoted by $p_1$ and $p_2$, respectively, and the fraction $t_d$ is called the {\em switching time}. Since the relay node operates dynamically, this switching time should be chosen to maximize the end-to-end instantaneous mutual information and can thus depend on all of $\mathcal{H}$.

We assume that the destination and relay nodes have global CSI $\mathcal{H}$, but the source node does not have any CSI. The relay node in this channel model is more capable than that on a relay channel with only receive CSI (CSIR) at all the nodes.\footnote{Note that the DMT of the relay channel with only CSIR at all the nodes can not be better than the DMT of the relay channel considered in this paper, since the relay node can choose to not use any CSI except $H_{SR}$. In fact, intuitively it seems that the latter may use the additional information ($H_{SD}$ and $H_{RD}$) at the relay for instance to optimize the time to switch from its listening mode to the transmit mode and achieve a better DMT than the former, on which the switching time can only be a function of $H_{SR}$. Interestingly, in this paper we shall prove that depending on the number of antennas at different nodes the DMT of the above two relay channels (with and without global CSI at the relay and destination node) can be identical, cf. Section~\ref{sec:only-CSIR}.} We assume short term power constraints at the source and relay, i.e., these nodes cannot allocate power across different fades of the channel as a function of $\mathcal{H}$, see equation \eqref{eq:power-constraint-at-relay}.

Further, we also assume that the source and relay nodes transmit information at fixed rates; in particular, the relay node does not use the available transmit CSI (CSIT)  to transmit information at a variable, channel dependent rate. Note that an information outage can be avoided on a communication link if CSIT is used to allocate power across different fading blocks (cf. \cite{HanlyTse,PrasadN:OUT-THMS:IT06,Kim_Skoglund}) under the so-called long-term power constraint and/or transmit information at a variable rate as a function of the instantaneous channel realizations. It was shown in \cite{Kim_Skoglund} that the DMT of a point-to-point MIMO channel can be improved by either of these two techniques.

Denoting the signals transmitted by the source at time $t$ in phases one and two as $X_{S_1}[t]$ and $X_{S_2}[t]$, respectively, and the signal transmitted by the relay at time $t$ as $X_R[t]$ (in phase two), the received signals at the destination and relay nodes in phase one are given as
\begin{eqnarray}
Y_{D_1}[t] &=& H_{SD}X_{S_1}[t]+Z_{D_1}[t], \\
Y_R[t]&=& H_{SR}X_{S_1}[t]+Z_R[t],
\end{eqnarray}
and the received signal at the destination node in phase two is given as
\begin{equation*}
Y_{D_2}[t]=H_{SD}X_{S_2}[t]+H_{RD}X_R[t]+Z_{D_2}[t],
\end{equation*}
where $Z_{D_1}[t], Z_{D_2}[t]\in\mathbb{C}^{n\times 1}$ and $Z_R[t]\in\mathbb{C}^{k\times 1}$ represent mutually independent additive noise random vectors at the destination and relay nodes, respectively. All the entries of these random vectors are assumed to be i.i.d. $\mathcal{CN}(0,1)$. The power constraints at the relay and the source nodes are\footnote{Allowing distinct powers at the source and relay nodes of $\rho$ and $c\rho$, respectively, where $c$ is a constant independent of $\rho$, does not alter the diversity-multiplexing tradeoff. We assume $c=1$ for ease of disposition.}
\begin{IEEEeqnarray}{l}
\label{eq:power-constraint-at-source}
\frac{1}{L_b}\left\{\sum_{t=1}^{\lceil t_dL_b \rceil}\textrm{Tr} \left( \mathbb{E}\left(X_{S_1}[t]X_{S_1}[t]^\dagger\right) \right) +\sum_{t=(\lceil t_d L_b\rceil+1)}^{L_b}\textrm{Tr} \left( \mathbb{E}\left(X_{S_2}[t]X_{S_2}[t]^\dagger\right) \right)
\right\} \leq  \rho; \\
\label{eq:power-constraint-at-relay}
\frac{1}{(L_b-\lceil t_dL_b \rceil)}  \sum_{t=(\lceil t_d L_b\rceil+1)}^{L_b}\textrm{Tr} \left( \mathbb{E}\left(X_R[t]X_R[t]^\dagger\right) \right) \leq  \rho.
\end{IEEEeqnarray}

Let $\{\mathcal{C}(\rho)\}$ be a sequence of codebooks, where for each $\rho$ the corresponding codebook $\mathcal{C}(\rho)$ consists of $2^{L_bR(\rho)}$ codewords, each of which is a $m\times L_b$ matrix satisfying equation \eqref{eq:power-constraint-at-source}. The sequence of codebooks is said to have a multiplexing gain of $r$ if
\begin{equation*}
    \lim_{\rho\to \infty}\frac{R(\rho)}{\log(\rho)}=r.
\end{equation*}

Further, suppose for such a coding scheme $\mathcal{C}(\rho)$, $P_e(\mathcal{C}(\rho),\rho)$ represents the average probability of decoding error at the destination node (averaged over the Gaussian noise, channel realizations and the different codewords of the codebook) at an SNR of $\rho$, then the optimal diversity order of the channel at a multiplexing gain $r$ is defined as

\begin{equation}\label{eq:def-diversity-order}
    d^*(r)\triangleq  \lim_{\rho\to \infty}\frac{-\log\left(P_e^*(\rho)\right)}{\log(\rho)},
\end{equation}
where $P_e^*(\rho)$ represents the minimum average probability of error achievable on a relay channel minimized over the collection of all possible coding schemes, $\mathscr{C}(\rho)$, i.e.,
\begin{equation}
\label{eq:pestar}
    P_e^*(\rho)\triangleq \min_{\{\mathcal{C}(\rho)\in \mathscr{C}\}}P_e(\mathcal{C}(\rho),\rho).
\end{equation}

In Subsection~\ref{subsec:optimization-problem}, we shall show that the optimal diversity order at a multiplexing gain of $r$ can be written as
\begin{IEEEeqnarray}{rl}
\label{eq:d-as-function-asymptotic-evs}
d^*(r)=\lim_{\rho \to \infty} -\frac{\log\left(\Pr\left\{r^*(\bar{\alpha},\bar{\beta},\bar{\delta})\leq r\right\}\right)}{\log(\rho)},
\end{IEEEeqnarray}
where $r^*(\bar{\alpha},\bar{\beta},\bar{\delta})$ is given by \eqref{eq:asymptotic-rate-expression} in Subsection \ref{subsec:optimization-problem} and $\bar{\alpha}$, $\bar{\beta}$ and $\bar{\delta}$ are vectors containing the negative SNR exponents of the eigenvalues (see equations \eqref{eq:transform-of-lambda-alpha}-\eqref{eq:transform-of-gamma-delta} in the following section) of $H_{SD}H_{SD}^{\dagger}$, $H_{SR}(I_m+\rho H_{SD}^{\dagger}H_{SD})^{-1}H_{SR}^{\dagger}$ and $ H_{RD}^{\dagger}(I_n+\rho H_{SD}H_{SD}^{\dagger})^{-1}H_{RD}$, respectively. Evidently, to evaluate the DMT, we need the joint probability density function (pdf) of the eigenvalues of these matrices, which we obtain next.

\subsection{Joint eigenvalue distribution of three mutually correlated Wishart matrices}
\label{subsec:eigenvalue_distribution}
Let us denote the matrices $H_{SD}H_{SD}^{\dagger}$, $H_{SR}(I_m+\rho H_{SD}^{\dagger}H_{SD})^{-1}H_{SR}^{\dagger}$ and $H_{RD}^{\dagger}(I_n+\rho H_{SD}H_{SD}^{\dagger})^{-1}H_{RD}$ as $W_1$, $W_2$ and $W_3$, respectively. It is evident from their structure that these matrices are not mutually independent. In general, finding the joint pdf of 2 or more mutually correlated random matrices is a difficult problem in the theory of random matrices. However, in this section we show that by exploiting the specific structure of these matrices and the distribution of the constituent matrices, we can compute a closed form expression for the joint pdf of their eigenvalues.

Let $0 < \lambda_u \leq \cdots \leq \lambda_1$, $0 < \mu_p \leq \cdots \leq \mu_1$ and $0 < \gamma_q \leq \cdots \leq \gamma_1$ be the ordered non-zero eigenvalues with probability 1 (w.p. 1) of $W_1$, $W_2$ and $W_3$, respectively. Define $\bar{\lambda}=[\lambda_1, \cdots \lambda_u]$, $\bar{\mu}=[\mu_1, \cdots \mu_p]$ and $\bar{\gamma}=[\gamma_1, \cdots \gamma_q]$ with $u=\min\{m,n\}$, $p=\min\{m,k\}$ and $q=\min\{n,k\}$. It is convenient to denote the joint pdf of the three sets of eigenvalues as $F_{W_1W_2W_3}(\bar{\lambda},\bar{\mu},\bar{\gamma}) $ and similarly their marginal and conditional pdfs, i.e., the marginal pdf of $\bar{\lambda}$ is denoted as $F_{W_1}(\bar{\lambda})$, the conditional pdf of $\bar{\mu}$ conditioned on $\bar{\lambda}$ is denoted as $ F_{W_2|W_1}(\bar{\mu}|\bar{\lambda}) $, etc. Consider the following lemma which provides the first step towards simplifying the problem at hand.
\begin{lemma}
\label{lem_conditional_independence}
The eigenvalues of $W_2$ are independent of the eigenvalues of $W_3$ given the eigenvalues of $W_1$, i.e.,
\begin{eqnarray}
F_{W_2W_3|W_1}(\bar{\mu},\bar{\gamma}|\bar{\lambda})=
F_{W_2|W_1}(\bar{\mu}|\bar{\lambda})F_{W_3|W_1}(\bar{\gamma}|\bar{\lambda}).
\end{eqnarray}
\end{lemma}
\begin{proof}[Proof]
The proof is provided in Appendix~\ref{pf_lem_conditional_independence}.
\end{proof}
Using the above lemma, the joint pdf of the eigenvalues of the three matrices can be expressed as
\begin{eqnarray}
\label{eq:pdf-first-product}
F_{W_1W_2W_3}(\bar{\lambda},\bar{\mu},\bar{\gamma}) = F_{W_2|W_1}(\bar{\mu}|\bar{\lambda})F_{W_3|W_1}(\bar{\gamma}|\bar{\lambda})F_{W_1}(\bar{\lambda}).
\label{eq:pdf-second-product}
\end{eqnarray}
The joint pdf of the eigenvalues of $W_1$, which is a central Wishart matrix, can be found for example in \cite{CZZ} whereas the conditional pdfs $F_{W_2|W_1}(\bar{\mu}|\bar{\lambda})$ and $F_{W_3|W_1}(\bar{\gamma}|\bar{\lambda})$ involve complicated functions such as determinants whose components are hypergeometric functions of the eigenvalues (e.g., see the proof of Theorem 1 in \cite{SKMV_DDF_journal}). However, since we are interested only in a high SNR analysis, it is sufficient to obtain $F_{W_2|W_1}(\bar{\mu}|\bar{\lambda})$ and $F_{W_3|W_1}(\bar{\gamma}|\bar{\lambda})$ exactly just up to their SNR exponents, i.e., approximate expressions which have the same SNR exponents as the exact joint pdf. For this purpose, we use the following theorem from \cite{SKMV_DDF_journal}. 

\begin{thm}[Theorem~$1$ in \cite{SKMV_DDF_journal}]
\label{thm_eigenvalue_distribution_1}
Let $H_1 \in \mathbb{C}^{N_2\times N_1}$ and $H_2 \in \mathbb{C}^{N_2\times N_3}$ be two mutually independent random matrices with independent, identically distributed (i.i.d.) $\mathcal{CN}(0,1)$ entries. Suppose that $\xi_1 \geq \xi_2\geq \cdots \xi_v > 0$ and $\lambda_1 \geq \lambda_2 \geq \cdots \lambda_u > 0$ are the ordered non-zero eigenvalues w.p. 1 of $V_1\triangleq H_1^{\dagger}(I_{N_2}+\rho H_2H_2^{\dagger})^{-1}H_1$ and $V_2\triangleq H_2H_2^{\dagger}$, respectively, with $u=\min\{N_2, N_3\}$ and $v=\min\{N_1, N_2\}$, and where all the eigenvalues are assumed to vary exponentially with SNR in the sense of equation \eqref{eq:transform-of-lambda-alpha}.\footnote{This assumption greatly simplifies an otherwise very complicated expression of the pdf. Further, in the context of the problem being analyzed in this paper, the usefulness of this assumption will be evident shortly.} Then, the conditional asymptotic pdf of the eigenvalues $\bar{\xi}$ given $\bar{\lambda}$ is given as
\begin{eqnarray}
\mathbf{f_1}(\bar{\xi}|\bar{\lambda}) \, \dot{=} \, \prod_{j=1}^{v}(\xi_j^{(N_1+N_2-2j)}e^{-\xi_j}) \prod_{\substack{(n_1=1,n_2=1)\\
((n_1+n_2)= (N_2+1))}}^{(u,v)} \left({e}^{-\rho \xi_{n_2} \lambda_{n_1}}\right) 
\prod_{i=1}^{u}(1+\rho \lambda_i)^{N_1}  \prod_{j=1}^{v} \prod_{i=1}^{(N_2-j)\land N_3} \left(\frac{1-{e}^{-\rho \xi_j \lambda_i}}{\rho \xi_j \lambda_i}\right). \nonumber
\end{eqnarray}
\end{thm}

Note that the above theorem gives the conditional pdf of the joint eigenvalues of $V_1$ given the eigenvalues of $V_2$ exactly up to its exponential order. This asymptotic distribution is simpler to obtain than its exact counterpart and is also sufficient for the DMT analysis. It can be easily verified \cite{Tulino_Verdu} that the first product term corresponds to the joint distribution of the eigenvalues of $H_1^\dagger H_1$. The additional three product terms appear because $V_1$ is a Wishart matrix with a non-identity covariance matrix. To see this, note that $V_1$ converges to $H_1^\dagger H_1$ if each of the eigenvalues of $\rho H_2H_2^\dagger$ tends to zero. Indeed, putting $\rho\lambda_i \to 0,~\forall i$ in the above expression, it is easily shown that the last three terms converge to $1$ giving the joint distribution of $H_1^\dagger H_1$.

Clearly, Theorem \ref{thm_eigenvalue_distribution_1} can be used to derive the asymptotic conditional joint pdf of the eigenvalues of $W_3$ given the eigenvalues of $W_1$. Consequently, we have
\begin{IEEEeqnarray}{ll}
F_{W_3|W_1}(\bar{\gamma}|\bar{\lambda})\dot{=}\prod_{l=1}^{q}({\gamma}_l^{(k+n-2l)}e^{-\gamma_l}) & \prod_{\substack{(i=1,l=1)\\ ((i+l)= (n+1))}}^{(u,q)} \left({e}^{-\rho \gamma_l \lambda_i}\right) 
\prod_{i=1}^{u}(1+\rho \lambda_i)^{k}  \prod_{l=1}^{q} \prod_{i=1}^{(n-l)\land m} \left(\frac{1-{e}^{-\rho \gamma_l \lambda_i}}{\rho \gamma_l \lambda_i}\right).
\label{eq:pdf-ev-F3|1}
\end{IEEEeqnarray}

Next, since for each realization of $H_{SD}$ the eigenvalues of $H_{SD}H_{SD}^\dagger$ and $H_{SD}^\dagger H_{SD}$ are the same, the conditional joint pdf of the eigenvalues of $W_2$ given the eigenvalues of $W_1$ can also be derived from Theorem \ref{thm_eigenvalue_distribution_1} and is hence given as
\begin{IEEEeqnarray}{ll}
F_{W_2|W_1}(\bar{\mu}|\bar{\lambda})\dot{=}\prod_{j=1}^{p}(\mu_j^{(k+m-2j)}e^{-\mu_j}) &\prod_{\substack{(i=1,j=1)\\ ((i+j)= (m+1))}}^{(u,p)} \left({e}^{-\rho \mu_j \lambda_i}\right) 
\prod_{i=1}^{u}(1+\rho \lambda_i)^{k}  \prod_{j=1}^{p} \prod_{i=1}^{(m-j)\land n} \left(\frac{1-{e}^{-\rho \mu_j \lambda_i}}{\rho \mu_j \lambda_i}\right).
\label{eq:pdf-ev-F2|1}
\end{IEEEeqnarray}
Substituting equations \eqref{eq:pdf-ev-F3|1} and \eqref{eq:pdf-ev-F2|1} in \eqref{eq:pdf-second-product} and importing the expression for $F_{W_1}(\bar{\lambda})$ from \cite{CZZ} the joint pdf of $(\bar{\lambda},\bar{\mu},\bar{\gamma})$ up to its exponential order can be obtained.

Recall next that for the DMT analysis we need the joint pdfs of the negative SNR exponents of these eigenvalues, i.e., those of the transformed variables $(\bar{\alpha},\bar{\beta},\bar{\delta})$ defined via
\begin{IEEEeqnarray}{rl}
\label{eq:transform-of-lambda-alpha}
\lambda_i=\rho^{-\alpha_i},~1\leq i\leq u;\\
\label{eq:transform-of-mu-beta}
\mu_j=\rho^{-\beta_j},~1\leq j\leq p;\\
\label{eq:transform-of-gamma-delta}
\gamma_l=\rho^{-\delta_l},~1\leq l\leq q.
\end{IEEEeqnarray}
Using the above change of variables in equation \eqref{eq:pdf-second-product} we get the joint pdf of $(\bar{\alpha},\bar{\beta},\bar{\delta})$ (where each vector in this triple is simply the vector of the corresponding random variables), denoted as $f_{W_1W_2W_3}(\bar{\alpha},\bar{\beta},\bar{\delta})$, given as
\begin{eqnarray}
\label{eq:pdf-third-product}
f_{W_1W_2W_3}(\bar{\alpha},\bar{\beta},\bar{\delta})=
\left(\Pi_{j=1}^p|J(\mu_j)|\right)g_{W_2|W_1} (\bar{\beta}|\bar{\alpha})\left(\Pi_{l=1}^q|J(\gamma_l)|\right)g_{W_3|W_1} (\bar{\delta}|\bar{\alpha}) \left(\Pi_{i=1}^u |J(\lambda_i)|\right) g_{W_1} (\bar{\alpha}),
\end{eqnarray}
where $ g_{W_2|W_1}(\bar{\beta}|\bar{\alpha})$, $g_{W_3|W_1}(\bar{\delta}|\bar{\alpha})$ and $g_{W_1}(\bar{\alpha})$ are obtained by replacing the three sets of arguments in $(\bar{\lambda},\bar{\mu},\bar{\gamma})$ using the transformations \eqref{eq:transform-of-lambda-alpha}-\eqref{eq:transform-of-gamma-delta} in $ F_{W_2|W_1}(\bar{\mu}|\bar{\lambda}) $, $F_{W_3|W_1}(\bar{\gamma}|\bar{\lambda})$ and $F_{W_1}(\bar{\lambda})$, respectively. The quantities $J(\lambda_i) = - {\rho}^{-\alpha_i} \ln \rho $, $J(\mu_j) = - \rho^{-\beta_j} \ln \rho $ and $J(\gamma_l) = - \rho^{-\delta_l} \ln \rho $ represent the Jacobians of the transformations in equations \eqref{eq:transform-of-lambda-alpha}-\eqref{eq:transform-of-gamma-delta}, respectively.

We next evaluate the three sets of products of Jacobians and the associated $g$ functions in the overall product expression in the right hand side of equation \eqref{eq:pdf-third-product} up to exponential order. We begin with $ \left(\Pi_{j=1}^p|J(\mu_j)|\right)g_{W_2|W_1} (\bar{\beta}|\bar{\alpha}) $ first.
Using the transformations \eqref{eq:transform-of-lambda-alpha} and  \eqref{eq:transform-of-mu-beta} in equation \eqref{eq:pdf-ev-F2|1} we get
\begin{IEEEeqnarray}{rl}
g_{W_2|W_1}(\bar{\beta}|\bar{\alpha}) \, \dot{=} \, \prod_{j=1}^{p}\left(\rho^{-(k+m-2j)\beta_j}e^{-\rho^{-\beta_j}}\right) &\prod_{\substack{(i=1,j=1)\\ (s. t. \; i+j= m+1)}}^{(u,p)} \left(e^{-\rho^{(1- \beta_j -\alpha_i)}}\right)
\prod_{i=1}^{u}(1+\rho^{-\alpha_i})^{k}  \nonumber \\ &\prod_{j=1}^{p} \prod_{i=1}^{(m-j)\land n} \left(\frac{1-{e}^{-\rho^{(1- \beta_j -\alpha_i)}} }{\rho^{(1- \beta_j -\alpha_i)}}\right).
\label{eq_temp1a}
\end{IEEEeqnarray}
For asymptotic SNR ($\rho \to \infty$) we have
\begin{IEEEeqnarray}{l}
\lim_{\rho \to \infty} e^{-\rho^{-\beta_j}}= 0 , ~\textrm{if} ~\beta_j<0 ~\textrm{for any}~ 1\leq j\leq p;\\
\lim_{\rho \to \infty} e^{-\rho^{(1- \beta_j -\alpha_i)}}=0, ~\textrm{if} ~(\alpha_i+\beta_j)<1 ~\textrm{for any}~ (i+j)\geq (m+1);\\
\label{eq_asymptotic_last_term}
\lim_{\rho \to \infty} \left(\frac{1-{e}^{-\rho^{(1- \beta_j -\alpha_i)}} }{\rho^{(1- \beta_j -\alpha_i)}}\right)=\left\{\begin{array}{cc}
\rho^{-(1- \beta_j -\alpha_i)}& ~\textrm{if} ~(\beta_j+\alpha_i)\leq 1;\\
1& \textrm{otherwise};
\end{array}\right.~\left[\because \lim_{x\to 0} \frac{1-e^{-x}}{x}=1\right].
\end{IEEEeqnarray}

Substituting the above asymptotic approximations and the fact that the limiting value of a product of convergent sequences is equal to the product of the individual limiting values in equation \eqref{eq_temp1a}, we get
\begin{equation}
\label{eq:pdf-snr-exp-F2|1}
\left(\Pi_{j=1}^p|J(\mu_j)|\right)g_{W_2|W_1}(\bar{\beta}|\bar{\alpha}) \, \dot{=} \,
\left\{\begin{array}{cc}
{\rho}^{-E_2(\bar{\alpha},\bar{\beta})}, &\textrm{if}~(\bar{\alpha},\bar{\beta})\in \mathcal{S}_2 ;\\
0, &\textrm{otherwise},
\end{array}\right.
\end{equation}
where $\mathcal{S}_2=\{(\bar{\alpha},\bar{\beta}): 0\leq \alpha_1\leq \cdots \leq \alpha_u; 0\leq \beta_1\leq \cdots \leq \beta_p; (\beta_j+\alpha_i)\geq 1,~\forall (i+j)\geq (m+1)\} $ and
\begin{IEEEeqnarray}{rl}
E_2(\bar{\alpha},\bar{\beta})=\left[\sum_{j=1}^{p}(m+k-2j+1)\beta_j-k\sum_{i=1}^{u}(1-\alpha_i)^+ +\sum_{\substack{i,j=1\\j+i\leq m}}^{u,p}(1-\alpha_i-\beta_j)^+\right].
\end{IEEEeqnarray}

Similarly, it can be shown that
\begin{equation}
\label{eq:pdf-snr-exp-F3|1}
\left(\Pi_{l=1}^q|J(\gamma_l)|\right)g_{W_3|W_1}(\bar{\delta}|\bar{\alpha})\dot{=}
\left\{\begin{array}{cc}
{\rho}^{-E_3(\bar{\alpha},\bar{\delta})}, &\textrm{if}~(\bar{\alpha},\bar{\delta})\in \mathcal{S}_3 ;\\
0, &\textrm{otherwise},
\end{array}\right.
\end{equation}
where $\mathcal{S}_3=\{(\bar{\alpha},\bar{\beta}): 0\leq \alpha_1\leq \cdots \leq \alpha_u; 0\leq \delta_1\leq \cdots \leq \delta_q; (\delta_l+\alpha_i)\geq 1,~\forall (i+l)\geq (n+1)\} $ and
\begin{IEEEeqnarray}{rl}
E_3(\bar{\alpha},\bar{\delta})=\left[\sum_{l=1}^{q}(n+k-2l+1)\delta_l-k\sum_{i=1}^{u}(1-\alpha_i)^+ +\sum_{\substack{i,l=1\\l+i\leq n}}^{u,q}(1-\alpha_i-\delta_l)^+\right].
\end{IEEEeqnarray}

Finally, using the expression for the pdf of $\bar{\alpha}$ given in \cite{tse1} we have
\begin{equation}
\label{eq:pdf-snr-exp-F1}
\left(\Pi_{i=1}^u |J(\lambda_i)|\right) g_{W_1}(\bar{\alpha})\dot{=}\left\{
\begin{array}{cc}
\rho^{-\sum_{i=1}^{u}(m+n-2i+1)\alpha_i}, &\textrm{if}~ 0\leq \alpha_1\leq \cdots \leq \alpha_u;\\
0,& \textrm{otherwise}.
\end{array}\right.
\end{equation}
Finally, substituting equations \eqref{eq:pdf-snr-exp-F2|1}, \eqref{eq:pdf-snr-exp-F3|1} and \eqref{eq:pdf-snr-exp-F1} into equation
\eqref{eq:pdf-third-product} we get the main result of this section, namely, the joint distribution of $(\bar{\alpha},\bar{\beta},\bar{\gamma})$ up to exponential order, which we state in the following theorem.

\begin{thm}
\label{thm:eigenvalue-exponent-distribution}
If $\bar{\alpha}$, $\bar{\beta}$ and $\bar{\gamma}$ are the vectors containing the negative SNR exponents of the ordered eigenvalues of the matrices $W_1$, $W_2$ and $W_3$, respectively, as defined in the transformations \eqref{eq:transform-of-lambda-alpha}-\eqref{eq:transform-of-gamma-delta}, then the joint distribution of $(\bar{\alpha},\bar{\beta},\bar{\gamma})$ is given up to exponential order as
\begin{equation}
\label{eq:pdf-abd}
f_{W_1W_2W_3}(\bar{\alpha},\bar{\beta},\bar{\delta})\dot{=}\left\{
\begin{array}{l}
\rho^{-E(\bar{\alpha},\bar{\beta},\bar{\delta})},~ \textrm{if}~(\bar{\alpha},\bar{\beta},\bar{\delta})\in \mathcal{S};\\
0,~ \textrm{if}~(\bar{\alpha},\bar{\beta},\bar{\delta})\notin \mathcal{S},
\end{array}\right.
\end{equation}
where
\begin{IEEEeqnarray}{l}
\label{eq_support_set_of_pdf}
\mathcal{S}=\mathcal{S}_2\cap \mathcal{S}_3=\left\{(\bar{\alpha},\bar{\beta},\bar{\delta}):\begin{array}{c}(\alpha_i+\beta_j)\geq 1,~ \forall (i+j)\geq(m+1);\\
(\alpha_i+\delta_l)\geq 1,~ \forall (i+l)\geq(n+1); \\
0\leq \alpha_1\leq \cdots \leq \alpha_u,\\
0\leq \beta_1\leq \cdots \leq \beta_p,  \\
0\leq \delta_1\leq \cdots \leq \delta_q,
\end{array} \right\},
\end{IEEEeqnarray}
and
\begin{IEEEeqnarray}{rl}
\label{eq_distribution_abs}
E(\bar{\alpha},\bar{\beta},\bar{\delta})=\sum_{i=1}^{u}&(n+m-2i+1)\alpha_i + \sum_{j=1}^{p}(k+m-2j+1)\beta_j+ \sum_{l=1}^{q}(k+n-2l+1)\delta_j \nonumber \\
&-2k\sum_{i=1}^{u}(1-\alpha_i)^+ +\sum_{\substack{i,j=1\\j+i\leq m}}^{u,p}(1-\alpha_i-\beta_j)^+ +\sum_{\substack{i,l=1\\l+i\leq n}}^{u,q}(1-\alpha_i-\delta_l)^+.
\end{IEEEeqnarray}
\end{thm}

\section{DMT of the MIMO HD-RC}
\label{sec:DMT-HD-RC}

Assuming global CSI $\mathcal{H}$ at the relay node, it was proved in \cite{YEE} that the DCF protocol based on the CF scheme of \cite{Cover_Gamal} can achieve the DMT of the MIMO HD-RC. The actual DMT was however not obtained therein. Here, using the asymptotic eigenvalue distribution result of Theorem \ref{thm:eigenvalue-exponent-distribution} of the previous section, the DMT of the MIMO HD-RC is first characterized as the solution of a convex optimization problem (see Theorem \ref{thm_optimization_problem}) and then simplified to a two-variable optimization problem (see Theorem \ref{thm_simplified_optimization}). Moreover, since it is shown that the dynamic QMF protocol achieves this fundamental DMT, only knowledge of the optimal switching time is required at the relay to achieve the DMT of the MIMO HD-RC. This is in contrast to the DCF protocol of \cite{YEE} which requires global CSI at the relay. Later in Section~\ref{sec:only-CSIR}, it is proved that that even the switching time information, while sufficient, is not necessary under certain conditions. In particular, it is shown that the DMT of the $(n,1,n)$ and $(1,k,1)$ HD-RCs can be achieved without CSI at the relay. This is also the case for more general classes of MIMO HD-RCs but only for sufficiently high multiplexing gains.

To characterize the DMT, we first prove that $P_e^*(\rho)$ (see \eqref{eq:pestar}), the minimum average probability of decoding error achievable on the channel at an SNR of $ \rho$, is exponentially equal (recall definition in \eqref{eq:exp-equality}) to the probability of an appropriately defined {\em outage} event. In Subsections \ref{subsec:upper-bound} and \ref{subsec:lower-bound} we derive an upper bound and a lower bound for the outage probability, respectively, which are in turn exponentially equal. The lower bound to the outage probability is based on an upper bound on the instantaneous cut-set bound of the channel. The upper bound on the outage probability is derived from an achievable rate expression of the QMF protocol operating dynamically on the relay channel. Finally, analyzing these bounds in Section~\ref{subsec:optimization-problem}, we derive the DMT of the channel by computing the negative SNR exponent of the outage probability.

It is well known that on a slow fading point-to-point channel the maximum rate at which information can be reliably transferred to the receiver depends on the channel realization, and is hence a random quantity. In what follows, this rate will be referred to as the instantaneous capacity of the channel. For a particular channel realization, if the rate of transmission is larger than the instantaneous capacity of a point-to-point channel, we say the channel is in outage. The same is true for a relay channel, where in addition a relay node helps the end-to-end transmission between the source and the destination nodes. Further, on a dynamic HD-RC, the instantaneous capacity of the channel also depends on the switching time of the relay node and should be chosen optimally to maximize it. Let $\hat{t}_d(\mathcal{H})$ represent the optimal switching time and let the instantaneous capacity be denoted as $C_o\left(\mathcal{H},\hat{t}_d(\mathcal{H})\right)$. Using this notation we next define the outage event.

\begin{defn}[Outage event]
\label{def:outage-event}
The HD-RC is said to be in outage if for a particular channel realization, $\mathcal{H}$ and SNR $\rho$, and the rate of transmission, $R=r\log(\rho)$ (in bits per channel use (Bpcu)) is larger than its instantaneous capacity. The corresponding outage event is denoted as $\mathcal{O}$, so that
\begin{IEEEeqnarray}{c}
\label{eq:outage-event}
\mathcal{O}\triangleq \left\{\mathcal{H}:C_o\left(\mathcal{H},\hat{t}_d(\mathcal{H})\right)<r\log(\rho)\right\}.
\end{IEEEeqnarray}
\end{defn}
Let $ \Pr(\mathcal{O} )$ denote the outage probability and let $d_{O}(r)$ denote its diversity order, i.e., $ d_{O}(r) \triangleq \lim_{\rho\to \infty}-\frac{\log(\Pr(\mathcal{O}))}{\log(\rho)} $. We have the following lemma.
\begin{lemma}
\label{lem:Pe-equalto-Po}
The minimum probability of decoding error achievable on the MIMO HD-RC, $ P_e^*(\rho) $ (see \eqref{eq:pestar}) is exponentially equal to the outage probability. Hence the corresponding diversity orders are also equal, so that
\begin{equation}\label{eq:Pe-is-exponentially-equal-to-Po}
    P_e^*(\rho) \doteq \Pr(\mathcal{O}) \quad  \Longrightarrow \quad d^*(r)=d_{O}(r),
\end{equation}
where $d^*(r)$ is defined in \eqref{eq:def-diversity-order}.
\end{lemma}
\begin{proof}[Proof of Lemma~\ref{lem:Pe-equalto-Po}]
The proof is identical to that in \cite{tse1}.
\end{proof}

In the next section, an upper bound on the DMT of the MIMO HD-RC is obtained.
\subsection{An upper bound on instantaneous capacity (and DMT)}
\label{subsec:upper-bound}
From the discussion in Section \ref{sec_introduction} we have that an upper bound on the DMT for the family of  fixed and dynamic protocols is also an upper bound on the achievable DMT of any (cooperative) communication scheme on the MIMO HD-RC. Thus we restrict attention, without loss of generality, to the family of fixed and dynamic protocols. Assuming that the relay node switches from the listening mode to the transmit mode at time $t_d$, we have that any achievable rate $R$ on the relay channel for which the error probability can be made arbitrarily small is upper bounded using the cut-set bounds for the HD-RC
\cite{CT,Khojastepour_Sabharwal_Aazhang} so that
\begin{IEEEeqnarray}{rl}
\label{eq:rate-upper-bound-intermediate1}
R\leq \max_{\{t_d, P(X_S,X_R)\}}\min \left\{I_{C_S}(t_d), I_{C_D}(t_d)\right\}=\max_{t_d} \bar{C}(\mathcal{H},t_d),
\end{IEEEeqnarray}
 where $\bar{C}(\mathcal{H},t_d)$ denotes the cut-set bound of the channel for a given $t_d$ and
\begin{IEEEeqnarray}{rl}
\label{eq:cut-set-bound-around-source}
I_{C_S}(t_d)=t_d I\left(X_S;Y_D,Y_R|p_1\right)+(1-t_d) I\left(X_S;Y_D|X_R,p_2\right);\\
\label{eq:cut-set-bound-around-destination}
I_{C_D}(t_d)=t_d I\left(X_S;Y_D|p_1\right)+(1-t_d) I\left(X_S,X_R;Y_D|p_2\right),
\end{IEEEeqnarray}
represent the maximum mutual information that can flow across the cuts around the source and destination, respectively.

The following two-part lemma provides (i) upper bounds to both $I_{C_S}$ and $I_{C_D}$ and (ii) a lower bound to $\bar{C}(\mathcal{H},t_d)$.
\begin{lemma}
\label{lem:cut-set-upper-bound}

\begin{itemize}
\item[i.]
The cut-set mutual informations $I_{C_S}(t_d)$ and $I_{C_D}(t_d)$ in equations \eqref{eq:cut-set-bound-around-source} and \eqref{eq:cut-set-bound-around-destination}, are upper bounded as
\begin{IEEEeqnarray}{rl}
\label{eq_cutset_Bs}
\max_{\{P_{X_S,X_R}\}} I_{C_S}(t_d) \, \leq \, I^{'}_{C_S}(t_d) \, ~ \triangleq ~ & \, t_d \log\left(L_{S,RD}\right)+(1-t_d) \log\left(L_{SD}\right),\\
\label{eq_cutset_Bd}
\max_{\{P_{X_S,X_R}\}} I_{C_D}(t_d)  \leq I^{'}_{C_D}(t_d) ~ \triangleq ~ & t_d \, \log\left(L_{SD}\right)+(1-t_d) \log\left(L_{SR,D}\right),
\label{eq_expression_Lsd}
\end{IEEEeqnarray}
where $H_{S,RD} \triangleq \left[\begin{subarray}{c} H_{SR} \\ H_{SD}\end{subarray}\right]$, $H_{SR,D} \triangleq \left[H_{SD}~ H_{RD} \right]$ and
\begin{IEEEeqnarray}{rl}
L_{SD} ~ \triangleq ~ & \det\left(H_{SD}H_{SD}^{\dagger}\rho+I_n\right), \\
\label{eq_expression_Lsr_d}
L_{SR,D} ~ \triangleq ~ &\det\left(\rho H_{SR,D}H_{SR,D}^{\dagger}+I_n\right);\\
\label{eq_expression_Lsrd1}
L_{S,RD} ~ \triangleq ~ &\det\left(\rho H_{S,RD}H_{S,RD}^{\dagger}+I_{n+k}\right).
\end{IEEEeqnarray}

\item[ii.]
Moreover, the cut-set bound $\bar{C}(\mathcal{H},t_d)$ is lower bounded as follows:
\begin{IEEEeqnarray*}{rl}
\bar{C}(\mathcal{H},t_d)\geq & \min \{I_{C_S}^{'}(t_d), I_{C_D}^{'}(t_d)\}-(m+k).
\end{IEEEeqnarray*}
\end{itemize}
\end{lemma}

\begin{proof}
See Appendix \ref{app:lem:cut-set-upper-bound}.
\end{proof}

Now, continuing from equation \eqref{eq:rate-upper-bound-intermediate1} we have
\begin{IEEEeqnarray}{rl}
R\leq & \max_{\{t_d, P(X_S,X_R)\}}\min \left\{I_{C_S}(t_d), I_{C_D}(t_d)\right\},\nonumber\\
\leq & \max_{\{t_d\}}\min \left\{ \max_{\{P(X_S,X_R)\}} I_{C_S}(t_d), \max_{\{P(X_S,X_R)\}} I_{C_D}(t_d)\right\},\nonumber\\
\label{eq_dmt_rate_upper_bound_1}
\stackrel{(a)}{\leq} &\max_{\{t_d\}}\min \left\{I^{'}_{C_S}(t_d), I^{'}_{C_D}(t_d)\right\}=\max_{t_d} R_U(t_d),
\end{IEEEeqnarray}
where in step $(a)$ we used the set of upper bounds from the first part of Lemma~\ref{lem:cut-set-upper-bound} and  the definition $R_U(t_d)=\min \left\{I^{'}_{C_S}(t_d), I^{'}_{C_D}(t_d)\right\}$. Note that $t_d$ in equation \eqref{eq_dmt_rate_upper_bound_1} can be a function of the channel matrices since we are considering a dynamic HD-RC. Since the right hand side of equation \eqref{eq_dmt_rate_upper_bound_1} is maximized when $I^{'}_{C_S}(t_d)= I^{'}_{C_D}(t_d)$, equating equations \eqref{eq_cutset_Bs} and \eqref{eq_cutset_Bd} we get the optimal value for the switching time as
\begin{IEEEeqnarray}{rl}
\label{eq:cut-set-optimal-switching-time}
t_d^*=\frac{\log\left(\frac{L_{SR,D}}{L_{SD}}\right)}{\log\left(\frac{L_{SR,D}}{L_{SD}}\right)+\log\left(\frac{L_{S,RD}}{L_{SD}}\right)}.
\end{IEEEeqnarray}
Putting this value of $t_d$ in equation \eqref{eq_dmt_rate_upper_bound_1} we get
\begin{IEEEeqnarray}{rl}
\label{eq:rate-upper-bound}
R \leq \frac{\log\left(\frac{L_{SR,D}}{L_{SD}}\right) \log\left(\frac{L_{S,RD}}{L_{SD}}\right)}{\log\left(\frac{L_{SR,D}}{L_{SD}}\right)+\log\left(\frac{L_{S,RD}}{L_{SD}}\right)}
+\log\left(L_{SD}\right)\triangleq R^*_U.
\end{IEEEeqnarray}
Since any rate up to the instantaneous capacity $C_o\left(\mathcal{H},\hat{t}_d(\mathcal{H})\right)$ is achievable, we have $ C_o\left(\mathcal{H},\hat{t}_d(\mathcal{H})\right)\leq R^*_U $.
This inequality when used along with the definition of the outage probability in \eqref{eq:outage-event} yields
\begin{equation}
\label{eq-ou}
    \mathcal{O}=\{\mathcal{H}:C_o\left(\mathcal{H},\hat{t}_d(\mathcal{H})\right)<r\log(\rho)\}\supseteq \{\mathcal{H}:R_U^*<r\log(\rho)\}\triangleq \mathcal{O}_U,
\end{equation}
from which we have a lower bound on the outage probability, $ \Pr\{\mathcal{O}\}\geq \Pr\{\mathcal{O}_U\} $. Using \eqref{eq:Pe-is-exponentially-equal-to-Po}, we then obtain an exponential lower bound on the minimum achievable probability of decoding error, and hence an upper bound on the DMT as
\begin{IEEEeqnarray}{rl}
P_e^*(\rho) \dot{\geq } \Pr\{\mathcal{O}_U\} \quad \Longrightarrow \quad d^*(r) \leq d_U(r)
\label{eq:Pe-lower-bound}
\end{IEEEeqnarray}
where $ d_U(r)$ is the diversity order of $ \Pr\{\mathcal{O}_U\} $, i.e., $ d_{U}(r) \triangleq \lim_{\rho\to \infty}-\frac{\log(\Pr(\mathcal{O}_U))}{\log(\rho)} $.


\subsection{A lower bound on instantaneous capacity (and the DMT) via the QMF scheme}
\label{subsec:lower-bound}
Since the instantaneous capacity of a slow fading channel is the supremum of the achievable rates of all possible coding schemes, the achievable rate of a particular coding scheme yields a lower bound to it. We first derive such a lower bound for the HD-RC by computing the achievable rate of the QMF protocol~\cite{Avestimehr_Diggavi_Tse} which when substituted in the definition of the outage event results in an upper bound to the outage probability yielding in turn the desired lower bound to the DMT. But first a brief review of the QMF coding scheme.

\subsubsection{The coding and decoding strategies of the QMF scheme} ~The encoding method at the source node of the QMF scheme is a two step procedure involving two different codes. The inner codebook, denoted by $\mathcal{T}_{x_S}$ has $2^{RT}$ mutually independent codewords, i.e.,
\begin{equation*}
    \mathcal{T}_{x_S}=\{x_S^{(w)}:w=1,\cdots,  2^{RT}\},
\end{equation*}
where $x_S^{(w)}$ for each $w\in \{1,\cdots,  2^{RT}\}$ is a random $T$ length vector with i.i.d. $\mathcal{CN}(0,1)$ components. Each codeword of the inner codebook is treated as a symbol of the outer code. To transmit a message $\mathcal{U}\in \{1,\cdots,  2^{NRT}\}$, the source node first maps it onto a $N$ length sequence of symbols of the outer code, i.e., $w_1, w_2, \cdots, w_N$. Each of these symbols are then encoded into a codeword in $\mathcal{T}_{x_S}$, i.e., $w_k \to x_S^{(w_k)}$ for $1\leq k\leq N$. Finally, the message $\mathcal{U}$ is transmitted over $NT$ channel uses and hence at a rate of $R$ bits per channel use.

The relay node operates on each codeword of the inner codebook in two phases (listen and transmit). For the first $t_d$ fraction of each codeword $x_S^{(w_k)}$ the relay node listens to the source transmission. Let $y_R^{(w_k)}$ represent the received signal at the relay during the listening phase with the transmission of $x_S^{(w_k)}$. The relay node first quantizes $y_R^{(w_k)}$ at the noise floor and randomly maps the resulting quantized signal $\hat{y}_R^{(w_k)}$ to a random Gaussian codeword, $x_R^{(w_k)}$, i.e., $x_R^{(w_k)}=f_R(\hat{y}_R^{(w_k)})$, where $f_R(.)$ represents a random mapping scheme. The random vector $x_R^{(w_k)}$ is then transmitted by the relay during the remaining $(1-t_d)T$ channel uses. The same procedure is repeated by the relay for all $N$ inner codewords that make up the source codeword. Given the knowledge of the relay mapping, $f_R(.)$, global CSI $\mathcal{H}$, and the received sequence $y_D^{(w_k)}\in \mathbb{C}^{n\times T}$ for $k=1,\cdots, N$, the destination node decodes the message.


\subsubsection{Achievable rate of the QMF protocol}
Recall that the cut-set upper bound to the instantaneous capacity of the channel for a given listen-transmit scheduling of the relay (i.e., fixed $t_d$) node was denoted by $\bar{C}(\mathcal{H},t_d)$~\cite{CT,Khojastepour_Sabharwal_Aazhang}. In \cite{Avestimehr_Diggavi_Tse} it was proved that for a given $t_d$, the QMF protocol can achieve a rate $R_q(\mathcal{H},t_d)$ on a relay channel with channel matrices  $\mathcal{H}$, where

\begin{IEEEeqnarray}{l}
\label{eq:lower-bound-a}
  R_q(\mathcal{H},t_d)\geq \bar{C}(\mathcal{H},t_d)-\tau,
\end{IEEEeqnarray}
and $\tau$ is independent of both the channel matrices and $\rho$. The above rate satisfying equation \eqref{eq:lower-bound-a} can be achieved by the QMF protocol for any given $t_d$ as long as it is known to the relay node. In particular, putting $t_d=t_d^*$ (given by equation \eqref{eq:cut-set-optimal-switching-time}) in equation \eqref{eq:lower-bound-a} we get

\begin{IEEEeqnarray}{l}
\label{eq:lower-bound-b}
  R_q(\mathcal{H},t_d^*)\geq \bar{C}(\mathcal{H},t_d^*)-\tau.
\end{IEEEeqnarray}
In other words, a rate which is within constant number of bits to $\bar{C}(\mathcal{H},t_d^*)$ can be achieved by the QMF protocol. Note that, $t_d^*$ is a function of the instantaneous channel realizations, $\mathcal{H}$ (e.g., see equation \eqref{eq:cut-set-optimal-switching-time}) and can be computed by the relay node since we assume global CSI at the relay node.

From the second part of Lemma~\ref{lem:cut-set-upper-bound} we have
\begin{IEEEeqnarray}{rl}
\label{eq:lower-bound-c}
  \bar{C}(\mathcal{H},t_d^*)\geq & \min\{I_{C_S}^{'}(t_d^*),I_{C_D}^{'}(t_d^*)\}-(m+k),\nonumber \\
  \geq  & R_U^*-(m+k),\nonumber
\end{IEEEeqnarray}
where the last step follows from the fact that $I_{C_S}^{'}(t_d^*)=I_{C_D}^{'}(t_d^*)=R_U^*$ (see equation \eqref{eq:cut-set-optimal-switching-time}). Now, substituting the last lower bound to $\bar{C}(\mathcal{H},t_d^*)$ in equation \eqref{eq:lower-bound-b} we get

\begin{IEEEeqnarray}{rl}
  R_q(\mathcal{H},t_d^*)\geq & R_U^*-\underbrace{(m+k+\tau)}_{R_0}, \nonumber \\
= & R_U^* - R_0, \nonumber \\
\label{eq:rate-lower-bound}
\triangleq & R_L^*,
\end{IEEEeqnarray}
where $   R_0= (m+k+\tau) $.

Clearly, the instantaneous capacity $C_o\left(\mathcal{H},\hat{t}_d(\mathcal{H})\right)$ is larger than or equal to any achievable rate on the channel, i.e., $
    C_o\left(\mathcal{H},\hat{t}_d(\mathcal{H})\right)\geq R^*_L $.
This inequality along with the definition of outage probability in \eqref{eq:outage-event} yields
\begin{equation*}
    \mathcal{O}=\{\mathcal{H}:C_o\left(\mathcal{H},\hat{t}_d(\mathcal{H})\right)<r\log(\rho)\}\subseteq \{\mathcal{H}:R_L^*<r\log(\rho)\}\triangleq \mathcal{O}_L,
\end{equation*}
which in turn implies that $ P_e^*(\rho) \doteq \Pr\{\mathcal{O}\}\leq \Pr\{\mathcal{O}_L\} $,
where the exponential equality is from \eqref{eq:Pe-is-exponentially-equal-to-Po}. Now, since $R_L^* = R^*_U - R_0$ and $R_0$ is independent of the SNR ($\rho$) and  $\mathcal{H}$, we have at asymptotically high SNR that
\begin{IEEEeqnarray*}{rl}
 \Pr\{\mathcal{O}_L\}= & \Pr\{R^*_L<r \log(\rho)\}= \Pr\{R^*_U-R_0<r \log(\rho)\}\\
 \dot{=}& \Pr\{R^*_U<r \log(\rho)\}=\Pr\{\mathcal{O}_U\}.
\end{IEEEeqnarray*}
Hence, $ P_e^*(\rho)\dot{\leq} \Pr\{\mathcal{O}_U\}$ and combining with \eqref{eq:Pe-lower-bound} we have that $\Pr\{\mathcal{O}_U\}$ characterizes $P_e^*(\rho)$ exactly up to exponential order, i.e.,
\begin{equation*}
    P_e^*(\rho)\dot{=}\Pr\{\mathcal{O}_U\}
    \end{equation*}
so that the DMT of the MIMO HD-RC can be expressed as
\begin{equation}
\label{eq:d-and-Ou-relation}
    d^*(r) = d_U(r) = \lim_{\rho\to \infty}-\frac{\log\left(\Pr\{\mathcal{O}_U\}\right)}{\log(\rho)}.
\end{equation}
with $ \mathcal{O}_U$ defined in \eqref{eq-ou} in terms of $R^*_U $ which in turn is defined in \eqref{eq:rate-upper-bound}. In the next section, we evaluate this DMT.

\begin{rem}
\label{thm:optimality-of-QMF-without-global-CSI}
The QMF protocol can achieve the DMT of the MIMO HD-RC with knowledge of only $t_d^*$, the switching time that maximizes the cut-set bound (in lieu of the true optimal switching time $ \hat{t}_d $)., i.e., it does not require the explicit knowledge of $H_{SR}$, $H_{SD}$ and $H_{RD}$.
\end{rem}

\subsection{The DMT as a solution to an optimization problem}
\label{subsec:optimization-problem}

Evidently, to obtain $d^*(r)$ the probability distribution of $R_U^*$, which is a function of the three channel matrices, is needed. However, by simplifying the expression for $R_U^*$, it is shown that just the joint eigenvalue distribution of the three composite channel matrices $W_i$, for $1\leq i\leq 3$, defined in Section \ref{subsec:eigenvalue_distribution}, suffices. Further simplification shows that only the joint distribution of the SNR exponents of these eigenvalues is sufficient to obtain $d^*(r)$.

\begin{lemma}
\label{lem:d-in-terms-of-R*}
The optimal diversity order of the MIMO HD-RC can be written as
\begin{IEEEeqnarray}{rl}
\label{eq:d-as-function-asymptotic-evs}
d^*(r)=\lim_{\rho \to \infty} -\frac{\log\left(\Pr\left\{r^*(\bar{\alpha},\bar{\beta},\bar{\delta})\leq r\right\}\right)}{\log(\rho)},
\end{IEEEeqnarray}
where $r^*(\bar{\alpha},\bar{\beta},\bar{\delta})$ is given as
\begin{equation}
\label{eq:asymptotic-rate-expression}
r^*(\bar{\alpha},\bar{\beta},\bar{\delta}) \triangleq  ~\left[  \frac{\sum_{l=1}^{q}(1-\delta_l)^+ \sum_{j=1}^{p}(1-\beta_j)^+}{\sum_{l=1}^{q}(1-\delta_l)^+ +\sum_{j=1}^{p}(1-\beta_j)^+} \right]+ \left(\sum_{i=1}^{u}(1-\alpha_i)^+\right);
\end{equation}
and $\alpha_i$'s, $\beta_j$'s and $\delta_l$'s are the negative SNR exponents of the eigenvalues of $H_{SD}H_{SD}^{\dagger}$, $H_{SR}(I_m+\rho H_{SD}^{\dagger}H_{SD})^{-1}H_{SR}^{\dagger}$ and $ H_{RD}^{\dagger}(I_n+\rho H_{SD}H_{SD}^{\dagger})^{-1}H_{RD}$, respectively.
\end{lemma}
\begin{proof}[Proof]
The proof is given in Appendix~\ref{App:lem:d-in-terms-of-R*}.
\end{proof}



Using the joint pdf of $\{\bar{\alpha},\bar{\beta},\bar{\delta}\}$ given by equation~\eqref{eq:pdf-abd}, $\Pr\left\{r^*(\bar{\alpha},\bar{\beta},\bar{\delta})\leq r\right\}$ is evaluated and using equation \eqref{eq:d-as-function-asymptotic-evs}, the optimal diversity order $d^*(r)$ is obtained, leading to the following theorem.

\begin{thm}
\label{thm_optimization_problem}
The solution of the following optimization problem yields the fundamental DMT of the MIMO HD-RC:
\begin{IEEEeqnarray}{l}
\label{general_opt}
\min_{(\bar{\alpha},\bar{\beta},\bar{\delta})}~ F\left(\bar{\alpha},\bar{\beta},\bar{\delta}\right)
\end{IEEEeqnarray}
subject to the following constraints
\begin{IEEEeqnarray}{rl}
\label{mi_constraint}
\sum_{i=1}^{u}(1-\alpha_i)+\frac{\sum_{j=1}^{p}(1-\beta_j)\sum_{l=1}^{q}(1-\delta_l)}{\sum_{j=1}^{p}(1-\beta_j)+\sum_{l=1}^{q}(1-\delta_l)}\leq & r,\\
\label{redun1}
\alpha_{m-j+1}+\beta_j \geq  1~ \forall 1\leq j \leq &p,\\
\label{redun12}
\alpha_{n-l+1}+\delta_l \geq  1~ \forall 1\leq l \leq &q,\\
\label{redun2}
0\leq \alpha_1 \leq \cdots \leq \alpha_u \leq & 1,\\
0\leq \beta_1 \leq \cdots \leq \beta_p \leq & 1, \\
\label{redun4}
0\leq \delta_1 \leq \cdots \leq \delta_q \leq & 1;
\end{IEEEeqnarray}
where
\begin{IEEEeqnarray}{ll}
\label{eq_objective_func}
F\left(\bar{\alpha},\bar{\beta},\bar{\delta}\right)=\sum_{i=1}^{u}(n+m+2k-2i+1)\alpha_i &+ \sum_{j=1}^{p}(k+m-2j+1)\beta_j+ \sum_{l=1}^{q}(k+n-2l+1)\delta_j -2ku \nonumber \\
&+\sum_{\substack{i,j=1\\j+i\leq m}}^{u,p}(1-\alpha_i-\beta_j)^+ +\sum_{\substack{i,l=1\\l+i\leq n}}^{u,q}(1-\alpha_i-\delta_l)^+.
\end{IEEEeqnarray}
\end{thm}
\begin{proof}[Outline of proof]
It is clear from equation \eqref{eq:d-as-function-asymptotic-evs} that the optimal diversity order is equal to the negative SNR exponent of $\Pr\left\{r^*(\bar{\alpha},\bar{\beta},\bar{\delta})\leq r\right\}$. Using the joint pdf of $(\bar{\alpha},\bar{\beta},\bar{\delta})$ obtained in Subsection \ref{subsec:eigenvalue_distribution}, this probability can be written as an integral of the pdf over the subset of the sample space of $(\bar{\alpha},\bar{\beta},\bar{\delta})$ where $r^*(\bar{\alpha},\bar{\beta},\bar{\delta}) \leq r$ (call it $\mathcal{D}$). From Laplace's method it follows that this integral is dominated by a term having the minimum negative SNR exponent over $\mathcal{D}$. The details are provided in Appendix \ref{app:thm_optimization_problem}.
\end{proof}

\begin{rem}
It is well known that the fundamental DMTs of the $(m,n)$ and the $(n,m)$ point-to-point MIMO channels are identical. From \eqref{dmt-fd-rc} it is also clear that the DMTs of the $(m,k,n)$ and the $(n,k,m)$ MIMO FD-RCs are identical. The above theorem proves that this {\em reciprocity} property of DMT extends to the MIMO HD-RC as well as can be seen from the symmetry in $m$ and $n$ of the optimization problem of \eqref{general_opt}. In other words, the fundamental DMTs on the $(m,k,n)$ and $(n,k,m)$ MIMO HD-RCs are identical. Henceforth, we let $m \geq n$ without loss of generality.
\end{rem}

Note that $\sum_{i=1}^{u}\alpha_i$, $\sum_{j=1}^{p}\beta_j$ and $\sum_{l=1}^{q}\delta_l$ are affine functions of the $\alpha_i$'s, $\beta_j$'s and $\delta_l$'s, respectively. Furthermore, by computing its Hessian, it can be easily proved that the function $\frac{(p-x)(q-y)}{(p-x)+(q-y)}$ is convex with respect to $x$ and $y$. Further, since $f(g(V))$ is convex whenever $f(.)$ is convex and $g(.)$ is affine, it is evident that $\frac{(p-\sum_{j=1}^{p}\beta_j)(q-\sum_{l=1}^{q}\delta_l)}{(p-\sum_{j=1}^{p}\beta_j)+(q-\sum_{l=1}^{q}\delta_l)}$ is a convex function. Hence the left hand side of the inequality constraint \eqref{mi_constraint} is a convex function. It is also clear that the objective function in \eqref{eq_objective_func} is convex and that all the other constraints \eqref{redun1}-\eqref{redun4} are convex. Thus the optimization problem of Theorem~\ref{thm_optimization_problem} is a convex optimization problem and it can be solved using convex programming methods \cite{BV}.

The number of variables in the optimization problem in \eqref{general_opt} however, increases with $m, k$ and $n$ linearly. In what follows, we show that the problem can be simplified to an optimization problem with only two variables, independently of $m, k$ and $n$.

\begin{figure}[htp]
  \begin{center}
    \subfigure[Cases where DMTs of FD- and HD-RCs are identical]{\label{fig_dmt_comparison1-a}\includegraphics[scale=0.5]{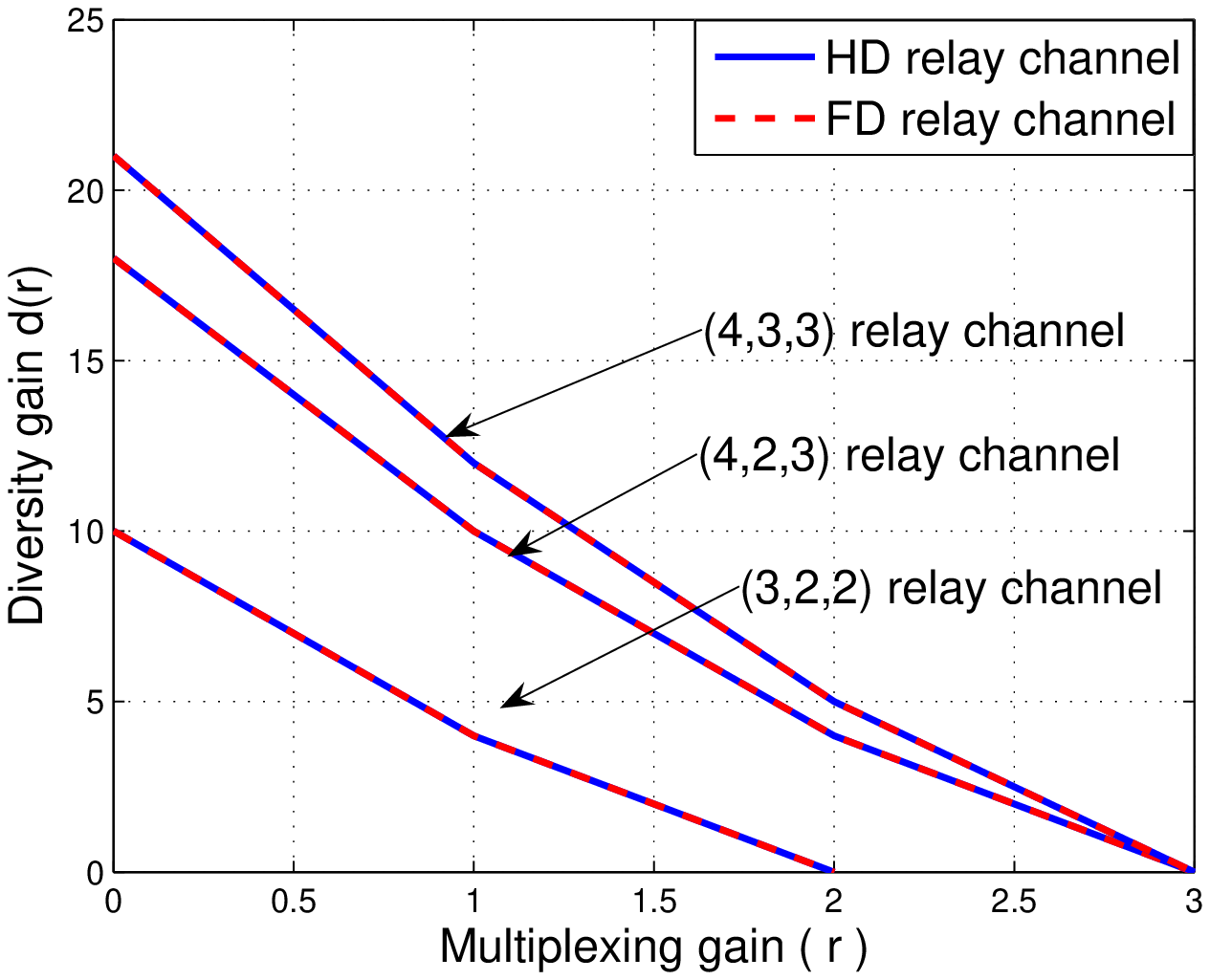}}
    \subfigure[Cases where DMTs of FD- and HD-RCs are not identical]{\label{fig_dmt_comparison1-b}\includegraphics[scale=0.5]{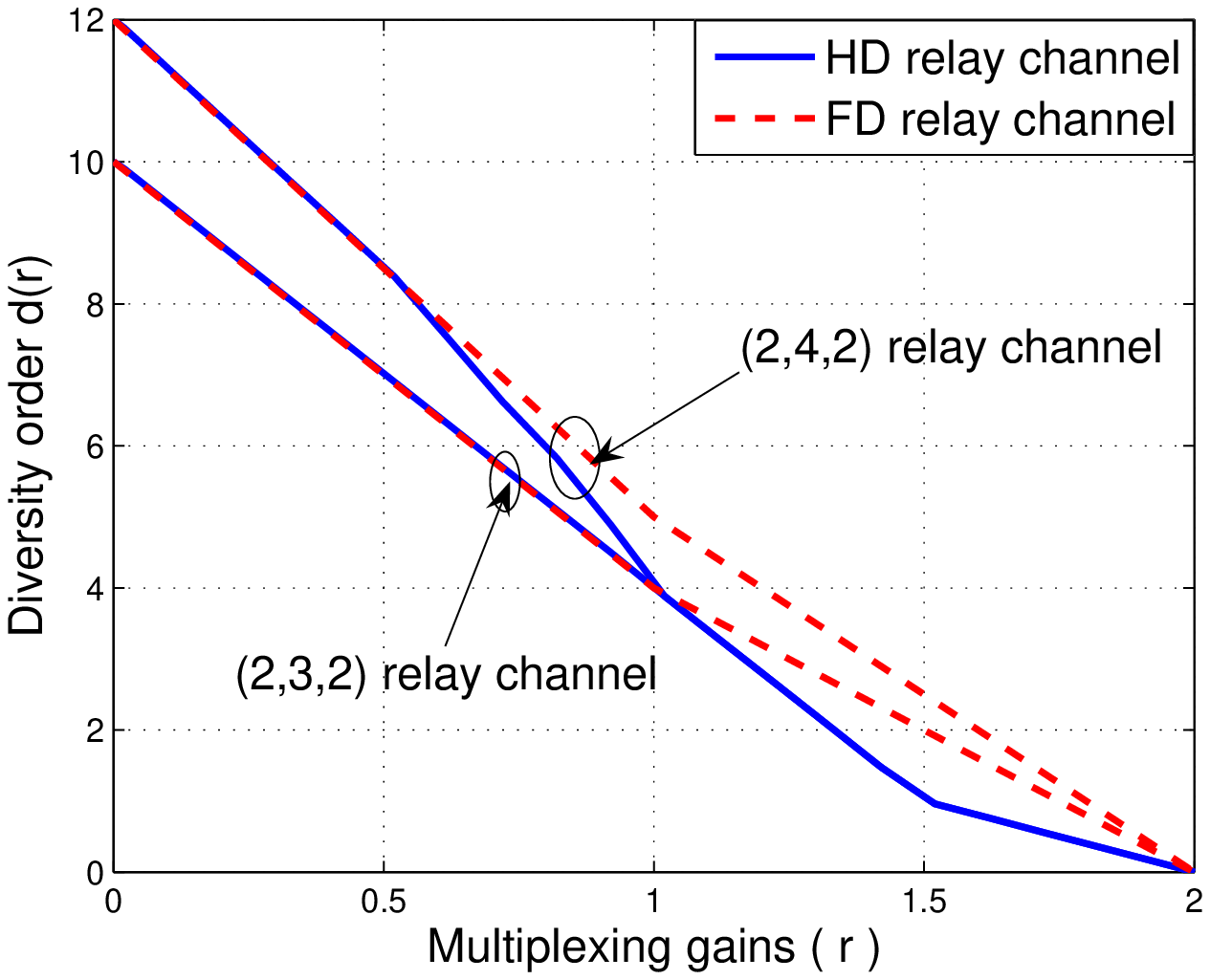}}
  \end{center}
  \caption{DMTs Comparisons for MIMO HD- vs. FD-RCs.}
  \label{fig_dmt_comparison1}
\end{figure}

\begin{thm}
\label{thm_simplified_optimization}
The fundamental diversity-muliplexing tradeoff of the $(m,k,n)$ HD-RC is given as
\begin{IEEEeqnarray}{l}
d^*(r)=\min_{\left\{a\in \mathcal{R},~b\in \mathcal{B} \right\}} F\left(\phi_{{\alpha}}(a),\phi_{{\beta}}(b),\phi_{{\delta}}\left(\frac{b(r-a)}{(b-r+a)}\right)\right),
\end{IEEEeqnarray}
where the interval $\mathcal{R}$ is specified in equation~\eqref{eq_region_R} in Appendix~\ref{pf_thm_simplified_opt}, $\mathcal{B}=\left[\frac{s_m(r-a)}{(s_m-r+a)}, b_m\right]$, $b_m=\min \{p,(m-a)\}$, $s_m=\min \{q,(n-a)\}$ and $\phi_i$'s are as defined in equations \eqref{eq_phi_alpha}-\eqref{eq_phi_gamma} in Appendix~\ref{pf_thm_simplified_opt}.
\end{thm}

\begin{proof}[Proof]
The proof is given in Appendix~\ref{pf_thm_simplified_opt}.
\end{proof}

\begin{ex}
We illustrate the advantage of relaying relative to point-to-point communication by considering the networks $\textrm{CN}_1$ and $\textrm{CN}_2$ of Fig. \ref{cooperative-networks}(a) and Fig. \ref{cooperative-networks}(b). Fig. \ref{scenario1-a} applies to the uplink of $\textrm{CN}_1$ (in which $m = k < n$) and depicts the DMT performance of the relay channel with respect to that achievable on the corresponding point-to-point channel. Similarly, Fig. \ref{scenario1-b} applies to the uplink of $\textrm{CN}_2$ (in which $m < k < n$). These figures clearly demonstrate the superior performance of cooperative MIMO over point-to-point MIMO communication.

\end{ex}

\begin{figure}[htp]
  \begin{center}
    \subfigure[Relay channels of $\textrm{CN}_1$]{\label{scenario1-a}\includegraphics[scale=0.5]{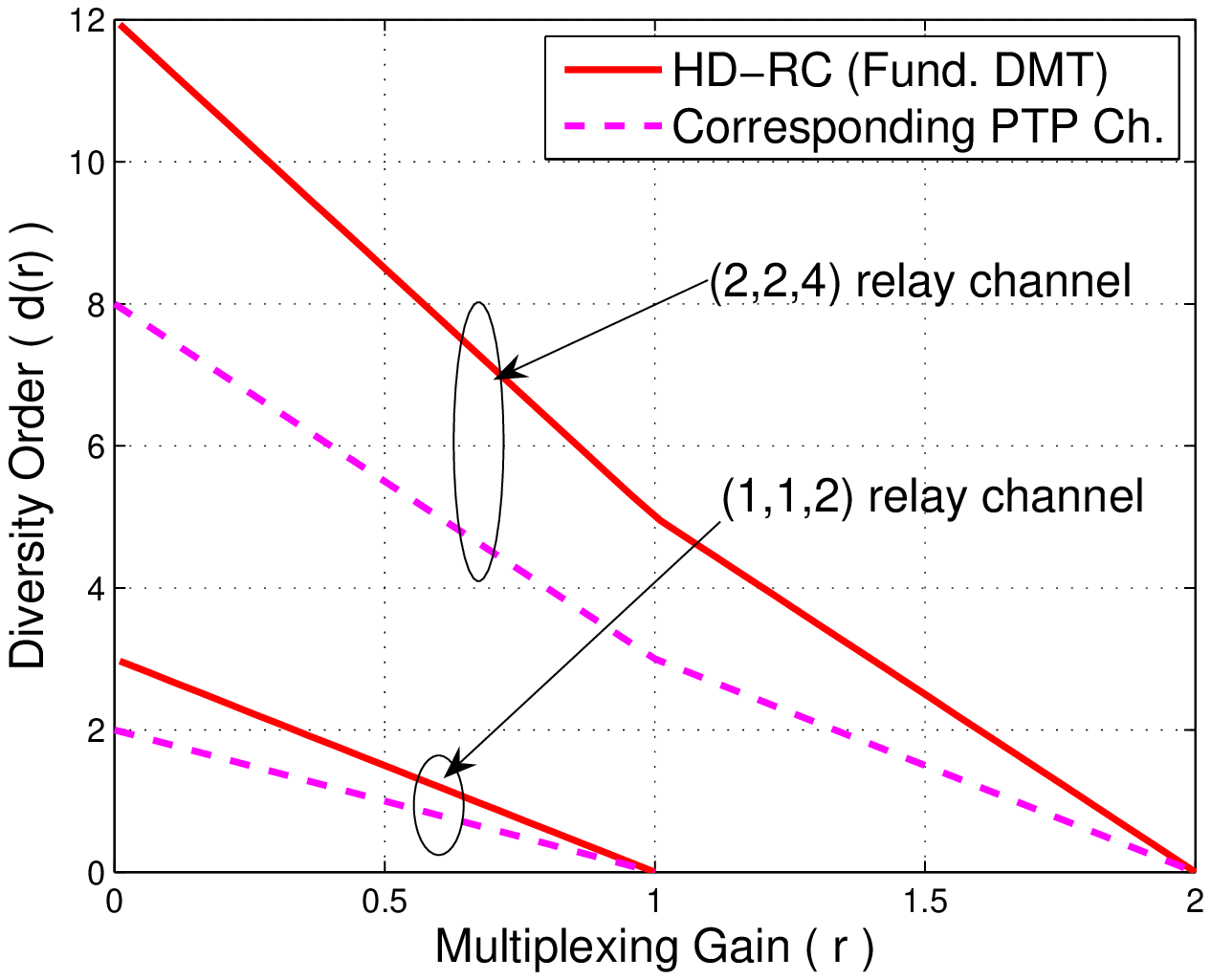}}
    \subfigure[Relay channels of $\textrm{CN}_2$]{\label{scenario1-b}\includegraphics[scale=0.5]{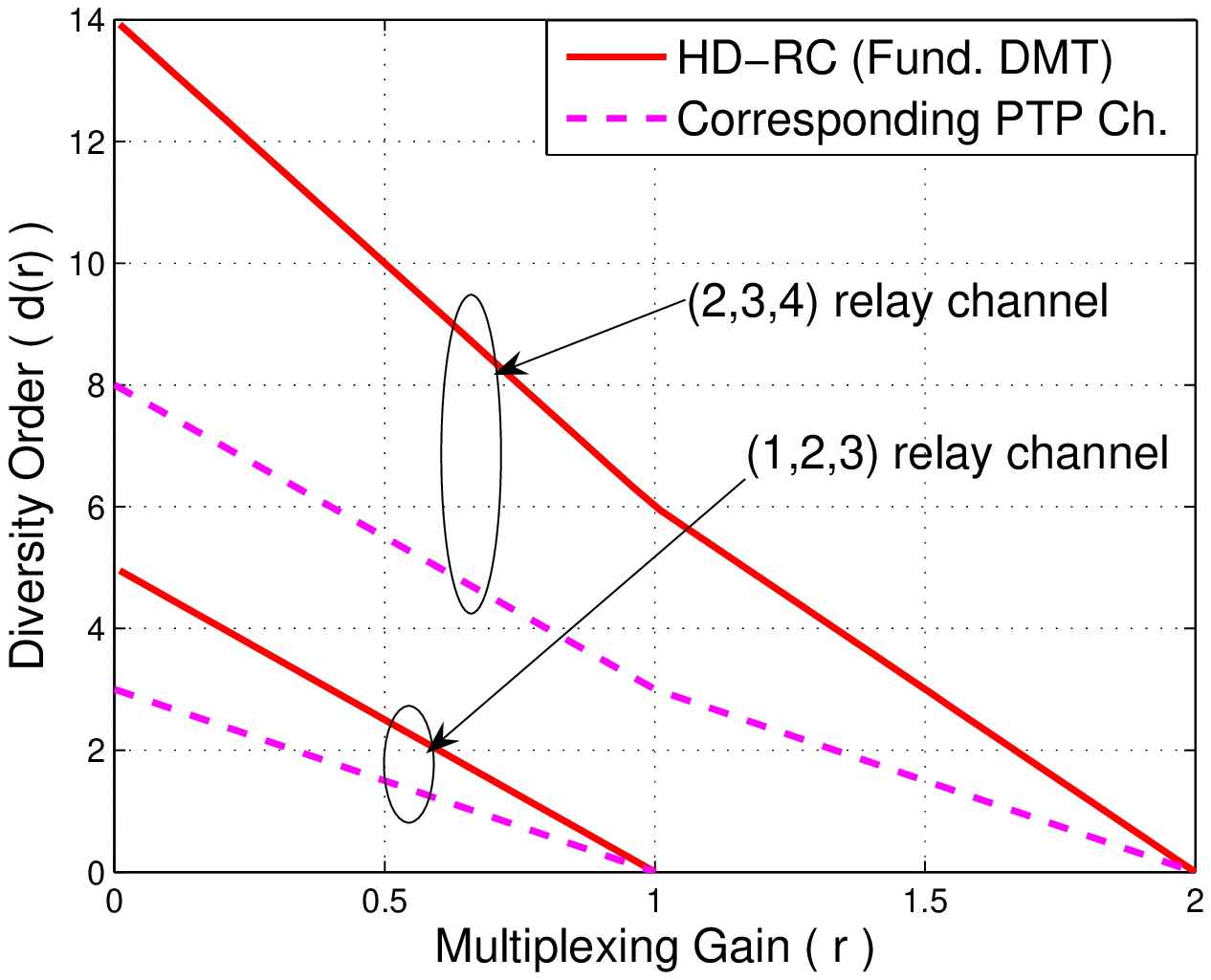}}
  \end{center}
  \caption{The fundamental DMTs of the MIMO HD-RC and the corresponding point-to-point MIMO channel.}
  \label{scenario1}
\end{figure}

\begin{rem}
\label{rem:fd=hd}
The explicit numerical computation of the fundamental DMT reveals several interesting characteristics of the MIMO HD-RC. For example, for the class of $(m,k,n)$ HD-RCs where $m>n\geq k$ we found that the DMT is identical to that of the corresponding MIMO FD-RC. It appears to be difficult however to show this analytically.  Note that this scenario applies to the downlink of the two networks $\textrm{CN}_1$ (with $m>n=k$) and $\textrm{CN}_2$ (with $m>n>k$) of Fig. \ref{cooperative-networks}(a) and Fig. \ref{cooperative-networks}(b), respectively.  Fig. \ref{fig_dmt_comparison1-a} illustrates this fact for a few specific examples of MIMO RCs. Thus, for the class of MIMO HD-RCs for which $m>n\geq k$, the half-duplex constraint does not appear to be restrictive in terms of DMT performance. In general however, the MIMO HD-RC has different DMT characteristics than the MIMO FD-RC. For instance, see Fig. \ref{fig_dmt_comparison1-b} (which is relevant for the sensor network $\textrm{CN}_3$ of Fig. \ref{cooperative-networks}(c)). This will be further discussed in the next section.
\end{rem}
\begin{conj}
For the class of $(m,k,n)$ HD-RCs for which $m>n\geq k$ the DMT is equal to that of the corresponding MIMO FD-RC.
\end{conj}

\section{Closed form expressions for the DMTs of a few classes of relay channels}
\label{sec:closed-form-expressions}
A closed form expression of the DMT would provide more insights about the system than a numerical solution. Motivated by this fact, we next provide closed-form solutions for the fundamental DMT of special classes of MIMO HD-RCs specified by the relationship between the numbers of antennas at the three nodes, including the $(n,k,n)$ (henceforth, called symmetric since $m=n$) HD-RC.

\begin{thm}
\label{thm_closedform_sol_symmetric_case}
The optimal diversity order $d^*(r)$, at a multiplexing gain of $r$, of the $(n,k,n)$ HD-RC is upper bounded by $d_U^s(r)$ (defined below) as
\begin{equation}
\label{eq_dmt_symmetric}
d^*(r)\leq d_U^s(r) \triangleq \left\{\begin{array}{c}
\min_{\{1\leq i\leq (3+p)\}}d_{U_i}(r), ~~k\leq n,\\
\min_{\{1\leq i\leq (3+2p)\}}d_{U_i}(r), ~~k \geq n,
\end{array}\right.
\end{equation}
where recall that $p=\min\{m,k\}=\min\{k,n\}$ (since $m=n$) and for $N=1, \cdots , p $, and recalling that $d_{(n_t,n_r)}(\cdot)$ represents the fundamental
DMT of a MIMO point-to-point channel with $n_t$ transmit and $n_r$ receive antennas, we define
\begin{IEEEeqnarray}{l}
d_{U_1}(r)=d_{(n,n+k)}(r), \quad {\rm for} \quad 0\leq r\leq n,\\
\label{eq_dmt_symmetric_bound2}
d_{U_2}(r)=d_{(2n,2n)}(2r), \quad {\rm for} \quad n-\frac{p}{2}\leq r\leq n,\\
\label{eq_dmt_symmetric_bound3}
d_{U_3}(r)=n^2+\sum_{l=1}^{p}(n+k-2l+1)\left(1-\left(\frac{pr}{(p-r)}-l+1\right)^+\right)^+,
\quad {\rm for} \quad 0\leq r\leq \frac{p}{2},\\
d_{U_{(3+N)}}(r)=N^2+d_{(n-N),(n+2k-N)}\left(r-\frac{N}{2}\right),
\quad {\rm for} \quad \frac{N}{2}\leq r\leq \min \left\{n-\frac{N}{2}, n-\frac{N^2}{(2p-N)} \right\},\\
\label{eq_dmt_symmetric_bound_32p}
d_{(3+p+N)}(r)=\sum_{i=1}^{(n-N)}(2n+k-N-2i+1)\left(1-\left(a_N-i+1\right)^+\right)^++N^2,
\; {\rm for} \; \frac{Nn}{(N+n)}\leq r\leq n-\frac{N}{2},
\end{IEEEeqnarray}
and $a_N$ is given by equation \eqref{eq_exp_for_a_N} in Appendix \ref{app-E}.
\end{thm}

\begin{proof}[Proof]
The proof is given in Appendix~\ref{app-E}.
\end{proof}

Evaluating the exact DMT of the $(n,k,n)$ relay channel using Theorem~\ref{thm_simplified_optimization} for several values of $n$ and $k$ we found that $d^*(r)=d_U^s(r)$ which leads us to make the following conjecture that the upper bound $d_U^s(r)$ is in fact tight.
\begin{conj}\label{conj:nkn}
On a symmetric $(n,k,n)$ relay channel $d^*(r)=d_U^s(r)$, where $d_U^s(r)$ is given by equation \eqref{eq_dmt_symmetric}.
\end{conj}

\begin{rem}
From the expression of $d_{U_2}(r)$ in equation \eqref{eq_dmt_symmetric_bound2} we see that this particular upper bound does not depend on $k$ for $k\geq n$. Thus, when this bound is active, adding an extra antenna at the relay node does not improve the DMT performance of the channel. This is an interesting difference between the HD- and FD-RCs since for FD-RCs every additional antenna at the relay improves the diversity order for all values of multiplexing gains
(recall \eqref{dmt-fd-rc}). Empirical results show that $d_{U_2}(r)$ is a tight bound for the DMT on the $(n,k,n)$ HD-RC for $r\geq \frac{n}{2}$ and $k\geq \lceil\frac{3n}{2}\rceil$. For example, Fig. \ref{fig_dmt_comparison1-b} illustrates this fact by showing that while adding an extra antenna on the $(2,3,2)$ FD-RC uniformly increases the achievable diversity orders at all multiplexing gains, the achievable diversity order on the corresponding HD-RC does not change for $r\geq 1$.
\end{rem}

\begin{rem}
\label{rem:extra-antenna-at-the-relay}
Another interesting fact revealed by Theorem~\ref{thm_closedform_sol_symmetric_case} is that, as the number of antennas increases at the relay node, the difference in the DMT performance between the FD-RC and the HD-RC increases. From the expression of the upper bound $d_{U_3}(r)$ in equation \eqref{eq_dmt_symmetric_bound3} we see that at a multiplexing gain of $r=\frac{p}{2}$ the diversity order achievable on the HD-RC is upper bounded by $n^2$ but on an FD-RC this is clearly not the case where the diversity order increases with $k$. Fig. \ref{fig_dmt_comparison2-a} demonstrates this phenomenon on a $(2,k,2)$ relay channel which is applicable to the $\textrm{CN}_3$ scenario of Fig. \ref{cooperative-networks}(c). Intuitively, the above phenomenon occurs because as the number of antennas at the relay increases the signal forwarded by the relay node can significantly contribute to enhancing the diversity of the received signal at the destination node and hence the half-duplex constraint becomes increasingly more restrictive (relative to the FD relay) because the relay node can not transmit in the listening phase.

\end{rem}

\begin{figure}[htp]
  \begin{center}
    \subfigure[MIMO HD-RC Vs. MIMO FD-RC.]{\label{fig_dmt_comparison2-a}\includegraphics[scale=0.5]{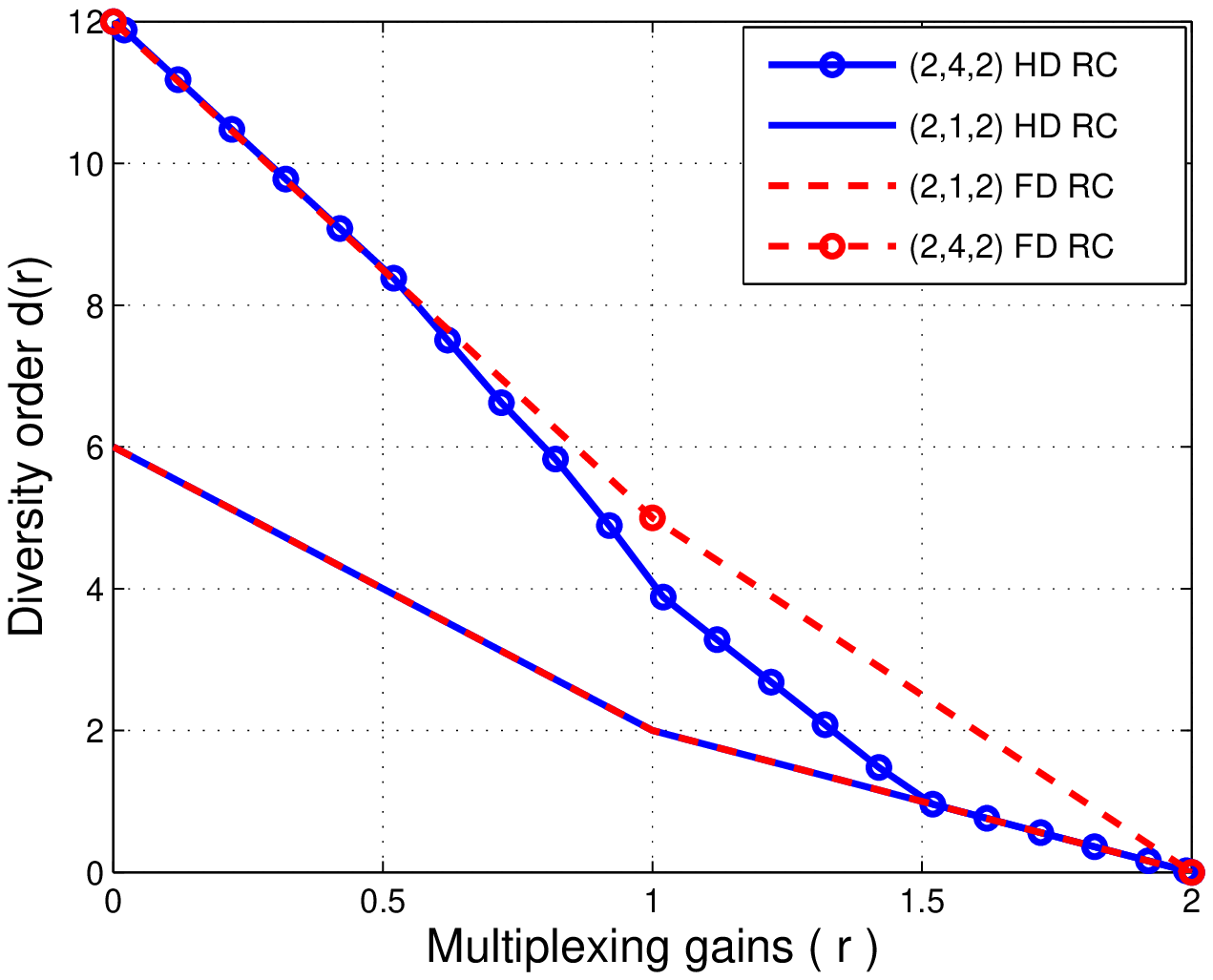}}
    \subfigure[Static Vs. Dynamic HD-RCs.]{\label{fig_dmt_comparison2-b}\includegraphics[scale=0.5]{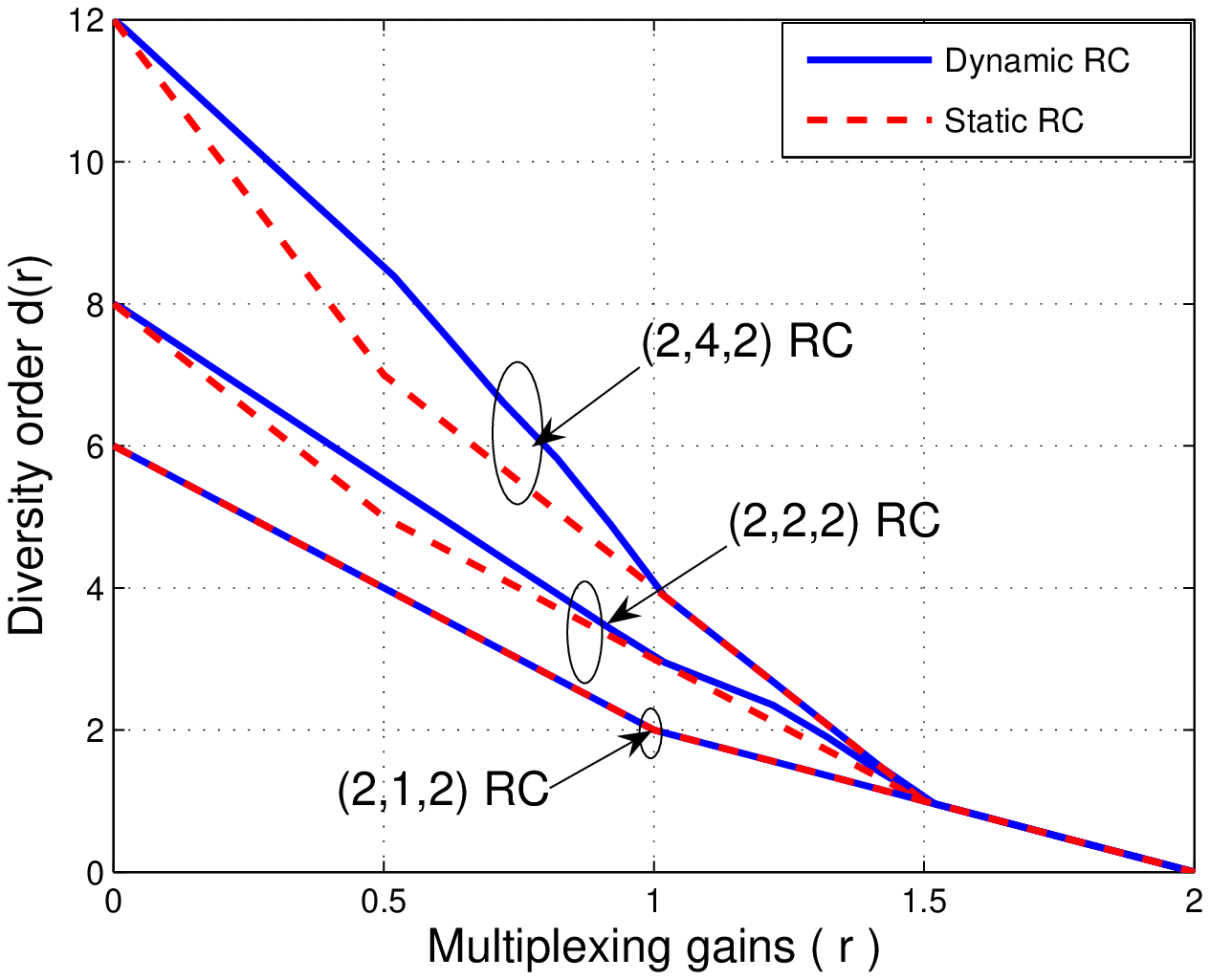}}
  \end{center}
  \caption{A comparison of DMTs on Dynamic vs. Static and HD- vs. FD-RCs.}
  \label{fig_dmt_comparison2}
\end{figure}

\begin{rem}
A similar argument holds for the comparison between the DMT performance of the static and dynamic HD-RCs. In contrast to a static channel, since on a dynamic channel the switching time varies depending on the instantaneous channel matrices, it is expected that a larger number of antennas at the relay node make a bigger difference on the DMT performance. 
\end{rem}

The following result gives an explicit formula for the DMT of the $(1,k,1)$ HD-RC.

\begin{thm}
\label{lem:DMT_1k1_channel}
The optimal DMT of a half-duplex $(1,k,1)$ HD-RC is given as
\begin{equation}
d^*_{(1,k,1)}(r)=\left\{\begin{array}{c}
(k+1)(1-r), ~0\leq r \leq \frac{1}{(k+1)};\\
1+ k \left(\frac{1-2r}{1-r}\right),~\frac{1}{(k+1)}\leq r \leq \frac{1}{2};\\
2(1-r),~\frac{1}{2}\leq r \leq 1.
\end{array}\right.
\label{eq:dmt1k1rc}
\end{equation}
\end{thm}

\begin{proof}[Proof]
The proof is given in Appendix~\ref{pf:lem:DMT_1k1_channel}.
\end{proof}

\begin{rem}
A comparison with the DMT of the class of static $(1,k,1)$ HD-RCs (derived in~\cite{OCM}), numerical examples of which are given in Fig. \ref{fig_dmt_1k1_channel-a} and Fig. \ref{fig_dmt_1k1_channel-b}, reveals that the DMT of such static HD-RCs are strictly smaller than that for their
dynamic counterparts for $ r \leq \frac{1}{2}$.
\end{rem}

The next result gives an explicit DMT formula for the class of symmetric HD-RCs with single-antenna relays.
\begin{thm}
\label{lem:DMT_n1n_channel}
The DMT of the $(n,1,n)$ HD-RC is given by a piece-wise linear curve whose corner points at integer values of $r$ are given as
\begin{eqnarray}
\label{eq_dmt_n1n_channel}
d^*_{(n,1,n)}(r) &= & d_{(n+1,n)}(r) \nonumber \\
&=& (n-r)(n+1-r), ~~0\leq r\leq n.
\end{eqnarray}
\end{thm}


\begin{proof}[Proof]
The proof is given in Appendix~\ref{pf:lem:DMT_n1n_channel}.
\end{proof}

Before proceeding to the next section, we summarize the findings of the previous two sections. The explicit computation of the DMT of the MIMO HD-RC enabled the proof of the existence of classes of HD-RCs whose DMT performances are (a) strictly inferior to that of the corresponding FD-RCs (with certain $m,k,n$) and (b) equal to that of the corresponding FD-RCs (when $m>n\geq k$, see Remark \ref{rem:fd=hd}). Furthermore, closed-form solutions for the two-variable optimization problem are obtained for two classes of symmetric HD-RCs, namely, the $(n,1,n)$ and the $(1,k,1)$ HD-RCs (see Theorems \ref{lem:DMT_1k1_channel} and \ref{lem:DMT_n1n_channel}). More generally, the explicit DMT is upper bounded for the symmetric $(n,k,n)$ HD-RCs and is conjectured to be tight (see Theorem \ref{thm_closedform_sol_symmetric_case} and Conjecture \ref{conj:nkn}). These solutions reveal certain special characteristics of the HD-RC which are different from that of the FD-RC and the static HD-RC. For example, while an extra antenna at the relay node uniformly improves the DMT for an FD-RC this is not the case for the HD-RC. Moreover, the greater the number of antennas at the relay, the greater is the difference between the DMT performances of the FD- and HD-RCs. It is also observed that within the class of HD-RCs, static operation of the relay in general limits performance relative to unconstrained (or dynamic) operation of the relay with the difference in DMT performance becoming more pronounced with an increasing number of antennas at the relay.

\section{Achievability of the DMT without switching time at the relay node}
\label{sec:only-CSIR}

In the previous section we established the fundamental DMT for the MIMO HD-RC. It hence represents the DMT achievable by the best cooperative protocol among all admissible ones with global CSI at both the relay and the destination. The QMF scheme of Section \ref{subsec:lower-bound} however achieves this fundamental tradeoff with just the relay having knowledge of the switching time $t_d^*$ defined in \eqref{eq:cut-set-optimal-switching-time} (with the destination node having global CSI). In this section, we explore the question of whether there are situations in which even the switching time knowledge at the relay is not necessary. Note that the DCF protocol was shown to be optimal from the DMT perspective in \cite{YEE} so that the DMT of the previous section is achievable by this protocol too but it requires global CSI at the relay. In the case of the SISO HD-RC, the QMF scheme of \cite{Avestimehr_Diggavi_Tse} was shown in \cite{PAT} to achieve, with switching time of $1/2$ (and hence without knowledge of switching time at relay), the fundamental DMT of the SISO FD-RC which in turn is an upper bound for the SISO HD-RC, so that dynamic operation of the relay in this case does not help from the DMT perspective. Does this result generalize to MIMO HD-RCs? It turns out that it does in some cases. In particular, we show that for the $(n,1,n)$ MIMO HD-RCs the DMT (given by Theorem \ref{lem:DMT_n1n_channel}) can indeed be achieved by the QMF protocol with a channel independent switching time. In other words, on this class of RCs even global CSI at the relay node does not help in terms of DMT performance, generalizing the SISO HD-RC result of \cite{PAT}. Moreover, for the class of $(1,k,1)$ HD-RCs we show that for all multiplexing gains $r \in [0,1/2]$ the optimal tradeoff curve can be achieved by the dynamic decode-and-forward (DDF) protocol analyzed for the general MIMO HD-RC by the authors in \cite{SKMV_DDF_journal} and for all multiplexing gains $r \in [1/2,1]$ it is achieved by the static QMF protocol and neither of these protocols requires knowledge of $H_{SD}$ and $H_{RD}$ at the relay node.

We shall show that while on the two specific classes of RCs as specified earlier the optimal diversity order at all multiplexing gains can be achieved without any CSI at the relay node, the same is also true in general but only for higher multiplexing gain values (e.g., see Fig. \ref{fig_dmt_comparison2-b}).

\begin{rem}
\label{thm:optimality_of_SCF}
The DMT of the static MIMO HD-RC was obtained in \cite{OCM} with the static CF protocol as the achievability scheme which requires global CSI at the relay. From Theorem~8.5 in \cite{Avestimehr_Diggavi_Tse} however, we have that on the static HD-RC, the QMF protocol, which doesn't require any CSI at the relay node (cf. Section~VIII-A of \cite{Avestimehr_Diggavi_Tse}), can achieve the instantaneous capacity within a constant gap. Since a constant number of bits is insignificant in DMT metric, the QMF protocol can hence achieve the DMT of the static HD-RC.
\end{rem}

\begin{rem}
\label{cor:dmt-when-CSIR-is-sufficient}
If the DMT of the static and dynamic HD-RCs are identical over some range of multiplexing gains then the optimal diversity orders at those values of the multiplexing gains can be achieved without any CSI at the relay node.
\end{rem}

\begin{ex}
\label{ex:dynamic-static-eq-at-high-MG}
From Fig. \ref{fig_dmt_comparison2-b} it is clear that the optimal diversity order achievable on the $(2,2,2)$ and $(2,4,2)$ HD-RCs can be achieved by the QMF protocol without any CSI at the relay for $r\in [1.5,~ 2]$ and $r\in [1,~ 2]$, respectively. For all other values of $r$, the relay node requires the optimal switching time (or $t_d^*$) to achieve the maximum diversity orders achievable by a dynamic protocol.
\end{ex}

We turn our attention now to the MIMO FD-RC. The optimality of the CF protocol in the DMT metric was proved in \cite{YEE} (and the DMT itself was found to be given by \eqref{dmt-fd-rc}) but this protocol requires global CSI as discussed earlier. 

\begin{rem}
\label{lem:optimality_of_FD_CF}
It was shown in \cite{Avestimehr_Diggavi_Tse} that the (full-duplex version of the) QMF protocol can achieve the instantaneous capacity of the FD-RC to within a constant number of bits. Hence, the DMT of the MIMO FD-RC can be achieved by the QMF protocol without any CSI at the relay node and global CSI at the destination node.
\end{rem}

As an immediate application of Remark \ref{thm:optimality_of_SCF} and Theorem \ref{lem:DMT_1k1_channel} we get the following result.

\begin{cor}
\label{thm:No_CSIT_1k1_channel}
The fundamental DMT of the $(1,k,1)$ relay channel can be achieved by the DDF protocol, for multiplexing gains in $\left[0, \frac{1}{2}\right]$ and by the static
QMF protocol for multiplexing gains in the interval $\left[\frac{1}{2}, 1\right]$. While the DDF protocol requires only CSIR, the QMF protocol does not require any CSI at the relay node, i.e., neither the DDF nor the QMF protocol requires global CSI at the relay node.
\end{cor}

\begin{proof}[Proof]
Theorem~\ref{lem:DMT_1k1_channel} provides the optimal DMT on a $(1,k,1)$ relay channel. Comparing it with the
DMT of the DDF protocol on this channel derived in \cite{SKMV_DDF_journal} which is restated here for convenience, namely,
\begin{equation*}
d_{(1,k,1)}^{\textrm{DDF}}(r) = \left\{\begin{array}{c}
(k+1)(1-r),  ~0\leq r \leq \frac{1}{(k+1)}; \\
1+ k \left(\frac{1-2r}{1-r}\right),~\frac{1}{(k+1)}\leq r \leq \frac{1}{2};\\
\left(\frac{(1-r)}{r}\right),  ~\frac{1}{2}\leq r\leq 1,
\end{array}\right.
\end{equation*}
it is evident that the fundamental DMT of the $(1,k,1)$ HD-RC can be achieved by the DDF protocol for $0\leq r\leq \frac{1}{2}$. Further this DMT can be achieved by the DDF protocol with only the knowledge of $H_{SR}$ at the relay node.

On the other hand, it was proved in \cite{OCM} that the DMT of the static $(1,k,1)$ HD-RC is $2(1-r)$ for $\frac{1}{2}\leq r\leq 1$. For $\frac{1}{2}\leq r\leq 1$, this DMT can be achieved by the QMF protocol without any CSI at the relay by Remark~\ref{cor:dmt-when-CSIR-is-sufficient}.
\end{proof}

The key enabling result for Corollary \ref{thm:No_CSIT_1k1_channel} beyond the DMTs of the DDF and the static HD-RC is the explicit DMT of the $(1,k,1)$ HD-RC of Theorem \ref{lem:DMT_1k1_channel}. Moreover, to the best of our knowledge, Corollary \ref{thm:No_CSIT_1k1_channel} is the first result on the achievability of the DMT of a non-SISO HD-RC without global CSI at the relay node. This result however requires two different protocols for the two ranges of multiplexing gains. In this sense, the above result doesn't truly generalize the result of \cite{PAT} in which it is shown that the QMF protocol achieves the DMT of the SISO HD-RC.

\begin{ex}
Comparing the DMT of the DDF protocol with the fundamental DMT of the $(1,2,1)$ relay channel depicted in Fig. \ref{fig_dmt_1k1_channel-b}, we see that the DDF protocol is DMT optimal on this channel for a multiplexing gain in the range $[0, \frac{1}{2}]$. Moreover, since the static DMT is strictly smaller than that of the corresponding dynamic channel in this range, the CF and QMF protocols require global CSI and the optimal switching time information at the relay node, respectively, to achieve optimal DMT performance. However, the DDF protocol needs only source-to-relay CSI at the relay node. Clearly, the cooperative protocol of choice in this case (i.e., with $ r \in [0, \frac{1}{2}]$ ) is the DDF protocol.
\end{ex}


In what follows, we identify a class of non-SISO RCs, namely the $(n,1,n)$ HD-RCs, on which the DMT of the channel can be achieved by a single protocol, namely the static QMF protocol with no CSI at the relay node, thereby generalizing the result of \cite{PAT}. This result is shown by proving that the DMTs of the static and dynamic $(n,1,n)$ HD-RCs are identical. In other words, for this class of HD-RCs dynamic operation of the relay does not help from the DMT perspective. We start by first finding in closed form the DMT of the static $(n,1,n)$ HD-RC.

\begin{thm}
\label{lem:dmt-static-n1n}
The DMT of the static and dynamic $(n,1,n)$ HD-RCs are identical, i.e., is given by
\begin{equation}
\label{eq:dmt-static-n1n}
d^{stat}_{(n,1,n)}(r)=d^{*}_{(n,1,n)}(r) = d^{ptp}_{(n+1),n}(r), ~~0\leq r\leq n.
\end{equation}
Hence, the DMT of the $(n,1,n)$ HD-RC can be achieved by the static QMF protocol with no CSI at the relay node (i.e., without the knowledge of even the optimal switching time).
\end{thm}

\begin{proof}[Proof]
The second equality in \eqref{eq:dmt-static-n1n} is just the result of Theorem \ref{lem:DMT_n1n_channel}. If the first equality holds, it means that the static QMF protocol (which, by Remark~\ref{thm:optimality_of_SCF}, achieves the DMT of the static HD-RC without any CSI at the relay node) achieves DMT of the dynamic $(n,1,n)$ HD-RC without any CSI at the relay node. It remains to prove the first equality.

The DMT of the symmetric $(n,k,n)$ static HD-RC was established as the solution of a convex optimization problem in \cite{OCM} and an analytic expression for an upper bound to the DMT was provided therein. The proof of the above lemma is given by obtaining an exact closed form solution to that optimization problem for the case of $k=1$. Our starting point is thus equation (13) in \cite{OCM} which is restated here for convenience,
\begin{IEEEeqnarray}{rl}
\label{eq:dmtstatn1n}
d_{(n,1,n)}^{stat}(r)=\min_{\{(\bar{\alpha},\beta_1)\in \mathcal{T}\}} \sum_{i=1}^n(2n-2i+1)\alpha_i & +n\beta_1-\sum_{i=1}^n(1-\alpha_i)^++\sum_{i=1}^{n-1}(1-\beta_1-\alpha_i)^+
\end{IEEEeqnarray}
where
\begin{IEEEeqnarray}{rl}
\mathcal{T}=\Big\{(\bar{\alpha},\beta_1): \sum_{i=1}^n(1-\alpha_i)^++\frac{1}{2}(1-\beta_1)^+ \leq r;
0\leq \alpha_1\leq \alpha_2\leq,\cdots \leq \alpha_n ;
0\leq \beta_1; (\beta_1+\alpha_n)\geq 1\Big\}.
\end{IEEEeqnarray}
Using an argument similar to that in the proof of Theorem \ref{thm_optimization_problem} in Appendix \ref{app:thm_optimization_problem}, it can be shown that a further restriction to $\alpha_n,\beta_1\in [0,1]$ can be made without changing the solution of \eqref{eq:dmtstatn1n} but it greatly simplifies the problem as
\begin{IEEEeqnarray}{rl}
\label{eq:objective-function-n1n-lower-bound-0}
d_{(n,1,n)}^{stat}(r)=\min_{\{(\bar{\alpha},\beta_1)\in \hat{\mathcal{T}}\}} \sum_{i=1}^n(2n-2i+2)\alpha_i  +n\beta_1-n+\sum_{i=1}^{n-1}(1-\beta_1-\alpha_i)^+.
\end{IEEEeqnarray}
where
\begin{IEEEeqnarray}{rl}
\label{eq:constraint-n1n-lower-bound-1}
\hat{\mathcal{T}}=\Big\{(\bar{\alpha},\beta_1): \sum_{i=1}^n(1-\alpha_i)+\frac{1}{2}(1-\beta_1) \leq r;
0\leq \alpha_1\leq \alpha_2\leq,\cdots \leq \alpha_n \leq 1;
0\leq \beta_1\leq 1 ; (\beta_1+\alpha_n)\geq 1\Big\}.
\end{IEEEeqnarray}
The rest of the proof that the solution of \eqref{eq:objective-function-n1n-lower-bound-0} is given by the lemma follows from induction, the details of which are relegated to Appendix~\ref{pf:lem:dmt-static-n1n}.
\end{proof}

\begin{rem}
\label{rem:n1n-dmt-FD-HD-equal}
Note that the DMT of the FD $(n,1,n)$ RC is also given be $d_{(n+1),n}^{ptp}(r)$~\cite{YEE}. Therefore, on the $(n,1,n)$ HD-RC, the DMT of the $(n,1,n)$ FD-RC can be achieved by an HD relay without any CSI at the relay node.
\end{rem}

\begin{rem}
Although, for a large number of $n\in \mathbb{N}$ the DMT of the static $(n,1,n)$ RC was computed in \cite{OCM} and was observed to be equal to $d_{(n+1),n}^{ptp}(r)$, the analysis of \cite{OCM} only proves that $d_{(n+1),n}^{ptp}(r)$ represents an upper bound to the DMT of the static $(n,1,n)$ RC. Therefore, the conclusion of Theorem~\ref{lem:dmt-static-n1n} cannot be obtained from the result of \cite{OCM}.
\end{rem}

%
%

\section{Conclusion}
\label{sec:conclusion}

The fundamental DMT of the three-terminal $(m,k,n)$ HD-RC is characterized. This allows an in-depth comparison of half-duplex and full duplex relaying as well as dynamic and static operation of the relay as a function of the numbers of antennas at the three nodes. Unlike in the single-antenna relay channel, half-duplex relaying in general results in a penalty relative to a full-duplex relaying and an improved performance relative to static half-duplex relaying at high SNR performance as measured by the DMT metric. The achievability of the fundamental DMT is shown via the dynamic QMF protocol~\cite{Avestimehr_Diggavi_Tse} which requires only the knowledge of the optimal switching time at the relay. Classes of HD-RCs for which dynamic operation of the relay doesn't improve performance over that of static relaying are identified. For such RCs the knowledge of switching time is not needed either. The problem of characterizing the DMT of the relay channel with multiple relays is one for future research as is the problem of finding finite block-length coding schemes that are DMT optimal.


\appendices
\section{Proof of Lemma~\ref{lem_conditional_independence}}
\label{pf_lem_conditional_independence}
Let the singular value decomposition of $H_{SD}$ be given as $U\Lambda_0 V^\dagger$, where $U\in \mathbb{C}^{n\times n}$ and $V\in \mathbb{C}^{m\times m}$ are mutually independent unitary random matrices distributed uniformly over the set of square unitary matrices of corresponding dimensions (e.g., see equation (3.9) in \cite{Edelman_Rao}). Using this fact we can write
\begin{IEEEeqnarray}{rl}
H_{SD}H_{SD}^\dagger=U\Lambda U^\dagger,~ \textrm{and}~ H_{SD}^\dagger H_{SD}=V\Lambda_2 V^\dagger,
\end{IEEEeqnarray}
where the sets of non-zero elements of $\Lambda_2$ and $\Lambda$ are identical. In particular, given one, the other is fixed. Putting this in the expressions for $W_2$ and $W_3$ we get
\begin{IEEEeqnarray}{rl}
W_2=H_{SR}V\left(I_m+\Lambda_2\right)^{-1}V^\dagger H_{SR}^\dagger = \widetilde{H}_{SR}\left(I_m+\Lambda_2\right)^{-1} \widetilde{H}_{SR}^\dagger,\\
W_3=H_{RD}^\dagger U\left(I_n+\Lambda\right)^{-1}U^\dagger H_{RD} = \widetilde{H}_{RD}^\dagger \left(I_n+\Lambda\right)^{-1} \widetilde{H}_{RD} ,
\end{IEEEeqnarray}
where $\widetilde{H}_{SR}=H_{SR}V$ and $\widetilde{H}_{RD}=U^\dagger H_{RD}$ are mutually independent random matrices that have the same distributions as $H_{SR}$ and $H_{RD}$, respectively, since $H_{SR}$ and $H_{RD}$ are unitarily invariant (cf. \cite{Tulino_Verdu}). Letting $A=\widetilde{H}_{SR}\left(I_m+\Lambda_2\right)^{-\frac{1}{2}}$ and $B=\widetilde{H}_{RD}^\dagger \left(I_n+\Lambda\right)^{-\frac{1}{2}}$ we realize that both $A$ and $B$ still have mutually independent Gaussian entries conditioned on $\Lambda $. Computing the conditional correlation between the two we get
\begin{IEEEeqnarray*}{rl}
    \mathbb{E}\left(B^\dagger A|\Lambda_2,\Lambda\right)= & \mathbb{E}\left( \left(I_n+\Lambda\right)^{-\frac{1}{2}}\widetilde{H}_{RD} \widetilde{H}_{SR}\left(I_m+\Lambda_2\right)^{-\frac{1}{2}}|\Lambda_2,\Lambda\right),\\
    =&  \mathbb{E}\left(\left(I_n+\Lambda\right)^{-\frac{1}{2}} U^\dagger {H}_{RD} {H}_{SR}V\left(I_m+\Lambda_2\right)^{-\frac{1}{2}}|\Lambda_2,\Lambda\right)
     =0_{k\times k},
\end{IEEEeqnarray*}
where the last step follows from the fact that components of $H_{SR}$ and $H_{RD}$ are zero mean and mutually independent. This, along with the fact that $A$ and $B$ are Gaussian~\cite{Bilodeau_Brenner} proves that conditioned on $\Lambda_2$ or $\Lambda$, they are independent. This in turn implies that $W_2 = A A^\dagger $ and $W_3 = B B^\dagger $ are independent given $\Lambda$. Consequently, the eigenvalues of $W_2$ are independent of the eigenvalues of $W_3$ given $\Lambda$.

\section{Proof of Lemma~\ref{lem:cut-set-upper-bound}}
\label{app:lem:cut-set-upper-bound}

\subsection{Proof of Part i}
The proof consists of upper bounding  in two steps the mutual information terms of the form $I(X;X+Z)$, subject to a sum power constraint on $X$. First, we use the fact Gaussian input is optimal, and then in the second step, we use the monotonically increasing property of the $\log\det(.)$ function in the cone of positive semi-definite matrices.

Suppose $Y=HX+Z$, where $Z\sim \mathcal{CN}(0,I)$, $H\in \mathbb{C}^{N\times M}$ and $\textrm{Cov}(X)\preceq Q$, then it is well known~\cite{CT} that
\begin{equation}
\label{eq:G-optimal}
    \max_{\{\textrm{Cov}(X)\preceq Q\}}I(HX+Z;X)=I(HX_G+Z;X_G)=\log\det\left(I+HQH^\dagger\right),
\end{equation}
where $X_G\sim \mathcal{CN}(0,Q)$. Similarly, for a sum power constraint we have
\begin{IEEEeqnarray}{rl}
\max_{\{\textrm{Tr}\left(\textrm{Cov}(X)\right)\preceq \rho\}}I(HX+Z;X)   =&\max_{\textrm{Tr}(Q)\leq \rho} \max_{\{\textrm{Cov}(X)\preceq Q\}}I(HX+Z;X) \nonumber \\
\label{eq:pf-lem-cut-set-ubound-a}
\stackrel{(a)}{=}&\max_{\textrm{Tr}(Q)\leq \rho} \log\det\left(I+HQH^\dagger\right),\\
\label{eq:pf-lem-cut-set-ubound-b}
{\leq }& \log\det\left(I+\rho HH^\dagger\right),
\end{IEEEeqnarray}
where step $(a)$ follows from \eqref{eq:G-optimal} and the last step follows from the fact that $Q\preceq \rho I$ and $\log\det(.)$ is a monotonically increasing function in the cone of semi-definite matrices.

Using equation \eqref{eq:pf-lem-cut-set-ubound-b} we have
\begin{IEEEeqnarray*}{rl}
\max_{P(X_S,X_R)} I(X_S,X_R;Y_D|p_2)=& ~ \max_{P(X_S,X_R)}I\left(H_{SR,D}\left[\begin{array}{c}X_S\\ X_R\end{array}\right]+Z_D;\left[\begin{array}{c}X_S\\ X_R\end{array}\right]\right) \\
\leq & ~\log \det\left(I_{n}+\rho H_{SR,D}H_{SR,D}^\dagger\right)=\log(L_{SR,D}),
\end{IEEEeqnarray*}
where $H_{SR,D}=[H_{SD}~H_{RD}]$. Using a similar method we obtain
\begin{IEEEeqnarray*}{rl}
\max_{P(X_S,X_R)} I(X_S;Y_D|X_R,p_2) \leq & ~ \log\det\left(I_n+\rho H_{SD} H_{SD}^\dagger\right)=\log(L_{SD});\\
\max_{P(X_S,X_R)} I(X_S;Y_D|p_1) \leq & ~ \log \det\left(I_n+\rho H_{SD} H_{SD}^\dagger\right)=\log(L_{SD});\\
\max_{P(X_S,X_R)} I(X_S;Y_R,Y_D|p_1) \leq & ~ \log \det\left(I_n+\rho H_{S,RD}H_{S,RD}^\dagger\right)=\log(L_{S,RD}).
\end{IEEEeqnarray*}
Finally, substituting the above set of upper bounds in equation \eqref{eq:cut-set-bound-around-source} and \eqref{eq:cut-set-bound-around-destination} we get
\begin{IEEEeqnarray*}{rl}
\max_{P(X_S,X_R)} I_{C_S}(t_d) \leq & t_d \log\left(L_{S,RD}\right)+(1-t_d) \log\left(L_{SD}\right)\triangleq I^{'}_{C_S}(t_d), \\
\max_{P(X_S,X_R)} I_{C_D}(t_d)\leq & t_d \log\left(L_{SD}\right)+(1-t_d) \log\left(L_{SR,D}\right)\triangleq I^{'}_{C_D}(t_d).
\end{IEEEeqnarray*}
This proves the first part of the lemma.

\subsection{Proof of Part ii}

Let $P^*$ represent the distribution where $X_S\sim \mathcal{CN}(0,\frac{\rho}{m}I_m)$ and $X_R\sim \mathcal{CN}(0,\frac{\rho}{k}I_k)$ and $X_S$ and $X_R$ are mutually independent. Note that $P^*$ satisfies the input power constraints at the source and relay given in \eqref{eq:power-constraint-at-source}  and \eqref{eq:power-constraint-at-relay}. Then denoting the mutual information $I(X_S,X_R;Y_D|p_2)$ evaluated at $P^*$ by $I(X_S,X_R;Y_D|p_2)\Big\rvert_{ P^*}$ we see
\begin{IEEEeqnarray}{rl}
I(X_S,X_R;Y_D|p_2)\Big\rvert_{ P^*}=& ~ I\left(H_{SR,D}\left[\begin{array}{c}X_S\\ X_R\end{array}\right]+Z_D;\left[\begin{array}{c}X_S\\ X_R\end{array}\right]\right)\Bigg\rvert_{ P^*}, \nonumber \\
= & ~\log \det\left(I_{n}+\frac{\rho}{m} H_{SD}H_{SD}^\dagger+\frac{\rho}{k} H_{RD}H_{RD}^\dagger\right),\nonumber \\
\stackrel{(a)}{\geq} & ~\log \det\left(\frac{I_n}{m+k}+\frac{\rho}{m+k} H_{SD}H_{SD}^\dagger+\frac{\rho}{m+k} H_{RD}H_{RD}^\dagger\right),\nonumber \\
\label{eq:MI-at-P*-a}
= & ~\log \det\left(I_{n}+\rho H_{SR,D}H_{SR,D}^\dagger\right)-(m+k),\nonumber \\
=& \log(L_{SR,D})-m
\end{IEEEeqnarray}
where step $(a)$ follows from the fact that $\log\det(.)$ is a monotonically increasing function in the cone of positive semi-definite matrices. Using a similar method we get
\begin{IEEEeqnarray}{rl}
\label{eq:MI-at-P*-b}
I(X_S;Y_D|X_R,p_2)\Big\rvert_{ P^*}\geq &\log\det\left(I_n+\rho H_{SD} H_{SD}^\dagger\right)-m =\log(L_{SD})-m ;\\
\label{eq:MI-at-P*-c}
I(X_S;Y_D|p_1)\Big\rvert_{ P^*}\geq & \log\det\left(I_n+\rho H_{SD} H_{SD}^\dagger\right)-m =\log(L_{SD})-m;\\
\label{eq:MI-at-P*-d}
I(X_S;Y_R,Y_D|p_1)\Big\rvert_{ P^*}\geq &\log \det\left(I_n+\rho H_{S,RD}H_{S,RD}^\dagger\right)-m =\log(L_{S,RD})-m .
\end{IEEEeqnarray}

Now, from the definition of $\bar{C}(\mathcal{H},t_d)$ in equation \eqref{eq:rate-upper-bound-intermediate1} we get
\begin{IEEEeqnarray*}{rl}
\bar{C}(\mathcal{H},t_d)= & \max_{\{P(X_S,X_R)\}} \min \{I_{C_S}(t_d),I_{C_D}(t_d)\},\\
\stackrel{(a)}{\geq } &\max_{\{P(X_S,X_R)=P^*\}} \min \left\{I_{C_S}(t_d),I_{C_D}(t_d)\right\},\\
=& \min \left\{I_{C_S}(t_d)\Big\rvert_{ P^*},I_{C_D}(t_d)\Big\rvert_{ P^*}\right\},\\
\stackrel{(b)}{\geq } &  \min \{I_{C_S}^{'}(t_d)-m, I_{C_D}^{'}(t_d)-t_d m-(1-t_d)(m+k)\},\\
\geq & \min \{I_{C_S}^{'}(t_d), I_{C_D}^{'}(t_d)\}-(m+k),
\end{IEEEeqnarray*}
where step $(a)$ follows from the fact that instead of maximizing over all possible input distributions, we are evaluating the right hand side of the equation at a particular distribution $P^*$ and in step $(b)$ we substituted the set of lower bounds from equations \eqref{eq:MI-at-P*-a}-\eqref{eq:MI-at-P*-d} in the expressions for $I_{C_S}(t_d)$ and $I_{C_D}(t_d)$.

\section{Proof of Lemma \ref{lem:d-in-terms-of-R*}}
\label{App:lem:d-in-terms-of-R*}
We will prove that
\begin{equation*}
    R_U^*=r^*(\bar{\alpha},\bar{\beta},\bar{\delta})\log(\rho),
\end{equation*}
which, when substituted in equation \eqref{eq:d-and-Ou-relation}, proves the lemma.
From the expression of $L_{S,RD}$ in equation \eqref{eq_expression_Lsrd1} and using elementary properties of determinants, we have
\begin{IEEEeqnarray*}{rl}
\log(L_{S,RD})=& ~ \log\det\left(I_{n+k}+\rho H_{S,RD}H_{S,RD}^{\dagger}\right),\\
=& ~ \log\det\left(I_m+\rho H_{S,RD}^\dagger H_{S,RD}\right)\stackrel{(a)}{=}\log\det\left(I_m+\rho H_{SR}^\dagger H_{SR}+\rho H_{SD}^\dagger H_{SD}\right),\\
=& ~ \log\det\left(I_m+\rho H_{SR}^\dagger H_{SR}\left(I_m+\rho H_{SD}^\dagger H_{SD}\right)^{-1}\right)+\log\det\left(I_m+\rho H_{SD}^\dagger H_{SD}\right),\\
=& ~ \log\det\left(I_k+\rho H_{SR}\left(I_m+\rho H_{SD}^\dagger H_{SD}\right)^{-1}H_{SR}^\dagger \right)+\log(L_{SD}).
\end{IEEEeqnarray*}
where in step $(a)$ we have used the fact that $H_{S,RD}^\dagger=[H_{SR}^\dagger~ H_{SD}^\dagger]$. Hence, we have
\begin{IEEEeqnarray*}{rl}
\log\left(\frac{L_{S,RD}}{L_{SD}}\right)=&\log\det\left(I_k+\rho H_{SR}\left(I_m+\rho H_{SD}^{\dagger}H_{SD}\right)^{-1}H_{SR}^{\dagger}\right).
\end{IEEEeqnarray*}
 Similarly, from the expression for $L_{SR,D}$ in equation \eqref{eq_expression_Lsr_d} we get
\begin{IEEEeqnarray}{l}
\log\left(\frac{L_{SR,D}}{L_{SD}}\right)=\log\det\left(I_k+\rho H_{RD}^{\dagger}\left(I_n+\rho H_{SD}H_{SD}^{\dagger}\right)^{-1}H_{SR}\right).
\end{IEEEeqnarray}
Now, assuming $0 < \lambda_u \leq \cdots \leq \lambda_1$, $0 < \mu_p \leq \cdots \leq \mu_1$ and $0 < \gamma_q \leq \cdots \leq \gamma_1$ represent the ordered non-zero (w.p. 1) eigenvalues of the matrices $W_1\triangleq H_{SD}H_{SD}^{\dagger}$, $W_2\triangleq H_{SR}(I_n+\rho H_{SD}^{\dagger}H_{SD})^{-1}H_{SR}^{\dagger}$ and $W_3\triangleq H_{RD}^{\dagger}(I_n+\rho H_{SD}H_{SD}^{\dagger})^{-1}H_{RD}$ and substituting these in equation \eqref{eq:rate-upper-bound} we get
\begin{equation}
R_U^* =  \frac{ \log\left(\prod_{j=1}^{p}  (1+\rho \mu_j)\right) \log\Big(\prod_{l=1}^{q}  (1+\rho \gamma_l)\Big) }{\log\left(\prod_{j=1}^{p}  (1+\rho \mu_j)\right) + \log\Big(\prod_{l=1}^{q}  (1+\rho \gamma_l)\Big)} +\log \left(\prod_{i=1}^{u}  (1+\rho \alpha_i)\right),
\end{equation}
To further simplify the expression on the right hand side of the above equation we use the following transformations: $\lambda_i=\rho^{-\alpha_i}, ~1\leq i\leq u$, $\mu_j=\rho^{-\beta_j}, ~1\leq j\leq p$ and $\gamma_l=\rho^{-\delta_l}, ~1\leq l\leq q$ and get
\begin{equation}
\label{eq_dmt_rate_upper_bound_3}
R_U^* =  ~\log (\rho)\left[  \frac{\sum_{l=1}^{q}(1-\delta_l)^+ \sum_{j=1}^{p}(1-\beta_j)^+}{\sum_{l=1}^{q}(1-\delta_l)^+ +\sum_{j=1}^{p}(1-\beta_j)^+} \right]+ \log (\rho)\left(\sum_{i=1}^{u}(1-\alpha_i)^+\right)  \triangleq r^*(\bar{\alpha},\bar{\beta},\bar{\delta})\log(\rho).
\end{equation}

\section{Proof of theorem~\ref{thm_optimization_problem}}
\label{app:thm_optimization_problem}
From equation \eqref{eq:d-and-Ou-relation} $d^*(r)$ is equal to the negative SNR exponent of $\Pr\{\mathcal{O}_U\}$. However, from Lemma~\ref{lem:d-in-terms-of-R*}, $\Pr\{\mathcal{O}_U\}$ is exponentially equal to $\Pr\{r^*(\bar{\alpha},\bar{\beta},\bar{\delta})<r\}$. Hence,
\begin{IEEEeqnarray}{rl}
\Pr\{\mathcal{O}_U\}\dot{=}& ~ \rho^{-d^*(r)} \\ \dot{=} & ~ \Pr\{r^*(\bar{\alpha},\bar{\beta},\bar{\delta})<r\} \\ =& ~ \int_{\{(\bar{\alpha},\bar{\beta},\bar{\delta})\in \mathcal{O}_U\}}f\left(\bar{\alpha},\bar{\beta},\bar{\delta}\right) ~d\bar{\alpha} ~d\bar{\beta} ~d\bar{\delta},\nonumber\\
=& ~ \int_{\{(\bar{\alpha},\bar{\beta},\bar{\delta})\in \mathcal{O}_U\cap \mathcal{S}\}}\rho^{-E\left(\bar{\alpha},\bar{\beta},\bar{\gamma}\right)} ~d\bar{\alpha} ~d\bar{\beta} ~d\bar{\delta},
\end{IEEEeqnarray}
where $\mathcal{O}_U=\{(\bar{\alpha},\bar{\beta},\bar{\delta}):~r^*(\bar{\alpha},\bar{\beta},\bar{\delta})<r\}$ and $f(.)$ is the joint pdf of $(\bar{\alpha},\bar{\beta},\bar{\delta})$.

Roughly, the above integral is a sum of an infinite number of terms of the form $\rho^{-E(\bar{\alpha},\bar{\beta},\bar{\delta})}$, one for each $(\bar{\alpha},\bar{\beta},\bar{\delta})$-tuple in $\mathcal{O}_U$. Laplace's method suggest that at the asymptotic SNR only the term having minimum negative exponent dominates, i.e.,
\begin{IEEEeqnarray}{rl}
\label{eq_pf_thm_opt_problem_of}
d^*(r)=\min_{\{(\bar{\alpha},\bar{\beta},\bar{\delta})\in \mathcal{S}\cap \mathcal{O}_U\}}E\left(\bar{\alpha},\bar{\beta},\bar{\delta}\right),
\end{IEEEeqnarray}
where $\mathcal{S}$ represents the support set of the joint pdf of $(\bar{\alpha},\bar{\beta},\bar{\delta})$ given in equation \eqref{eq_support_set_of_pdf} and the expression of $E(.)$ is given in equation \eqref{eq_distribution_abs}. Suppose at a given $r$, the objective function attains the minimum value for an $\bar{\alpha}\in \mathcal{S}$ where $\alpha_i>1$ for one or more $i$'s. Let $\widetilde{\bar{\alpha}}=\min \{[1, 1, \cdots ,1],\bar{\alpha}\}$, where the minimization is component-wise. Clearly, $\widetilde{\bar{\alpha}}\in \mathcal{S}$ but at this point $E(.)$ has a strictly smaller value. This proves that in the optimal solution, $\alpha_i\in [0, 1]$ for all $i$. The same is true for $\bar{\beta}$ and $\bar{\gamma}$. This, however, simplifies both the objective function and the constraint set giving the optimization problem in equations \eqref{general_opt}-\eqref{redun4} and \eqref{eq_objective_func}in the statement of Theorem \ref{thm_optimization_problem}.

\section{Proof of Theorem~\ref{thm_simplified_optimization}}
\label{pf_thm_simplified_opt}
The proof essentially contains two simplifying steps which consecutively simplify the optimization problem of Theorem~\ref{thm_optimization_problem}. In the first step a transformation of variables yields an equivalent problem having three variables independent of the values of $m,k$ and $n$. 
Analyzing the domain of definition of the equivalent problem we find that in the optimal solution one of the variables is a function of the other two resulting in an optimization problem having only two variables. We start with the first step.
\\
{\bf Step~1:}~The objective function in \eqref{eq_objective_func} decreases strictly monotonically as $\alpha_i$ is decreased for any $i$ and the rate of decrease with $\alpha_i$ is smaller for a larger value of $i$. The same is true for $\bar{\beta}$ and $\bar{\delta}$. Thus, following a similar method as in $\cite{tse1}$, it can be shown that if $\sum_{i=1}^{u}(1-\alpha_i)=a$, $\sum_{j=1}^{p}(1-\beta_j)=b$, $\sum_{l=1}^{q}(1-\delta_l)=s$ and $\left(\bar{\alpha},\bar{\beta},\bar{\delta}\right)$ satisfy equations \eqref{redun2}-\eqref{redun4}, then the optimal choice of $\left(\bar{\alpha},\bar{\beta},\bar{\delta}\right)$ that minimizes $F(.)$ is given by $(\phi_{\alpha}(a),\phi_{\beta}(b),\phi_{\delta}(s))$, where
\begin{IEEEeqnarray}{l}
\label{eq_phi_alpha}
\phi_{\alpha}(a)=[\hat{\alpha}_1, \hat{\alpha}_2, \cdots, \hat{\alpha}_u ]^T~ :~\hat{\alpha}_i=\left(1-\left(a-i+1\right)^+\right)^+, ~1\leq i\leq u,\\
\label{eq_phi_beta}
\phi_{\beta}(b)=[\hat{\beta}_1, \hat{\beta}_2, \cdots, \hat{\beta}_p ]^T~ :~\hat{\beta}_j=\left(1-\left(b-j+1\right)^+\right)^+, ~1\leq j\leq p,\\
\label{eq_phi_gamma}
\phi_{\delta}(y)=[\hat{\delta}_1, \hat{\delta}_2, \cdots, \hat{\delta}_t ]^T~ :~\hat{\delta}_l=\left(1-\left(s-l+1\right)^+\right)^+, ~1\leq l\leq q.
\end{IEEEeqnarray}

Denoting by $\mathcal{T}(a,b,s)$ the following set
\begin{IEEEeqnarray*}{rl}
\Big\{(\bar{\alpha}, \bar{\beta},\bar{\delta}):& \sum_{i=1}^{u}(1-\alpha_i)=a, \sum_{j=1}^{p}(1-\beta_j)=b, \sum_{l=1}^{q}(1-\delta_l)=s,~\left(\bar{\alpha},\bar{\beta},\bar{\delta}\right) ~\textrm{satisfy equations \eqref{redun2}-\eqref{redun4}}\Big\},
\end{IEEEeqnarray*}
we have from the above argument that
\begin{equation}
\label{eq_pf_opt_problem_first_min}
\min_{\{\mathcal{T}(a,b,s)\}} F\left((\bar{\alpha}, \bar{\beta},\bar{\delta})\right)=F\left(\phi_{\alpha}(a),\phi_{\beta}(b),\phi_{\delta}(s)\right).
\end{equation}
Let us now define the following set of new variables
\begin{IEEEeqnarray}{l}
\mathcal{O}_1= \left\{(a,b,s): a+\frac{bs}{b+s}\leq r,~ (a+b)\leq m,~ (a+s)\leq n,~ 0\leq a\leq u,~ 0\leq b\leq p,~0\leq s\leq q\right\}.
\end{IEEEeqnarray}
It is clear from the definition of $\mathcal{T}(a,b,s)$ that,
\begin{equation}
\hat{\mathcal{O}}_1 \triangleq \bigcup_{\{(a,b,s)\in \mathcal{O}_1\}}\mathcal{T}\left(a,b,s\right)\supset \hat{\mathcal{O}}.
\end{equation}
Since the minimum of a function over a set is not larger than the minimum of that function over a subset of it, the above relation along with equation~\eqref{eq_pf_opt_problem_first_min} imply that
\begin{IEEEeqnarray}{rl}
\label{eq_equivalent_opt_lb}
\min_{\{(a,b,s)\in \mathcal{O}_1\}} F\left(\phi_{\alpha}(a),\phi_{\beta}(b),\phi_{\delta}(s)\right)=& ~ \min_{\{(\bar{\alpha}, \bar{\beta},\bar{\delta})\in \hat{\mathcal{O}}_1\}} F\left(\bar{\alpha}, \bar{\beta},\bar{\delta}\right) \\
\leq & ~ \min_{\{(\bar{\alpha}, \bar{\beta},\bar{\delta})\in \hat{\mathcal{O}}\}} F\left(\bar{\alpha}, \bar{\beta},\bar{\delta}\right).
\end{IEEEeqnarray}
Before proceeding further we take note of a few properties of the newly defined variables $a, b$ and $s$. From the definition of $\phi_i$'s it is clear that if $(a,b,s)\in \mathcal{O}_1$, then $\left(\phi_{\alpha}(a),\phi_{\beta}(b),\phi_{\delta}(s)\right)$ satisfy equations \eqref{mi_constraint} and \eqref{redun2}-\eqref{redun4}. Suppose for some $(i+j)=(m+1)$, $(\hat{\alpha}_i+\hat{\beta_j})<1$, then it can be shown that $\sum_{i=1}^{u}(1-\hat{\alpha}_i)+\sum_{j=1}^{p}(1-\hat{\beta}_j)> m$, which is impossible. Thus $(\hat{\alpha}_i+\hat{\beta_j})\geq 1$ for all $(i+j)\geq (m+1)$. Similarly, it can be shown that $(\hat{\alpha}_i+\hat{\delta_l})\geq 1$ for all $(i+l)\geq (n+1)$, which in turn imply that the $\left(\phi_{\alpha}(a),\phi_{\beta}(b),\phi_{\delta}(s)\right)$ tuple also satisfies equations \eqref{redun1} and \eqref{redun12}. That is $ (a,b,s)\in \mathcal{O}_1 ~\Rightarrow~ (\phi_{\alpha}(a),\phi_{\beta}(b),\phi_{{\delta}}(s))\in \hat{\mathcal{O}} $ which implies that
\begin{IEEEeqnarray*}{l}
\label{eq_equivalent_opt_ub}
\min_{\{(\bar{\alpha}, \bar{\beta},\bar{\delta})\in \hat{\mathcal{O}}\}} F\left(\bar{\alpha}, \bar{\beta},\bar{\delta}\right) \leq \min_{\{(a,b,s)\in \mathcal{O}_1\}} F\left(\phi_{{\alpha}}(a),\phi_{{\beta}}(b),\phi_{{\delta}}(s)\right).
\end{IEEEeqnarray*}
Combining this with equation \eqref{eq_equivalent_opt_lb}, we get
\begin{equation}
\label{eq_equivalent_opt_1}
\min_{\{(\bar{\alpha}, \bar{\beta},\bar{\delta})\in \hat{\mathcal{O}}\}} F\left(\bar{\alpha}, \bar{\beta},\bar{\delta}\right) ~=~ \min_{\{(a,b,s)\in \mathcal{O}_1\}} F\left(\phi_{\alpha}(a),\phi_{\beta}(b),\phi_{\delta}(s)\right).
\end{equation}
Therefore, we have an equivalent optimization problem to that presented in Theorem~\ref{thm_optimization_problem}, but with fewer variables, i.e., $d^*(r)$ can be equivalently written as
\begin{equation}
d^*(r)=\min_{\left\{(a,b,s)\in \mathcal{O}_1 \right\}}~ F\left(\phi_{{\alpha}}(a),\phi_{{\beta}}(b),\phi_{{\delta}}(s)\right).
\end{equation}

When $(a+b)=m$ or $(a+s)=n$, the objective function has a property which we state now that will be helpful to solve the minimization problem in the next section.
\begin{cl}
\label{cl_a_plus_b_equal_1}
The fundctions $F\left(\phi_{\alpha}(a),\phi_{\beta}(m-a),\phi_{\delta}(s)\right)$ and $F\left(\phi_{{\alpha}}(a),\phi_{{\beta}}(b),\phi_{{\delta}}(n-a)\right)$ are monotonically decreasing with $a$ for a given $s$ and $b$, respectively, whereas $F\left(\phi_{{\alpha}}(a),\phi_{{\beta}}(m-a),\phi_{{\delta}}(n-a)\right)$ is monotonically decreasing with $a$.
\end{cl}
\begin{proof}
It can be shown using equations \eqref{eq_phi_alpha}-\eqref{eq_phi_gamma} that when $(a+b)=m$ we have
\begin{equation*}
\label{eq_iplusj_equalto1}
    (\hat{\alpha}_i+\hat{\beta}_j)=1, ~\forall (i+j)=(m+1)~\textrm{and}~
    (\hat{\alpha}_i+\hat{\beta}_j)\leq 1, ~\forall (i+j)\leq m.
\end{equation*}
Using these relations in the expression for $F\left(\bar{\hat{\alpha}},\bar{\hat{\beta}},\bar{\hat{\delta}}\right)$, after some algebra we get
\begin{IEEEeqnarray*}{rl}
\label{eq_ab_m}
F\left(\bar{\hat{\alpha}},\bar{\hat{\beta}},\bar{\hat{\delta}}\right)=& \sum_{i=1}^{p}(m+n+k+1-2i)\hat{\alpha}_i+\sum_{l=1}^{q}(n+k+1-2l)\hat{\delta}_l-ku+\sum_{\substack{i,l=1\\l+i\leq n}}^{u,q}(1-\hat{\alpha}_i-\hat{\delta}_l)^+.\\
\end{IEEEeqnarray*}
The above function is a strictly monotonically increasing function of $\hat{\alpha}_i$ for each $1\leq i\leq u$. Each of the $\hat{\alpha}_i$'s in turn is a monotonically decreasing function of $a$ which makes the above function a monotonically decreasing function of $a$.

Similarly, it can be shown that $F\left(\phi_{{\alpha}}(a),\phi_{{\beta}}(b),\phi_{{\delta}}(n-a)\right)$ is a monotonically decreasing function of $a$. However, when both $(a+b)=m$ and $(a+s)=n$, then we have
\begin{IEEEeqnarray*}{rl}
    (\hat{\alpha}_i+\hat{\beta}_j)=1,& ~ \forall ~ (i+j)=(m+1)~\textrm{and}~
    (\hat{\alpha}_i+\hat{\beta}_j)\leq 1, ~\forall ~ (i+j)\leq m;\nonumber\\
        (\hat{\alpha}_i+\hat{\delta}_l)=1,& ~ \forall ~ (i+l)=(n+1)~\textrm{and}~
    (\hat{\alpha}_i+\hat{\delta}_l)\leq 1, ~\forall ~ (i+l)\leq n.
\end{IEEEeqnarray*}
Using this in the expression for $F\left(\bar{\hat{\alpha}},\bar{\hat{\beta}},\bar{\hat{\delta}}\right)$ we get
\begin{IEEEeqnarray*}{rl}
\label{eq_ab_and_as_mn}
F\left(\phi_{{\alpha}}(a),\phi_{{\beta}}(m-a),\phi_{{\delta}}(n-a)\right)=F\left(\bar{\hat{\alpha}},\bar{\hat{\beta}},\bar{\hat{\delta}}\right)=& \sum_{i=1}^{p}(m+n+1-2i)\hat{\alpha}_i,
\end{IEEEeqnarray*}
which by a similar argument as above is a monotonically decreasing function of $a$.
\end{proof}

{\bf Step 2:} In this final step, we determine the minimum of $F(.)$ on $\mathcal{O}_1$ and establish the theorem. Note that if $a, b, s\in \mathcal{O}_1$, then $b\leq \min\{(m-a),p\}$ and $s\leq \min\{(n-a),q\}$. Let us denote these maximum values of $b$ and $s$ by $b_m$ and $s_m$, respectively. Depending on the value of $a$ the set of feasible $(b,s)$ pairs takes on different shapes as shown in the following figures. For example, when $a\in \mathcal{R}_1=\{a:\frac{b_m(r-a)}{(b_m-r+a)}\leq s_m \}$ the feasible set of $(b,s)$ pairs is the region ABCDE shown in Fig. \ref{a_range-a}.

\begin{figure}[htp]
  \begin{center}
    \subfigure[$\mathcal{R}=\left\{a: \frac{b_m(r-a)}{(b_m-r+a)}\leq s_m\right\}$]{\label{a_range-a}\includegraphics[scale=.3]{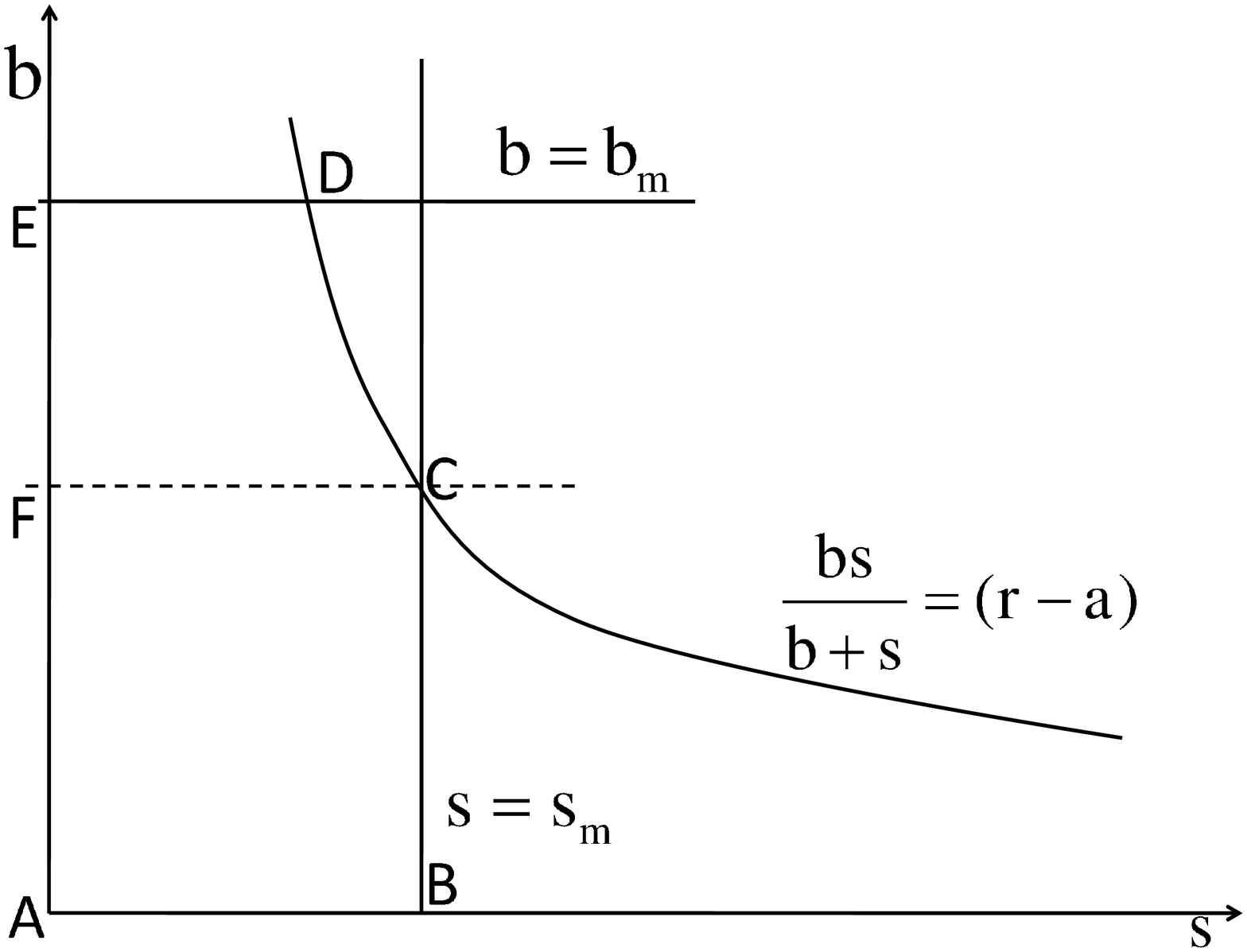}}
    \subfigure[$\mathcal{R}^c=\left\{a: s_m < \frac{b_m(r-a)}{(b_m-r+a)}\right\}$] {\label{a_range-b}\includegraphics[scale=0.3]{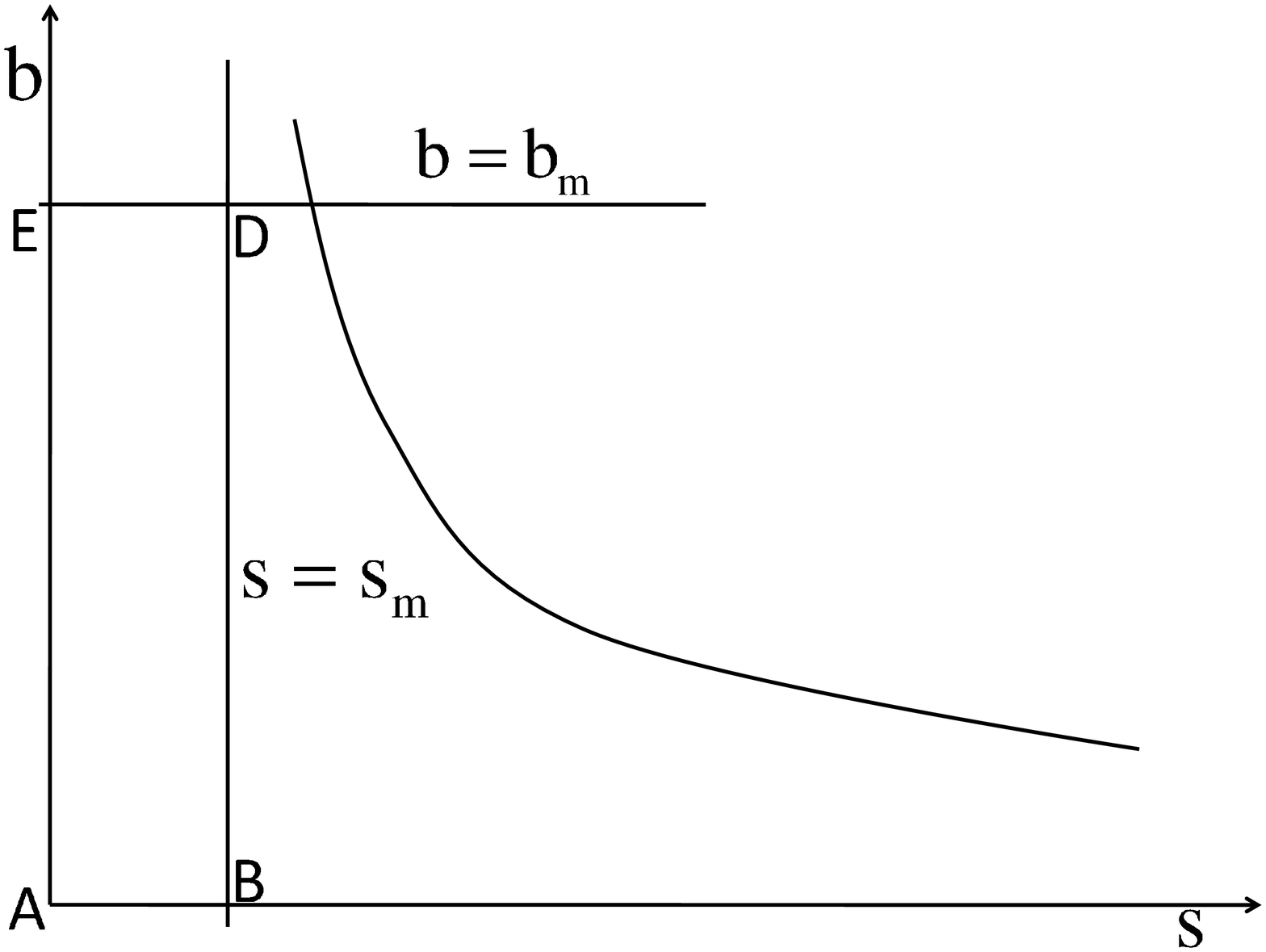}}
  \end{center}
  \caption{Sets of feasible $(b,s)$ tuples for different values of $a$.}
  \label{protocol_comparison_scenario1}
\end{figure}

For any given value of $a$ the following observations will help us solve the problem:
\begin{itemize}
\item  The optimal $(b,s)$ pair always lies on the boundary BCDE or BDE, because the objective function is monotonically decreasing with both $b$ and $s$.
\item By the same argument the optimal point on the line segment BC and ED are C and D, respectively.
\end{itemize}

Now, we argue that the optimal solution does not lie in $\mathcal{O}_1\cap \mathcal{R}^c$. Note that when $a\in \mathcal{R}^c$ the optimal solution for the $(b,s)$ tuple is point D where $(b,s)=(b_m,s_m)$. However, when $b=b_m$ we have either $b=p$ or $b=(n-a)$. In both of these cases the objective function is monotonically decreasing with $a$ (e.g., see Claim~\ref{cl_a_plus_b_equal_1}). The same is true for $s$. Therefore, it is clear from the definition of $\mathcal{O}_1$ that when $(b,s)=(b_m,s_m)$, $a$ should be (also can be) increased until the constraint $a+\frac{b_m s_m}{b_m+s_m}\leq r$ becomes active. In that case however, we have $a\in \mathcal{R}$ because
\begin{equation*}
    \frac{b_m s_m}{b_m+s_m}= (r-a) \quad \Longrightarrow \quad \frac{b_m(r-a)}{(b_m-r+a)} = s_m.
\end{equation*}
So, the objective function does not attain its minimum value when $a\in \mathcal{R}^c$ and we need to optimize the objective function only over the set $\mathcal{O}_1\cap \mathcal{R}$. In the definition of $\mathcal{R}$ the condition in terms of $s_m$ and $b_m$ can be converted to constraints on $a$ as
\begin{IEEEeqnarray}{rl}
\label{eq_region_R}
\mathcal{R}= \left[\max\{r-\frac{pq}{(p+q)},r-\sqrt{(m-r)(n-r)},a_n^*,a_m^*\}, r\right],
\end{IEEEeqnarray}
where
\begin{IEEEeqnarray}{rl}
a_n^*=&\left(\frac{n+r}{2}\right)-\sqrt{\left(\frac{n-r}{2}\right)^2+q(n-r)};\nonumber\\
a_m^*=&\left(\frac{m+r}{2}\right)-\sqrt{\left(\frac{m-r}{2}\right)^2+p(m-r)}.
\end{IEEEeqnarray}
Also by the previous argument the optimal $(b,s)$ tuple lies on the arc CD and satisfies $\frac{bs}{(b+s)}=(r-a)$. Further on the arc CD $b$ can take any value between point, E where $b=b_m$ and F, where $b=\frac{s_m(r-a)}{(s_m-r+a)}$ and thus lies in the range $\mathcal{B}=\left[\frac{s_m(r-a)}{(s_m-r+a)}, b_m\right]$. Using these facts, we see that the optimal solution is given by
\begin{equation}
\label{eq_pf_simple_opt_temp1}
d^*(r)=\min_{\left\{a\in \mathcal{R},~b\in \mathcal{B} \right\}} F\left(\phi_{{\alpha}}(a),\phi_{{\beta}}(b),\phi_{{\delta}}\left(\frac{b(r-a)}{(b-r+a)}\right)\right).
\end{equation}

\section{Proof of Theorem~\ref{thm_closedform_sol_symmetric_case}} \label{app-E}
To prove the theorem we evaluate the minimum value of the objective function in the optimization problem of Theorem~\ref{thm_simplified_optimization} over different carefully chosen subsets of the feasible set. The choice of these subsets also helps us to obtain a closed form expression for the optimal solution in each subset. The union of these sets might not be equal to the feasible set. The minimum of these different optimal solutions represent the minimum value of the objective function over a subset of the feasible set and hence yields only an an upper bound to the actual minimum. The proof is divided into different cases and each case considers a particular subset of the feasible set.

{\it Case 1} ($\mathcal{O}_1\cap \{a=r\}$): We know the optimal $(a,b,s)$ tuple satisfies $a+\frac{bs}{(b+s)}=r$. So, $a=r$ implies either $b=0$ or $s=0$. Since we are considering the symmetric case ($m=n$), without loss of generality we assume $s=0$. From the definition of $\phi_{\delta}$ we get $\hat{\delta}_l=1$ for $1\leq l\leq q$. Since the objective function is monotonically decreasing with $b$ for a given $a$ and $s$ to minimize the objective function the maximum possible value of $b$ should be chosen, i.e., $b=b_m=\min \{(n-r),p\}$. We need to consider two different cases: 1) $(n-r)\leq p$ and 2) $(n-r)\geq p$. In the first case, $b=(n-r)$ and $(a+b)=n$ which along with Claim~\ref{cl_a_plus_b_equal_1} implies that
\begin{IEEEeqnarray}{rl}
d_{11}(r) \triangleq \min F =&\sum_{i=1}^{n}(2n+k-2i+1)\hat{\alpha}_i+ \sum_{j=1}^{p}(k+n-2j+1)\hat{\delta}_j -kn +\sum_{\substack{i,j=1\\j+i\leq n}}^{n,q}(1-\hat{\alpha}_i-\hat{\delta}_j) ,\nonumber\\
\label{eq_closed_bound11}
\stackrel{(a)}{=}&\sum_{i=1}^{n}(2n+k-2i+1)\hat{\alpha}_i=d_{n,(n+k)}(r), ~(n-p)\leq r\leq n,
\end{IEEEeqnarray}
where step $(a)$ is obtained by putting $\hat{\delta}_l=1,~\forall l$ and the last step follows from the definition of $\phi_{\alpha}(a)$. Next we consider the case when $(n-a) \geq p$ and $b=p$, which in turn imply $\hat{\beta}_j=0, ~1\leq j\leq p$. Putting this in the objective function,  we get
\begin{IEEEeqnarray}{rl}
d_{12}(r)\triangleq \min F=&\sum_{i=1}^{n}(2n+2k-2i+1)\hat{\alpha}_i+ -kn +\sum_{\substack{i,j=1\\j+i\leq n}}^{n,p}(1-\hat{\alpha}_i) ,\nonumber\\
\label{eq_closed_bound12}
=&\sum_{i=1}^{n}(2n+k-2i+1)\alpha_i=d_{n,(n+k)}(r), ~0\leq r\leq (n-p).
\end{IEEEeqnarray}
Combining equations \eqref{eq_closed_bound11} and \eqref{eq_closed_bound12} we get the minimum value of the objective function over the chosen subset
\begin{equation}
\label{eq_dmt_symmetric_bound_1}
d_{U_1}\triangleq \min \{d_{11}(r),d_{12}(r)\}=d_{n,(n+k)}(r), ~0\leq r\leq n.
\end{equation}

{\it Case 2} ($\mathcal{O}_1\cap \{b=s=(n-a)\}$): Putting $b=s=(n-a)$ in the constraint $a+\frac{bs}{(b+s)}=r$ which is always active we get $a=(2r-n)$. Now, $a\in \mathcal{R}$ if and only if
\begin{equation}
\max \left\{r-\frac{p}{2}, a^*_n, 2r-n\right\}\leq (2r-n) \quad \Longrightarrow \quad (n- \frac{p}{2})\leq r.
\end{equation}
Since $(a+b)=n=(a+s)$, we know from Claim~\ref{cl_a_plus_b_equal_1} that the objective function gets simplified to
\begin{IEEEeqnarray}{rl}
d_{U_2}(r)\triangleq \min F = & ~\sum_{i=1}^{n}(2n-2i+1)\alpha_i \\ =& ~ d_{n,n}(2r-n) \\ =& ~ d_{2n,2n}(2r), ~ (n-\frac{p}{2})\leq r\leq n ~[\because a=(2r-n)].
\end{IEEEeqnarray}

{\it Case 3} ($\mathcal{O}_1\cap \{a=0, b=p\}$): We know from Theorem~\ref{thm_simplified_optimization} that $a\in \mathcal{R}$ if and only if
\begin{equation}
\max \left\{r-\frac{p}{2}, a^*_n, 2r-n\right\}\leq 0 \quad \Longrightarrow \quad (n- \frac{p}{2})\leq r.
\end{equation}
Further, from the definition of $\phi_{\alpha}$ and $\phi_{\beta}$ we get $\hat{\beta}_j=0, ~\forall j\leq p$ and $\hat{\alpha}_i=1, ~\forall i\leq n$. Putting this in the objective function we have
\begin{IEEEeqnarray}{rl}
d_{U_3}(r)\triangleq \min F &=n(n+2k) + \sum_{l=1}^{p}(k+n-2l+1)\hat{\delta}_j -2kn, \nonumber \\
&=n^2+\sum_{l=1}^{p}(n+k-2l+1)\left(1-\left(\frac{pr}{(p-r)}-l+1\right)^+\right)^+,~ 0\leq r\leq \frac{p}{2},.
\end{IEEEeqnarray}
where the last step follows from the fact that the optimizing $(a,b,s)$ tuple satisfies\footnote{Recall the optimal solution lies on the arc CD in Fig. \ref{a_range-a}.} $a+\frac{bs}{(b+s)}=r$.

{\it Case 4} ($\mathcal{O}_1\cap \{b=s=N, 1\leq N\leq p\}$): From the constraint $a+\frac{bs}{(b+s)}=r$ we get $a=r-\frac{N}{2}$. Now, $a\in \mathcal{R}$ if and only if
\begin{equation}
\max \left\{r-\frac{p}{2}, (2r-n), a^*_n\right\}\leq r-\frac{N}{2},~\Longrightarrow~\frac{N}{2}\leq r\leq \min \left\{n-\frac{N}{2}, n-\frac{N^2}{(2p-N)} \right\}.
\end{equation}
Since $b=s=N$, from the definition of $\phi_{i}$'s, we have $\hat{\delta}_j, \hat{\beta}_j=1,~\forall j\geq (N+1)$ and $\hat{\delta}_j, \hat{\beta}_j=0,~\forall j\leq N$. Putting this in the objective function we have
\begin{IEEEeqnarray}{ll}
d_{U_{(3+N)}}(r)\triangleq \min F&=\sum_{i=1}^{n}(2n+2k-2i+1)\hat{\alpha}_i + \sum_{j=(N+1)}^{p}2(k+n-2j+1)+ -2kn +2\sum_{j=1}^{N}\sum_{i=1}^{(n-j)}(1-\hat{\alpha}_i),\nonumber \\
&=\sum_{i=1}^{(n-N)}(2n+2k-2N-2i+1)\hat{\alpha}_i+N^2,\nonumber \\
&=N^2+d_{(n-N),(n+2k-N)}\left(r-\frac{N}{2}\right),~ \frac{N}{2}\leq r\leq \min \left\{n-\frac{N}{2}, n-\frac{N^2}{(2p-N)}\right\}.
\end{IEEEeqnarray}

{\it Case 5} ($\mathcal{O}_1\cap \{b=(n-a), s=N, 1\leq N\leq p\}$): We further assume $k\geq n$, i.e., $p=q=n$. Using the fact that the optimal solution always lies on the arc CD in Fig.~\ref{a_range-a} we have
\begin{equation}
a+\frac{N(n-a)}{(N+n-a)}=r.
\end{equation}
Solving the above equation for $a$ we get
\begin{equation}
\label{eq_exp_for_a_N}
a_N=\frac{(n+r)}{2}-\sqrt{\left(\frac{(n-r)}{2}\right)^2+N(n-r)}, \quad 1\leq N\leq p.
\end{equation}
Now, $a\in \mathcal{R}$ if and only if
\begin{IEEEeqnarray}{l}
\max \left\{a^*_n,(2r-n),r-\frac{n}{2}\right\}\leq \frac{(n+r)}{2}-\sqrt{\left(\frac{(n-r)}{2}\right)^2+N(n-r)} ,
\end{IEEEeqnarray}
which implies
\begin{IEEEeqnarray}{l}
\max \left\{\frac{Nn}{(N+n)}, n-p \right\}\leq r\leq n-\frac{N}{2}.
\end{IEEEeqnarray}
From Claim~\ref{cl_a_plus_b_equal_1} we get
\begin{IEEEeqnarray}{rl}
d_{U_{(3+p+N)}}(r)\triangleq \min F&=\sum_{i=1}^{n}(2n+k-2i+1)\hat{\alpha}_i + \sum_{j=1}^{n}(k+n-2j+1)\hat{\delta}_j-kn
+\sum_{i=1}^{n}\sum_{l=1}^{(n-i)}(1-\hat{\alpha}_i-\hat{\delta}_l). \nonumber
\end{IEEEeqnarray}
Since $s=N$, from the definition of $\phi_{\delta}$, we have $\hat{\delta}_j=1,~\forall j\geq (N+1)$ and $\hat{\delta}_j=0,~\forall j\leq N$. Putting this in the above equation we get
\begin{IEEEeqnarray*}{ll}
d_{U_{(3+p+N)}}(r)&=\sum_{i=1}^{n}(2n+k-2i+1)\hat{\alpha}_i -N(n+k-N) +\sum_{i=1}^{n}\sum_{l=1}^{(n-i)\land N}(1-\hat{\alpha}_i),\\
&\stackrel{(a)}{=}\sum_{i=1}^{n-N}(2n+k-N-2i+1)\hat{\alpha}_i -\frac{N(2k-N+1)}{2} +\sum_{i=(n-N+1)}^{n}(n+k-i+1),\\
&=\sum_{i=1}^{(n-N)}(2n+k-N-2i+1)\left(1-\left(a_N-i+1\right)^+\right)^++N^2,~\frac{Nn}{(N+n)}\leq r\leq n-\frac{N}{2},
\end{IEEEeqnarray*}
where step $(a)$ follows from the fact that $\sum_{i=1}^{n}(n-i)\land N=\frac{N(2n-N-1)}{2}$.

\section{Proof of Theorem \ref{lem:DMT_1k1_channel}}
\label{pf:lem:DMT_1k1_channel}
The optimization problem of Theorem~\ref{thm_simplified_optimization} is solved analytically for $m=n=1$. Denoting the optimal solution for this case by $d_{(1,k,1)}(r)$ and specializing Theorem~\ref{thm_simplified_optimization}, we get
\begin{IEEEeqnarray}{l}
\label{eq:1k1fn}
d_{(1,k,1)}(r)=\min_{\{a\in\mathcal{R},b\in \mathcal{B}\}}(2k+1)(1-a)+k(1-b)+k\left(1-s\right)-2k,
\end{IEEEeqnarray}
where $s=\frac{b(r-a)}{(b-r+a)}$ and $\mathcal{R}$ and $\mathcal{B}$ can be computed from equation \eqref{eq_region_R} and Theorem~\ref{thm_simplified_optimization} by setting $m=n=1$. Since the objective function is symmetric with respect to $b$ and $s$, without loss of generality, we can assume that $b\geq s=\frac{b(r-a)}{(b-r+a)}$, which in turn implies $b\geq 2(r-a)$. Differentiating with respect to $b$, we see that the function in \eqref{eq:1k1fn} is a monotonically decreasing function of $b$ for any given $a$ when $b\geq 2(r-a)$. It is thus minimized when $b=\max \mathcal{B}=\max \left[\frac{(1-a)(r-a)}{(1-r)}~(1-a)\right]=(1-a)$. Putting this in the above equation we get
\begin{IEEEeqnarray}{rl}
d_{(1,k,1)}(r)=&\min_{\{a\}}~(2k+1)(1-a)+ka+k\left(1-\frac{(1-a)(r-a)}{(1-r)}\right)-2k\nonumber\\
\label{eq_pf_1k1_DMT_temp}
=&\min_{\{a\}}~(k+1)(1-r)+\underbrace{(r-a)\left(1-\frac{k(r-a)}{(1-r)}\right)}_{T(r-a)},
\end{IEEEeqnarray}
where $0\leq (r-a)\leq (r\land (1-r))$. Note that $T(r-a)$ in the above equation is a concave function of $(r-a)$ of the form depicted in Fig.~\ref{fig_concavity} (in this figure, $k=3$ and $r=.35$) which intersects the x-axis at $(r-a)=R_1=\frac{(1-r)}{k}$. Thus the objective function is minimized for the following values of $(r-a)$

\begin{IEEEeqnarray*}{l}
(r-a)=\left\{\begin{array}{ll}0,&\textrm{if } (r-a)\leq \frac{(1-r)}{k};\\
(r\land (1-r)),&\textrm{if } (r-a)> \frac{(1-r)}{k}\end{array}\right.~\Longrightarrow~(r-a)=\left\{\begin{array}{ll}0,&\textrm{if }~0\leq r\leq \frac{1}{(1+k)};\\
r,&\textrm{if } ~\frac{1}{(1+k)}\leq r\leq  \frac{1}{2};\\
(1-r),&\textrm{if } ~\frac{1}{2}\leq r\leq 1.\end{array}\right.
\end{IEEEeqnarray*}

\begin{figure}[htp]
  \begin{center}
\includegraphics[scale=.5]{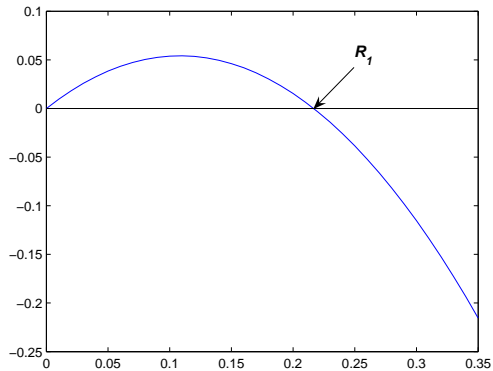}
  \end{center}
\caption{Plot of $T(r-a)$ vs. $(r-a)$.}
  \label{fig_concavity}
\end{figure}

Putting these values for optimal $(r-a)$ in equation \eqref{eq_pf_1k1_DMT_temp}, we get the DMT of the $(1,k,1)$ HD-RC as in \eqref{eq:dmt1k1rc}, thus proving the theorem.



\section{Proof of Theorem \ref{lem:DMT_n1n_channel}}
\label{pf:lem:DMT_n1n_channel}
To prove this result we solve the optimization problem in Theorem~\ref{thm_optimization_problem} analytically for the case $m=n$ and $k=1$. Let the optimal solution be denoted by $d_{(n,1,n)}^*(r)$. In what follows, the computation of $d_{(n,1,n)}^*(r)$ is carried out in two steps. First, we compute an upper bound, $d_{(n,1,n)}^u(r)\geq d_{(n,1,n)}^*(r)$, by computing the minimum of the objective function over a carefully chosen subset of the feasible set of the optimization problem in Theorem~\ref{thm_optimization_problem}. Then we compute the DMT of the static $(n,1,n)$ RC, denoted as $d_{(n,1,n)}^{stat}(r)$, which serves as a lower bound (since dynamic protocols include static protocols as a special case).

\subsection{The Upper bound}
Putting $m=n$ and $k=1$ in the optimization problem of Theorem~\ref{thm_optimization_problem} we get
\begin{IEEEeqnarray}{rl}
d_{(n,1,n)}^*(r)=\min_{\{(\bar{\alpha},\beta_1,\delta_1)\in \mathcal{S}\}} \sum_{i=1}^n(2n-2i+3)\alpha_i & +n\beta_1+n\delta_1-2n\nonumber\\
&+\sum_{i=1}^{n-1}\left\{(1-\beta_1-\alpha_i)^++(1-\delta_1-\alpha_i)^+\right\}.
\end{IEEEeqnarray}
where
\begin{IEEEeqnarray}{rl}
\mathcal{S}=\Big\{(\bar{\alpha},\beta_1,\delta_1): \sum_{i=1}^n(1-\alpha_i)+\frac{(1-\beta_1)(1-\delta_1)}{(1-\beta_1)+(1-\delta_1)}&\leq r; ~
0\leq \alpha_1\leq \alpha_2\leq,\cdots,\leq \alpha_n \leq 1;\nonumber \\
0\leq \beta_1 \leq 1,~0\leq \delta_1 &\leq 1; ~
(\beta_1+\alpha_n)\geq 1,~(\delta_1+\alpha_n) \geq 1\Big\}.
\end{IEEEeqnarray}
Now, suppose we reduce the size of $\mathcal{S}$ by putting the additional constraints $\delta_1=1$ and $(\alpha_n+\beta_1)=1$ and denote the resulting feasible set by $\hat{\mathcal{S}}$, i.e.,
\begin{IEEEeqnarray}{rl}
\hat{\mathcal{S}}=\Big\{(\bar{\alpha},\beta_1,\delta_1): \sum_{i=1}^n(1-\alpha_i)&\leq r; ~
0\leq \alpha_1\leq \alpha_2\leq,\cdots,\leq \alpha_n  \leq 1;~ 0\leq \beta_1 \leq 1,~0\leq \delta_1 \leq 1; ~  \beta_1+\alpha_n =  1\Big\}. \nonumber
\end{IEEEeqnarray}
The objective function minimized over $\hat{\mathcal{S}}$ would clearly be an upper bound on $d_{(n,1,n)}^*(r)$. Denoting the minimum over $\hat{\mathcal{S}}$ as $d^u_{(n,1,n)}(r)$, we have
\begin{IEEEeqnarray}{rl}
d_{(n,1,n)}^u(r)=\min_{\{(\bar{\alpha},\beta_1,\delta_1)\in \hat{\mathcal{S}}\}} \sum_{i=1}^n(2n-2i+3)\alpha_i +n\beta_1-n
\label{eq:objective-function-n1n-original}
+\sum_{i=1}^{n-1}(1-\beta_1-\alpha_i)^+,
\end{IEEEeqnarray}
and
\begin{equation*}
    d_{(n,1,n)}^*(r)\leq d_{(n,1,n)}^u(r).
\end{equation*}
Further, since $(\alpha_n+\beta_1)=1$ we have $(\alpha_i+\beta_1)\leq 1$ for $1\leq i\leq (n-1)$ which in turn implies that
\begin{equation*}
    \sum_{i=1}^{n-1}(1-\beta_1-\alpha_i)^+=\sum_{i=1}^{n-1}(1-\beta_1-\alpha_i).
\end{equation*}
Substituting this in equation \eqref{eq:objective-function-n1n-original} and eliminating $\beta_1$ from the objective function we get
\begin{IEEEeqnarray}{rl}
d_{(n,1,n)}^u(r)=\min_{\{\bar{\alpha}\}} \sum_{i=1}^n(2n-2i+2)\alpha_i
\end{IEEEeqnarray}
subject to
\begin{IEEEeqnarray*}{rl}
\sum_{i=1}^n(1-\alpha_i)\leq r~\textrm{and }~
0\leq \alpha_1\leq \alpha_2\leq,\cdots,\leq \alpha_n \leq 1.
\end{IEEEeqnarray*}
This is the optimization problem that arises in solving for the $(n+1,n)$ point-to-point MIMO channel \cite{tse1}.
Hence, we have
\begin{equation}\label{eq:dmt-n1n-upper-bound}
    d_{(n,1,n)}^*(r)\leq d_{(n,1,n)}^u(r)=d_{(n+1),n}^{ptp}(r),~0\leq r\leq n.
\end{equation}

\subsection{Lower bound}
The DMT of the static ($n,1,n$) HD-RC, denoted as $d_{(n,1,n)}^{stat}(r)$, and shown in Theorem \ref{lem:dmt-static-n1n} proved in Appendix \ref{pf:lem:dmt-static-n1n} to be equal to that of the $(n+1,n)$ point-to-point MIMO channel, is a lower bound to $ d_{(n,1,n)}^*(r) $.
\begin{equation}
\label{eq:dmt-n1n-lower-bound}
   d_{(n,1,n)}^*(r) \geq d_{(n,1,n)}^{stat}(r)=d_{(n+1),n}^{ptp}(r),~0\leq r\leq n.
\end{equation}
Combining this result with the upper bound in \eqref{eq:dmt-n1n-upper-bound}, we have Theorem \ref{lem:DMT_n1n_channel}.

\section{Proof of Theorem~\ref{lem:dmt-static-n1n} (Contd.)}
\label{pf:lem:dmt-static-n1n}

It is proved by mathematical induction that
\begin{equation}
\label{eq:dmt-n1n-static-RC}
    d_{(n,1,n)}^{stat}(r)=d_{(n+1),n}^{ptp}(r),~0\leq r\leq n.
\end{equation}
For $n=1$, the result is given in \cite{PAT}. Assuming that \eqref{eq:dmt-n1n-static-RC} is true for $n=(N-1)$,
we prove that it is also true for $n=N$. Now, from the objective function in equation \eqref{eq:objective-function-n1n-lower-bound-0} and the constraint \eqref{eq:constraint-n1n-lower-bound-1} it is clear that for $0\leq r\leq 1$, the objective function decays at the fastest rate if $\alpha_1$ is decreased with increasing $r$.\footnote{For $0\leq r\leq \frac{1}{2}$, the objective function decays at the same rate if $\beta_1$ is decreased, but then for $\frac{1}{2}\leq r\leq 1$, the objective function decreases at a strictly smaller rate than if $\alpha_1$ was decreased from the beginning. } Therefore, for $0\leq r\leq 1$, the optimal solution is $\alpha_1=(1-r)^+$, $\alpha_i=1$ for $2\leq i\leq N$ and $\beta_1=1$. Putting this in equation \eqref{eq:objective-function-n1n-lower-bound-0} we get
\begin{IEEEeqnarray}{rl}
d_{(N,1,N)}^{stat}(r)=&2N(1-r)+d_{(N-1),N}^{ptp}(0),~0\leq r\leq 1;\nonumber\\
\label{eq:objective-function-n1n-lower-bound-a}
=&d_{(N+1),N}^{ptp}(r),~0\leq r\leq 1.
\end{IEEEeqnarray}

On the other hand, since $\alpha_1=(1-r)^+$ for $r\geq 1$, substituting $\alpha_1=0$ in equation \eqref{eq:objective-function-n1n-lower-bound-0} we see that the optimization problem can be written as
\begin{IEEEeqnarray*}{rl}
\label{eq:dmtstatoptblah}
d_{(N,1,N)}^{stat}(r)=&\min_{\{(\bar{\alpha},\beta_1)\}} \sum_{i=2}^N(2N-2i+2)\alpha_i  +(N-1)\beta_1-(N-1)+\sum_{i=2}^{N-1}(1-\beta_1-\alpha_i)^+,\nonumber\\
=&\min_{\{(\bar{\hat{\alpha}},\beta_1)\}} \sum_{i=1}^{(N-1)}(2(N-1)-2i+2)\hat{\alpha}_i  +(N-1)\beta_1-(N-1)+\sum_{i=1}^{N-1}(1-\beta_1-\hat{\alpha}_i)^+,
\end{IEEEeqnarray*}
subject to the following constraints
\begin{IEEEeqnarray*}{rl}
\sum_{i=1}^{(N-1)}(1-\hat{\alpha}_i)+\frac{1}{2}(1-\beta_1)\leq (r-1);  \;
0\leq \hat{\alpha}_1\leq \hat{\alpha}_2\leq,\cdots \leq \hat{\alpha}_{(N-1)} \leq 1; \;
0\leq \beta_1\leq 1 ~\textrm{and}~(\beta_1+\hat{\alpha}_{(N-1)}) \geq 1.
\end{IEEEeqnarray*}
Evidently, the solution of \eqref{eq:dmtstatoptblah} at a given $r$ is the DMT of the static $(N-1,1,N-1)$ HD-RC evaluated at $r-1$,
which by the induction assumption is
\begin{IEEEeqnarray*}{rl}
d_{(N,1,N)}^{stat}(r)=& ~ d_{(N-1,1,N-1)}^{stat}(r-1), ~ 0\leq (r-1)\leq (N-1) \nonumber \\
=& ~ d_{(N-1),N}^{ptp}(r-1), ~ 0\leq (r-1)\leq (N-1) \nonumber \\
=& ~ d_{(N+1),N}^{ptp}(r), ~ 1\leq r\leq N.
\end{IEEEeqnarray*}
Combining this with equation \eqref{eq:objective-function-n1n-lower-bound-a} we get $
d_{(N,1,N)}^{stat}(r)= ~ d_{(N+1),N}^{ptp}(r), ~ 0\leq r\leq N $
Hence, by induction we have for all $n\in \mathbb{N}$,
\begin{IEEEeqnarray*}{rl}
d_{(n,1,n)}^{stat}(r)=& ~ d_{(n+1),n}^{ptp}(r), ~ 0\leq r\leq n .
\end{IEEEeqnarray*}

\bibliographystyle{IEEETran}
\bibliography{mybibliography}

\begin{thebibliography}{10}
\providecommand{\url}[1]{#1}
\csname url@samestyle\endcsname
\providecommand{\newblock}{\relax}
\providecommand{\bibinfo}[2]{#2}
\providecommand{\BIBentrySTDinterwordspacing}{\spaceskip=0pt\relax}
\providecommand{\BIBentryALTinterwordstretchfactor}{4}
\providecommand{\BIBentryALTinterwordspacing}{\spaceskip=\fontdimen2\font plus
\BIBentryALTinterwordstretchfactor\fontdimen3\font minus
  \fontdimen4\font\relax}
\providecommand{\BIBforeignlanguage}[2]{{%
\expandafter\ifx\csname l@#1\endcsname\relax
\typeout{** WARNING: IEEEtran.bst: No hyphenation pattern has been}%
\typeout{** loaded for the language `#1'. Using the pattern for}%
\typeout{** the default language instead.}%
\else
\language=\csname l@#1\endcsname
\fi
#2}}
\providecommand{\BIBdecl}{\relax}
\BIBdecl

\bibitem{SKMV_ISIT2010_Sym_relay}
S.~Karmakar and M.~K. Varanasi, ``Diversity-multiplexing tradeoff of the
  {S}ymmetric {M}{I}{M}{O} {H}alf-{D}uplex relay channel,'' in \emph{Proc. IEEE
  Intl. Symp. Inform. Th.}, Austin, Texas, Jun, 2010.

\bibitem{SEB1}
A.~Sendonaris, E.~Erkip, and B.~Aazhang, ``User cooperation diversity-part i:
  System description,'' \emph{IEEE Trans. Commun.}, vol.~51, pp. 1927--1938,
  Nov, 2003.

\bibitem{SEB2}
------, ``User cooperation diversity-part ii: Implementation aspects and
  performance analysis,'' \emph{IEEE Trans. Commun.}, vol.~51, pp. 1939--1948,
  Nov, 2003.

\bibitem{YHXM}
Y.~Yang, H.~Hu, J.~Xu, and G.~Mao, ``Relay technologies for wimax and
  lte-advanced mobile systems,'' \emph{IEEE Commun. Magazine}, vol.~47, pp.
  100--105, Oct, 2009.

\bibitem{relayTG}
``I. 802.16s relay task group,'' \texttt{http://www.ieee802.org/16/relay/}.

\bibitem{WSC}
W.~Wei, V.~Srinivasan, and K.-C. Chua, ``Using mobile relays to prolong the
  lifetime of wireless sensor networks,'' in \emph{Proc. of ACM MobiCom 2005,
  Cologne, Germany}, August 2005.

\bibitem{Wang_Zhang_Madsen}
B.~Wang, J.~Zhang, , and A.~Host-Madsen, ``On the capacity of {M}{I}{M}{O}
  relay channels,'' \emph{IEEE Trans. Inform. Th.}, vol.~51, pp. 29--43, Jan,
  2005.

\bibitem{KHP}
K.~Azarian, H.~E. Gamal, and P.~Schniter, ``On the achievable
  diversity-multiplexing tradeoff in half-duplex cooperative channels,''
  \emph{IEEE Trans. Inform. Th.}, vol.~51, pp. 4152--4172, Dec, 2005.

\bibitem{YEE}
M.~Yuksel and E.~Erkip, ``Multi-antenna cooperative wireless systems: A
  diversity multiplexing tradeoff perspective,'' \emph{IEEE Trans. Inform.
  Th.}, vol.~53, pp. 3371--3393, Oct, 2007.

\bibitem{OCM}
O.~Leveque, C.~Vignat, and M.~Yuksel, ``Diversity-multiplexing tradeoff for the
  mimo static half-duplex relay,'' Jul, 2010.

\bibitem{GK}
G.~Kramer, ``Models and theory for relay channels with receive constraints,''
  in \emph{Proc. Allerton Conf. on Comm., Control, and Comput.}, Monticello,
  IL, 2004.

\bibitem{tse1}
L.~Zheng and D.~Tse, ``Diversity and multiplexing: A fundamental tradeoff in
  multiple antenna channels,'' \emph{IEEE Trans. Inform. Th.}, vol.~49, pp.
  1073--1096, May, 2003.

\bibitem{Avestimehr_Diggavi_Tse}
A.~S. Avestimehr, S.~N. Diggavi, and D.~N.~C. Tse, ``Wireless network
  information flow: {A} deterministic approach,'' \emph{IEEE Trans. Inform.
  Th.}, vol.~57, pp. 1872 -- 1905, Apr, 2011.

\bibitem{Meulen}
E.~C.~V. der Meulen, ``Three-terminal communication channels,'' \emph{Adv. App.
  Prob.}, vol.~3, pp. 120--154, 1971.

\bibitem{Cover_Gamal}
T.~M. Cover and A.~A.~E. Gamal, ``Capacity theorems for the relay channel,''
  \emph{IEEE Trans. Inform. Th.}, vol.~25, pp. 572--584, Sept, 1979.

\bibitem{ElGamal-Aref}
A.~E. Gamal and M.~R. Aref, ``The capacity of the semideterministic relay
  channel,'' \emph{IEEE Trans. Inform. Th.}, vol.~28, pp. 536 -- 536, May,
  1982.

\bibitem{Madsen_Zhang}
A.~Host-Madsen and J.~Zhang, ``Capacity bounds and power allocation for the
  wireless relay channel,'' \emph{IEEE Trans. Inform. Th.}, vol.~51, pp.
  2020--2040, Jun, 2005.

\bibitem{Kramer_Gastpar_Gupta}
G.~Kramer, M.~Gastpar, and P.~Gupta, ``Cooperative strategies and capacity
  theorems for relay networks,'' \emph{IEEE Trans. Inform. Th.}, vol.~51, pp.
  3037--3063, Aug, 2005.

\bibitem{Noisy_NC}
S.~H. Lim, Y.~Kim, A.~E. Gamal, and S.~Y. Chung, ``{N}oisy {N}etwork
  {Coding},'' Mar, 2010, available online:
  \texttt{http://arxiv.org/abs/1002.3188}.

\bibitem{LTW0}
J.~N. Laneman, D.~N.~C. Tse, and G.~Wornell, ``Cooperative diversity in
  wireless networks: efficient protocols and outage behavior,'' \emph{IEEE
  Trans. Inform. Th.}, vol.~50, no.~12, p. 30623080, 2004.

\bibitem{LW}
J.~N. Laneman and G.~Wornell, ``Distributed spce-time-coded protocols for
  exploiting cooperative diversity in wireless networks,'' \emph{IEEE Trans.
  Inform. Th.}, vol.~49, pp. 2415--2425, Oct, 2003.

\bibitem{NabarRU:relay}
R.~U. Nabar, H.~Bolcskei, and F.~W. Kneubuhler, ``Fading relay channels:
  Performance limits and space-time signal design,'' \emph{Journ. Selec. Areas
  Commun.}, vol.~22, no.~6, pp. 1099--1109, Aug. 2004.

\bibitem{PrasadN:CO-OP:ISIT04}
N.~Prasad and M.~K. Varanasi, ``Diversity and multiplexing tradeoff bounds for
  cooperative diversity schemes,'' in \emph{Proc. IEEE Intl. Symp. Inform.
  Th.}, Chicago, IL, Jun. 2004.

\bibitem{NpV}
------, ``High performance static and dynamic cooperative communication
  protocols for the half duplex fading relay channel,'' in \emph{Proc. of
  Global TeleComm. Conf., San Francisco, CA}, Nov-Dec, 2006, pp. 1--5.

\bibitem{SKMV_DDF_journal}
S.~Karmakar and M.~K. Varanasi, ``The diversity-multiplexing tradeoff of
  dynamic decode-and-forward protocol on a mimo half-duplex relay channel,''
  Aug, 2010, available online: \texttt{http://arxiv.org/abs/1008.1455}.

\bibitem{Sanjay_Varanasi_IC_DMT}
------, ``The diversity-multiplexing tradeoff of the symmetric {M}{I}{M}{O}
  2-user interference channel,'' in \emph{Proc. IEEE Intl. Symp. Inform. Th.},
  Austin, Texas, Jun, 2010.

\bibitem{Sanjay_Varanasi_ZIC_DMT}
------, ``The diversity-multiplexing tradeoff of the {M}{I}{M}{O} {Z}
  interference channel,'' in \emph{Proc. IEEE Intl. Symp. Inform. Th.}, Austin,
  Texas, Jun, 2010.

\bibitem{PAT}
S.~Pawar, A.~Avestimehr, and D.~Tse, ``Diversity-multiplexing tradeoff of the
  half-duplex relay channel,'' in \emph{Proc. Allerton Conf. on Comm., Control,
  and Comput.}, Monticello, IL, Sept. 2008, pp. 27--33.

\bibitem{HanlyTse}
S.~Hanly and D.~N.~C. Tse, ``Multiaccess fading channels-part ii: Delay limited
  capacities,'' \emph{IEEE Trans. Inform. Th.}, vol.~44, pp. 2816--2831, 1998.

\bibitem{PrasadN:OUT-THMS:IT06}
N.~Prasad and M.~K. Varanasi, ``Outage theorems for {M}{I}{M}{O} fading
  channels,'' \emph{IEEE Trans. Inform. Th.}, vol.~52, no.~12, pp. 5284--5296,
  Dec. 2006.

\bibitem{Kim_Skoglund}
T.~T. Kim and M.~Skoglund, ``Diversity multiplexing tradeoff in mimo channels
  with partial {C}{S}{I}{T},'' \emph{IEEE Trans. Inform. Th.}, vol.~53, no.~8,
  p. 27432759, 2007.

\bibitem{CZZ}
M.~Chiani, M.~Z. Win, and A.~Zanella, ``On the capacity of spatially correlated
  mimo rayleigh-fading channels,'' \emph{IEEE Trans. Inform. Th.}, vol.~49, pp.
  2363--2371, Oct, 2003.

\bibitem{Tulino_Verdu}
A.~M. Tulino and S.~Verdu, \emph{Random Matrix Theory and Wireless
  Commun.}\hskip 1em plus 0.5em minus 0.4em\relax Now, 2004.

\bibitem{CT}
T.~M. Cover and J.~A. Thomas, \emph{Elements of Information Theory}.\hskip 1em
  plus 0.5em minus 0.4em\relax Wiley, 1991.

\bibitem{Khojastepour_Sabharwal_Aazhang}
M.~A. Khojastepour, A.~Sabharwal, and B.~Aazhang, ``Bounds on achievable rates
  for general multi-terminal networks with practical constraints,'' in
  \emph{Proc. of 2nd intnl. Workshop on Inf. processing}, 2003, pp. 146--161.

\bibitem{BV}
S.~Boyd and L.~Vandenberge, \emph{Convex Optimization}.\hskip 1em plus 0.5em
  minus 0.4em\relax Cambridge Univ. Press, 2004.

\bibitem{Edelman_Rao}
A.~Edelman and N.~R. Rao, \emph{Random Matrix Theory}.\hskip 1em plus 0.5em
  minus 0.4em\relax Cambridge University Press, 2005.

\bibitem{Bilodeau_Brenner}
M.~Bilodeau and D.~Brenner, \emph{Theory of Multivariate statistics}.\hskip 1em
  plus 0.5em minus 0.4em\relax Springer-Verlag New York, Inc, 1999.

\end{thebibliography}

\end{document}